%% file: Main.tex
\documentclass{article}
\usepackage[protrusion=true,expansion=true]{microtype}
\DisableLigatures{encoding = *, family = *}
\microtypesetup{activate=true} 

\usepackage{arxiv}

\usepackage[utf8]{inputenc} 
\usepackage[T1]{fontenc}    
\usepackage{hyperref}       
\usepackage{url}            
\usepackage{booktabs}       
\usepackage{amsfonts}       
\usepackage{nicefrac}       
\usepackage{microtype}      
\usepackage{lipsum}		
\usepackage{graphicx}
\usepackage{natbib}
\usepackage{doi}
\usepackage{indentfirst}
\usepackage{url}
\usepackage{caption}
\usepackage{amsmath}
\usepackage{amssymb}
\usepackage{enumitem}
\usepackage{amsthm}
\usepackage{algorithm}
\usepackage{algpseudocode}
\usepackage{IEEEtrantools}
\usepackage{actuarialangle}
\usepackage{tikz}
\usepackage{color}
\usepackage{xcolor}
\usepackage{bbm}
\usepackage{pgfplots}
\usepackage{pgfplotstable}
\usepackage{setspace}
\usepackage{tikz}
\usepackage{pgf-pie}
\usepackage{amssymb}
\usetikzlibrary{shapes, positioning}
\usetikzlibrary{patterns, patterns.meta}
\usetikzlibrary{pgfplots.groupplots}
\pgfplotsset{compat=1.18}
\setcitestyle{authoryear,round}
\theoremstyle{plain} 
\newtheorem{theorem}{Theorem} 

\title{Optimising pandemic response through vaccination strategies using neural networks}

\author{ \href{https://orcid.org/0009-0009-9574-7935}{\includegraphics[scale=0.06]{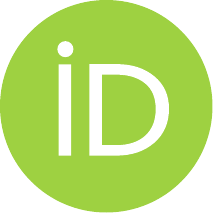}\hspace{1mm}Chang ~Zhai} \\
	Department of Economics\\
	The University of Melbourne\\
	Carlton, VIC, 3053, Australia \\
	\texttt{changzhai@unimelb.edu.au} \\
	\And
	\href{https://orcid.org/0000-0003-2885-289X}{\includegraphics[scale=0.06]{orcid.pdf}\hspace{1mm}Ping ~Chen} \\
	Department of Economics\\
	The University of Melbourne\\
	Carlton, VIC, 3053, Australia \\
	\texttt{pche@unimelb.edu.au} \\
        \And
	\href{https://orcid.org/0000-0002-9488-2993}
    {\includegraphics[scale=0.06]{orcid.pdf}\hspace{1mm}Zhuo ~Jin} \\
       Macquarie Business School\\
	Macquarie University\\
	Sydney, NSW, 2109, Australia\\
	\texttt{zhuo.jin@mq.edu.au} \\
        \And
	\href{https://orcid.org/0000-0002-6976-7168}{\includegraphics[scale=0.06]{orcid.pdf}\hspace{1mm}David ~Pitt} \\
	Department of Economics\\
	The University of Melbourne\\
	Carlton, VIC, 3053, Australia\\
	\texttt{david.pitt@unimelb.edu.au} \\
}

\renewcommand{\headeright}{ }
\renewcommand{\shorttitle}{ }

\hypersetup{
pdftitle={Optimising pandemic response through vaccination strategies using neural networks},
pdfsubject={q-bio.NC, q-bio.QM},
pdfauthor={Chang ~Zhai},
pdfkeywords={Optimal Control, COVID-19, Pandemic Risk, Epidemic Modelling, Machine Learning, Neural Networks },
}

\begin{document}
\maketitle

\begin{abstract}
    Epidemic risk assessment poses inherent challenges, with traditional approaches often failing to balance health outcomes and economic constraints. This paper presents a data-driven decision support tool that models epidemiological dynamics and optimises vaccination strategies to control disease spread whilst minimising economic losses. The proposed economic-epidemiological framework comprises three phases: modelling, optimising, and analysing. First, a stochastic compartmental model captures epidemic dynamics. Second, an optimal control problem is formulated to derive vaccination strategies that minimise pandemic-related expenditure. Given the analytical intractability of epidemiological models, neural networks are employed to calibrate parameters and solve the high-dimensional control problem. The framework is demonstrated using COVID-19 data from Victoria, Australia, empirically deriving optimal vaccination strategies that simultaneously minimise disease incidence and governmental expenditure. By employing this three-phase framework, policymakers can adjust input values to reflect evolving transmission dynamics and continuously update strategies, thereby minimising aggregate costs, aiding future pandemic preparedness.
\end{abstract}

\newpage

\section{Introduction}
\label{sec: introduction}

Recent outbreaks of the coronavirus have intensified the scientific investigation of infectious diseases. From the perspective of local governments, infectious diseases not only threaten public health but also impose substantial fiscal burdens. Given the significance of recent vaccination campaigns, vaccination plays a crucial role in preventing disease transmission. Consequently, we develop an economic epidemiological framework that assists social planners in minimizing overall pandemic-related expenses by adjusting vaccination rollout rates over time based on the progression of the disease. Figure \ref{fig: economic epidemiological approach} provides an overview of the framework, which is structured into three phases: modelling, optimizing, and analyzing. This framework incorporates actual disease data to study the evolution of the virus, develop an optimal vaccination strategy for the social planner, and generate a range of analytical insights. The proposed framework is applicable to epidemics that are not classified as catastrophic health crises, such as the Asian Flu, Swine Flu, and COVID-19. Specifically, if an epidemic has an extremely high case fatality rate (CFR) exceeding a critical threshold, likely around 10 to 20 percent or higher, then the government should prioritize saving lives over controlling aggregate spending, based on the intrinsic value of human life. In such scenarios, the societal and ethical imperative to prevent mass mortality outweighs cost-minimization policies. Therefore, we focus on the more common non-catastrophic health crises and provide policymakers with a data-driven framework for developing optimal vaccination strategies during such periods of time. In addition, in real-world settings, government vaccination campaigns are continually adjusted by updating the input data or assumptions based on the available data. Thus, the goal of this paper is to offer social planners a toolkit capable of developing optimal vaccination strategies based on the most recent data, allowing governments to update policies in response to changes in the external environment as the virus spreads.

\begin{figure}[htbp]
    \centering
    \includegraphics[width=\textwidth]
    {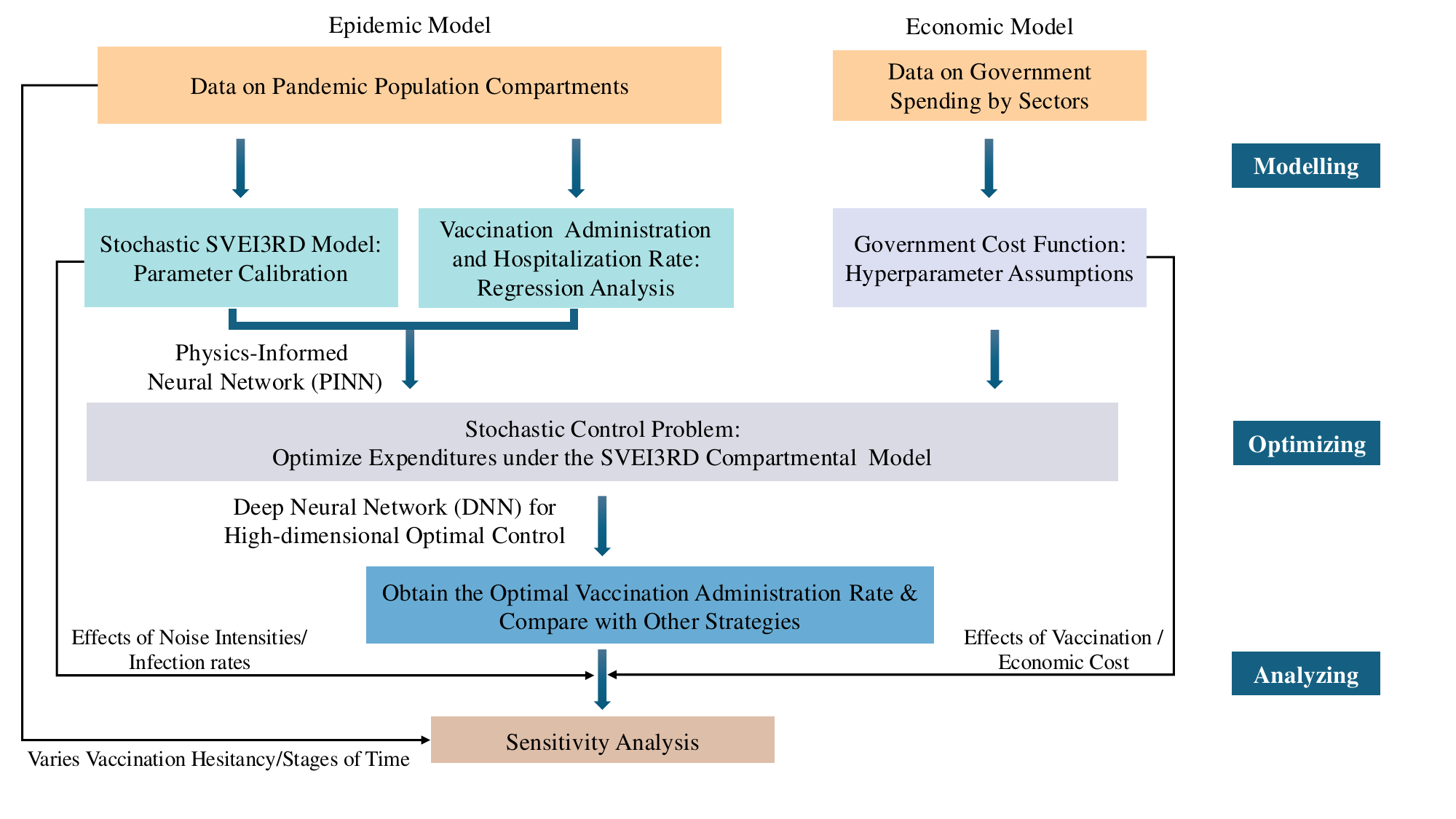}
    \caption{Overview of the economic epidemiological model framework.}
    \label{fig: economic epidemiological approach}
\end{figure}

In developing the framework, two strands of the literature are reviewed. The first strand focuses on epidemic modelling and examines compartmental models that describe the progression of infectious disease transmission. The second addresses epidemic control problems within the field of economic epidemiology. By integrating insights from these two areas, our framework becomes data-driven, incorporates real-world dynamics, and provides actionable strategic guidance for policymakers.

Developing a comprehensive understanding of infectious disease transmission is essential for shaping epidemic management strategies, and this focus underpins the first phase of our framework. Compartmental models represent a foundational tool in epidemiological research by enabling investigators to analyze the dynamics of disease spread and construct mathematical tools for understanding epidemics. In regard to the earliest investigations into an epidemic process, \citet{Kermack_1927} presents the fundamental Susceptible-Infectious-Recovered (SIR) model, which divides a population into three groups and describes their disease dynamics through a set of differential equations. Building on this foundation, researchers have expanded compartmental models by refining the underlying equations. For example, \citet{Bailey_1975} examines the proportion of individuals who die from the disease and introduces an additional Dead state into the model. Moreover, \citet{Kemper_1978} incorporates a carrier stage, accounting for individuals who can transmit the disease without displaying symptoms, and provides notable insights into infection pathways. Another line of research investigates how vaccination influences epidemic spread. A representative study by \citet{Cohen_2003} shows that vaccinating part of the population can substantially reduce overall virus transmission. Furthermore, to capture the substantial difference in medical costs across infected individuals based on illness severity, \citet{Zhai_2024} proposes a refined Susceptible-Exposed-Infected-with-3-Substates-Recovered-Dead (SVEI3RD) model that classifies infected individuals into mild, medium, and severe subgroups. This enhanced model provides a more comprehensive approach by subdividing the Infectious state, allowing governments to align distinct economic costs more closely with real-world conditions. In addition, diseases spanning longer time horizons often involve additional factors such as new births, elderly mortality, and migration, each of which shapes the compartment populations. \citet{Grenfell_2004} examines how host population dynamics affect viral transmission and evolution, noting that overall population size varies over time. A further consideration is the necessity of incorporating stochastic into the model in order to account for the fluctuations of real-world changing conditions. As shown by \citet{Beddington_1977} and \citet{Allen_2008}, stochastic models offer a more realistic representation of disease dynamics, enhancing the ability to analyze complex transmission patterns in uncertain settings. Accordingly, these models have gained prominence among researchers seeking richer insights into disease spread and more effective control measures. 

Building on these considerations, our framework begins by offering policymakers a comprehensive stochastic compartmental model with vital dynamics that capture the dynamics of disease transmission. This advanced model design includes more states than traditional models, enabling a more detailed analysis of population dynamics and the stochastic terms allow the mathematic model to incorporate environmental noise. However, the high-dimensional structure of our compartmental model presents computational challenges for parameter estimation, since most traditional numerical methods lack the capacity to calibrate a large set of parameters simultaneously. To overcome this, we employ recent developments in deep learning by applying Physics Informed Neural Networks (PINN) introduced by \citet{Raissi_2019} to calibrate the parameters of our compartmental model. This neural network approach is first applied to epidemiological models by \citet{Shaier_2021}, who introduces Disease Informed Neural Networks (DINN) and demonstrates their effectiveness in estimating parameters for compartmental models and providing accurate forecasts. While the DINN methodology is originally developed for parameter training of deterministic compartmental models, we extend this methodology to calibrate stochastic model parameters by simulating environmental noise through multiple Monte Carlo simulations and averaging the outcomes. Calibrating the parameters of our compartmental model allows us to study the disease's actual progression, forming the foundation of the epidemic modelling component of our proposed framework. For the economic part of the framework, we consider the government's aggregate pandemic-related expenditures, consisting of four parts: vaccination policy implementation costs, quarantine subsidies, healthcare system expenditures, and economic losses. The integration of these two models constitutes the first modelling layer within our framework.

The second phase of the study involves the formulation of the control problem that aligns with the policymakers' objective of reducing government expenditures by altering vaccination administration rates in response to disease progression. In addition to modelling the transmission dynamics of epidemic diseases, governments are particularly concerned with the broader impacts of pandemics. Historical evidence shows that past pandemics have significantly disrupted not only public health systems but also global economies, often triggering severe economic crises. This dual impact has heightened scholarly interest in economic epidemiology, with economists increasingly focused on devising optimal strategies for managing pandemics. \citet{McAdams_2021} surveys the literature, categorizing economic epidemiological models based on assumptions regarding immunity, transmission mechanisms, and economic impacts. Regarding the assessment of policy interventions, \citet{Giordano_2020} emphasizes the importance of non-medical measures, while \citet{Glover_2023} examines the intergenerational trade-offs and conflicts that pandemics create for policymakers. In the context of lockdown policies, studies by \citet{Acemoglu_2021}, \citet{Alvarez_2021}, \citet{Jones_2021}, \citet{Arias_2023}, and \citet{Dasaratha_2023} construct their optimal design and implementation. \citet{Garibaldi_2024} distinguishes between static and dynamic externalities, comparing decentralized and optimal solutions when agents derive utility from social interactions. Additionally, \citet{Carnehl_2023} and \citet{Chen_2023} examine optimal strategies for social distancing. \citet{Acemoglu_2024} develops frameworks for optimal surveillance testing during epidemics. Regarding vaccination strategies, \citet{Tortorice_2024} analyzes the government's optimal funding levels for vaccine research and development to minimize pandemic costs. The labour market impact of pandemics has also received attention. For instance, \citet{Jackson_2024} incorporates human capital losses from unemployment into macroeconomic models, showing that pandemics lead to significant reductions in total factor productivity.

Most studies in the literature that address the control problem in economic epidemiology utilize simplistic compartmental models in order to address computational challenges. Although these models serve to reduce complexity, leveraging advancements in compartmental modelling allows for a more detailed representation of the disease transmission process. This, in turn, can significantly enhance the ability of governments to formulate effective policies during a pandemic. To address this gap, this paper employs a more comprehensive compartmental model to better capture the dynamics of disease transmission. Furthermore, to tackle the computational difficulties, we draw on recent developments in the field of machine learning and utilize neural networks to solve the proposed high-dimensional control problem. In this context, we consider the social dilemma faced by local governments during outbreaks of infectious diseases. On one hand, prioritizing saving lives leads to significantly higher vaccination implementation costs, as the government aims to extensively promote the vaccination campaign and sustain a high rate of inoculation over an extended period. On the other hand, allocating fewer funds to vaccination efforts places greater pressure on the healthcare system, potentially leading to social dissatisfaction and increased healthcare system expenditures as the disease spreads rapidly. This trade-off is driven by the negative correlation between vaccination administration rates and hospitalization rates. As \citet{Amato_2020} points out, higher vaccination rates are typically associated with reductions in hospitalizations. Building on this concept, we formulate the optimization problem to assist policymakers in seeking an optimal balance between these trade-offs during a disease outbreak. This optimization phase involves a high-dimensional stochastic control problem with disease dynamics governed by compartmental models. The literature indicates that obtaining analytical solutions for such problems is computationally intractable, and due to the "curse of dimensionality" discussed by \citet{Bellman_2015}, most traditional methods are incapable of finding numerical solutions. Motivated by advances in scientific machine learning (SciML), particularly the use of deep learning to solve differential equations, we investigate the deep neural network developed by \citet{Han_2016} to address the problem. This method has proven effective in solving high-dimensional stochastic differential equations (SDEs). In this case, the deep neural network architecture approximates the control directly at each time point using feedforward subnetworks and connects outputs across time to form the final loss function. Utilizing this approach, we determine the optimal vaccination strategy under the proposed economic epidemiological model over time, enabling policymakers to make informed decisions.

Finally, in the last layer of the proposed framework, we perform sensitivity analysis to examine how various factors affect our model framework. By conducting this analysis, we identify critical parameters that influence the system's outputs, thereby enabling us to guide policymakers' decision-making processes more robustly.

To demonstrate the application of our framework, we present a case study utilizing a real-time dataset\footnote{COVID-19 Data for Australia. Available online at the address https://github.com/M3IT/COVID-19 Data.} containing the official historical COVID-19 records for each state of Australia. In relation to our analysis, records specific to Victoria were extracted for the purpose of determining parameters during the modelling process. Additionally, in order to support the detailed modelling of specific dose injections, we obtained a supplementary set of vaccination records from the Department of Health and Aged Care\footnote{COVID-19 Vaccination Data. Available online at: https://www.health.gov.au/resources/publications.} to work out the vaccination rates according to each dose. The numerical analysis shows that the government should prioritize vaccinations in the early stages, and then gradually decrease the vaccination rate once the disease is kept under control. This finding aligns with the optimal vaccination schedules computed for the theoretical epidemic in the study of \citet{Hethcote_1973}, indicating that early vaccination of susceptibles is a key strategy for preventing or controlling an epidemic and minimizing the aggregate cost of the vaccination program. In this case study, we find that although the government's actual vaccination program reduced the number of infectious individuals relative to a constant vaccination strategy, our proposed framework could further reduce these numbers, thereby decreasing the associated economic costs even more. The results indicate that the recommended strategy not only eases the economic burden of the government, but also lessens the strain on the healthcare system by reducing the number of infected individuals at the same time.

The key factors within the framework that drive the final outcomes are examined through a sensitivity analysis to assess policymakers' responses to pandemics under uncertainty. The results show that in scenarios with high noise intensity, the government must allocate additional vaccinations to manage larger environmental fluctuations. For diseases with higher infection rates, the government is required to sustain a high vaccination administration rate over an extended period, leading to higher expenditures. We then examine how economic considerations influence vaccination strategies. When vaccination is more expensive, the framework recommends reducing vaccination administration rates more quickly. Conversely, when labour losses impose higher financial costs on the government, policymakers appear more inclined to maintain high injection rates for a longer period. Finally, we evaluate the performance of our framework under differing levels of vaccination hesitancy and at various pandemic stages. This analysis offers policymakers valuable insights into the most effective vaccination strategies across a range of scenarios.

The organization of the paper is as follows. Section 2 presents the compartmental framework and explains the construction of the economic cost function. Section 3 details the methodology of PINN for parameter calibration and describes the deep neural network approach used to address the high-dimensional stochastic control problem. Section 4 provides a numerical case study of Victoria, Australia, comparing the optimal vaccination strategies with alternative strategies based on cost and population dynamics. Section 5 evaluates the effects of key factors within the framework and examines vaccination strategies under various scenarios through sensitivity analysis. Finally, Section 6 concludes with a summary of the findings.

\section{Model formulation}
\label{sec: problem formulation}

In this section, we examine the first phase of our framework, which focuses on modelling. We begin by addressing the mathematical models used to represent epidemic dynamics. Following this, we develop the expression to quantify the government's aggregate expenditure during a pandemic over time.

\subsection{The epidemic model}
\subsubsection{The SVEI3RD model}
\label{subsec in setup: sto model without control}

For epidemic modelling, we employ the comprehensive SVEI3RD compartmental model to analyze disease transmission dynamics. This model builds on the foundational SIR framework, which divides the population into three compartments: $\text{S}$ for susceptibles, $\text{I}$ for infectious individuals, and $\text{R}$ for those who have recovered during the pandemic. In this framework, we focus on the proportion of the total population in each compartment over time, denoted by $S_t$, $I_t$, and $R_t$. The key parameter $\beta$ governs the infectious rate, while transitions from the Infectious state to the Recovered state depend on $\delta$, which reflects the probability of recovery over time.

To extend this fundamental model, we incorporate an additional Dead state, which accounts for individuals who succumb to the disease. The mortality rate among the infected population is captured by the parameter $\mu$. In addition, the model is further enriched with the inclusion of Vaccinated and Exposed states. The Exposed state represents individuals who are infected but remain asymptomatic, effectively acting as carriers of the disease. The parameter $\gamma$ specifies the incubation rate at which exposed individuals transition to the Infectious state. Vaccination dynamics are characterized by two additional parameters: $\sigma$, which measures the inefficiency of vaccination, and $\alpha$, which represents the rate of vaccine administration.

\begin{figure}[htbp]
    \centering
    \scalebox{0.85}{
\begin{tikzpicture}[auto, scale=1, every node/.style={scale=1}]
 
\node[draw,
    rounded rectangle,
    minimum width=3cm,
    minimum height=1cm
    ] (block_s) {Susceptible (S)};

\node[draw,
    rounded rectangle,
    below = of block_s,
    minimum width=3cm,
    minimum height=1cm] (block_v) {Vaccinated (V)};
 
\node[draw,
    rounded rectangle,
    right=of block_s,
    minimum width=2.5cm,
    minimum height=1cm,
    inner sep=0] (block_e) {Exposed (E)};

\node[draw,
    rounded rectangle,
    right=of block_e,
    minimum width=4.5cm,
    minimum height=1cm,
    inner sep=0] (block_i2) {Patients in the Hospital ($
    \text{I}_2$)};
    
\node[draw,
    rounded rectangle,
    above=of block_i2,
    minimum width=4.5cm,
    minimum height=1cm,
    inner sep=0] (block_i1) {Mild Infectious ($\text{I}_1$)};

\node[draw,
    rounded rectangle,
    below=of block_i2,
    minimum width=4.5cm,
    minimum height=1cm,
    inner sep=0] (block_i3) {Patients in the ICU ($\text{I}_3$)};

\node[draw,
    rounded rectangle,
    right=of block_i2,
    minimum width=2.5cm,
    minimum height=1cm,
    inner sep=0] (block_r) {Recovered (R)};
    
\node[draw,
    rounded rectangle,
    below=3cm of block_r,
    minimum width=2.5cm,
    minimum height=1cm,
    inner sep=0] (block_d) {Dead (D)};
 
\draw[-latex] (block_s) -- (block_e)
    node[pos=0.4,fill=white]{$\beta$};

\draw[-latex] (block_s) -- (block_v)
    node[pos=0.4,fill=white]{$\alpha$};



\draw[-latex] (block_i1) -- (block_i2)
    node[pos=0.4,fill=white]{$p_1$};

\draw[-latex] (block_i2) -- (block_i3)
    node[pos=0.4,fill=white]{$p_2$};
    

\draw[-latex] (block_i2) -- (block_r)
    node[pos=0.4,fill=white]{$\delta_2$};
    
    

\node[circle, inner sep=0.1pt, fill=black, above=1.3cm of block_s] (inflow_to_s) {};
\node[circle, inner sep=0.1pt, fill=black, below=1.5cm of block_e] (v_to_e) {};
\node[circle, inner sep=0.1pt, fill=black, above=1.5cm of block_e] (e_to_i1) {};
\node[circle, inner sep=0.1pt, fill=black, above=1.5cm of block_e] (e_to_i3) {};
\node[circle, inner sep=0.1pt, fill=black, above=1.5cm of block_r] (i1_to_r) {};
\node[circle, inner sep=0.1pt, fill=black, below=1.5cm of block_r] (i3_to_r) {};
\node[circle, inner sep=0.1pt, fill=black, below=1.5cm of block_i3] (i3_to_d) {};

\node[circle, inner sep=0.1pt, fill=black, below right=0.7cm and 1cm of block_s] (point_s) {};
\node[circle, inner sep=0.1pt, fill=black, below right=0.7cm and 1cm of block_v] (point_v) {};
\node[circle, inner sep=0.1pt, fill=black, below right=0.7cm and 1cm of block_e] (point_e) {};
\node[circle, inner sep=0.1pt, fill=black, below right=0.7cm and 0.9cm of block_i1] (point_i1) {};
\node[circle, inner sep=0.1pt, fill=black, below right=0.7cm and 0.9cm of block_i2] (point_i2) {};
\node[circle, inner sep=0.1pt, fill=black, below right=0.7cm and 0.9cm of block_i3] (point_i3) {};
\node[circle, inner sep=0.1pt, fill=black, below right=0.7cm and 1cm of block_r] (point_r) {};

\draw[->] (block_s) -- (point_s)
    node[pos=0.5,fill=white]{$\zeta$};
\draw[->] (block_v) -- (point_v)
    node[pos=0.5,fill=white]{$\zeta$};
\draw[->] (block_e) -- (point_e)
    node[pos=0.5,fill=white]{$\zeta$};
\draw[->] (block_i1) -- (point_i1)
    node[pos=0.65,fill=white]{$\zeta$};
\draw[->] (block_i2) -- (point_i2)
    node[pos=0.65,fill=white]{$\zeta$};
\draw[->] (block_i3) -- (point_i3)
    node[pos=0.65,fill=white]{$\zeta$};
\draw[->] (block_r) -- (point_r)
    node[pos=0.5,fill=white]{$\zeta$};

\draw[->] (inflow_to_s) -- (block_s)
    node[pos=0.5,fill=white]{$\Lambda$};

\draw[-] (block_v) -- (v_to_e);  
\draw[<-] (block_e) -- (v_to_e)
    node[pos=0.5,fill=white]{$\sigma \beta$};

\draw[-] (block_e) -- (e_to_i1);  
\draw[->] (e_to_i1) -- (block_i1)
    node[pos=0.5,fill=white]{$\gamma$};

\draw[-] (block_i1) -- (i1_to_r);  
\draw[->] (i1_to_r) -- (block_r)
    node[pos=0.5,fill=white]{$\delta_1$};

\draw[-] (block_i3) -- (i3_to_r);  
\draw[<-] (block_r) -- (i3_to_r)
    node[pos=0.5,fill=white]{$\delta_3$};

\draw[-] (block_i3) -- (i3_to_d);  
\draw[->] (i3_to_d) -- (block_d)
    node[pos=0.5,fill=white]{$\mu$};

\end{tikzpicture}
        }
    \caption{Flow diagram of the SVEI3RD model.}
    \label{fig: SVEI3RD flow chart}
\end{figure}
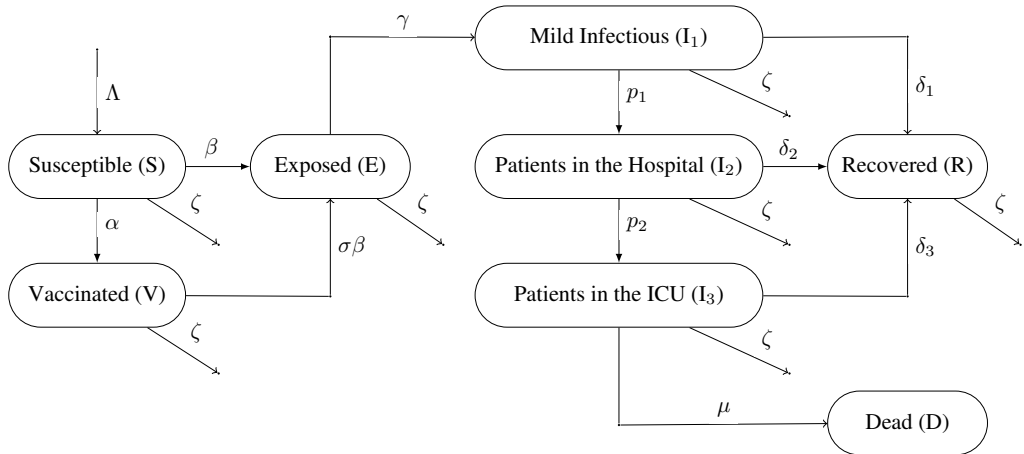

As a feature of this model, the infected population is categorized into three subgroups based on the severity of the illness. In this framework, the initial stage of infection corresponds to individuals experiencing only mild symptoms that do not require hospitalization. These individuals are classified as being in state $\text{I}_1$, which represents the first stage of infection. If symptoms worsen to a level requiring hospitalization, the individual transitions to the second infectious state, $\text{I}_2$. A further progression in severity, requiring admission to the Intensive Care Unit (ICU), moves the individual to the third stage, $\text{I}_3$. The parameter $p_1$ denotes the rate of transitioning from mild symptoms to hospitalization, while $p_2$ represents the probability of requiring ICU care after hospitalization. Additionally, the model takes into account varying recovery rates $\delta_1$, $\delta_2$, and $\delta_3$, which correspond to different infectious stages to better reflect real-world conditions. Finally, it is assumed that only individuals in the ICU are at risk of mortality, provided that all patients receive appropriate medical care as their condition worsens.

Moreover, vital dynamics are also incorporated into the model, where $\Lambda$ denotes the total inflow of population, capturing the effect of both births and net migration, and $\zeta$ denotes the outflow caused by deaths. Using these parameters, the model allows the population size to change over time as it accounts for new individuals entering through birth or immigration and those leaving through death or emigration. 

Additionally, in this study, we incorporate stochasticity into the compartmental model structure. In real-world scenarios, environmental conditions often change unpredictably, leading to random disruptions in population systems. As a result, stochastic epidemic models that account for random noise have been widely developed. Compared to deterministic models, these stochastic approaches provide a more realistic depiction of disease dynamics. Following the methodology outlined in \citet{Gray_2011} and \citet{Din_2020}, we incorporate independent standard Brownian motions into each compartment to represent fluctuations within the state population. The final compartmental model consists of a set of SDEs, depicted in Equation (\ref{eqn: sto SVEI3RD without control}), with its corresponding flow chart illustrated in Figure \ref{fig: SVEI3RD flow chart}. The detailed proof regarding the existence of unique positive solutions to this system of SDEs is provided in Appendix \ref{sec: appendix 1 existence of sde system}.

\begin{equation}
\left\{
\begin{aligned}
dS & = ( \Lambda -(\beta_1 I_1 + \beta_2 I_2 + \beta_3 I_3) S - \alpha S -\zeta S ) dt + \sigma_1 S d W_{1,t},
\\
dV & = (\alpha S - (\beta_1 I_1 + \beta_2 I_2 + \beta_3 I_3) \sigma V -\zeta V) dt + \sigma_2 V d W_{2,t} ,
\\
dE & = ( (\beta_1 I_1 + \beta_2 I_2 + \beta_3 I_3) (S + \sigma V )- \gamma E -\zeta E ) dt + \sigma_3 E d W_{3,t},
\\
dI_1 & = (\gamma E - (\delta_1 + p_1) I_1 -\zeta I_1 ) dt + \sigma_4 I_1 d W_{4,t},
\\
dI_2 & = (p_1 I_1 - (\delta_2 + p_2) I_2 -\zeta I_2) dt + \sigma_5 I_2 d W_{5,t},
\\
dI_3 & = (p_2 I_2 - (\delta_3 + \mu ) I_3 -\zeta I_3 ) dt + \sigma_6 I_3 d W_{6,t},
\\
dR & = (\delta_1 I_1 + \delta_2 I_2  + \delta_3 I_3 -\zeta R) dt + \sigma_7 R d W_{7,t},
\\
dD & = ( \mu I_3 ) dt + \sigma_8 D d W_{8,t}.
\label{eqn: sto SVEI3RD without control}
\end{aligned}\right.
\end{equation}

\subsubsection{The dependence structure between vaccination administration and hospitalization}
\label{subsec in setup: control model setup}

To account for the trade-offs faced by the government, we incorporate the dependence structure between the vaccination administration rate and the hospitalization rate into our model design. Our analysis assumes that vaccination efforts are distributed uniformly among all individuals classified as susceptible during the pandemic period. This assumption ensures that each individual in the Susceptible state has an equal likelihood of transitioning to the Vaccinated state. Furthermore, as the model focuses on a short timeframe, we assume that immunity conferred by vaccination remains effective throughout the model's duration. Consequently, we do not account for the possibility of individuals re-entering the Susceptible state after transitioning to the Recovered state.

Empirical evidence from the literature indicates a significant relationship between the hospitalization rate and the vaccination administration rate, as established by studies such as \citet{Uzun_2022} and \citet{Amato_2020}. Building on the premise that the primary goal of vaccine development is to reduce hospitalizations, we propose that the hospitalization rate, defined as the transition rate from state $\text{I}_1$ to $\text{I}_2$, depends on changes in vaccination administration rates, which are captured by the rate at which susceptible individuals enter the Vaccinated state. This suggests that effective vaccination strategies lead to a decline in hospitalizations associated with the disease. In light of this, we revise our compartmental model to incorporate the hypothesis that the hospitalization rate, $p_1$, is directly influenced by the control variable, vaccination administration rate, $\alpha_t$.

This concludes the section on the epidemic model within our framework, which will serve as a tool to capture the dynamics of pandemic spread as well as the effects of vaccination. By employing this approach together with a subsequent economic cost model, we aim to determine the optimal vaccination strategy that achieves widespread and efficient vaccination coverage while minimizing government expenditures.

\subsection{The expenditure function}
\label{subsec in setup: objective function formulation}

The primary objective of this work is to minimize the total government expenditures resulting from the pandemic. Drawing on the methodology of \citet{Caulkins_2023}, the objective function for our optimization problem encompasses four key components: vaccination costs, quarantine costs, economic costs, and healthcare costs.

Vaccination costs identified as $V_v(\alpha_t)$ represent the first component of government expenditures. According to a report\footnote{Public health response to COVID-19. Available online at: https://www.aph.gov.au/About\_Parliament/Parliamentary\_Departments/

Parliamentary\_Library/pubs/rp/BudgetReview202021/PublicHealthResponseCOVID-19.} by the Australian Government, these costs accounted for a substantial share of public spending during the epidemic. Following the approach of \citet{Lenhart_2007}, vaccination costs can be estimated based on the dose administration rate, expressed as
$$V_v(\alpha_t) = c_1\alpha_t^2,$$
where $c_1$ denotes the cost parameter associated with vaccine administration, and $\alpha_t$ indicates the vaccination administration rate at time $t$.

Furthermore, many countries implement lockdowns or mandatory quarantines for individuals at risk. Regarding these policies, some governments provide quarantine subsidies denoted as $V_q(\alpha_t)$ during lockdown periods, which comprises another segment of government expenditures. We model this type of cost as
$$V_q(\alpha_t) = c_2 E_t(\alpha_t).$$
This expression refers to the financial support provided by the government to individuals exposed to risk, a measure that is also essential during the pandemic.

Another significant component of government expenditures during the pandemic is healthcare expenditures expressed as $V_h(\alpha_t)$. Pandemic-related costs are influenced by the treatment expenses for individuals requiring different levels of care, which include medication, hospitalization, and intensive care. The parameters $c_3$, $c_4$, and $c_5$ are chosen to capture these specific healthcare costs. As a result, the total healthcare cost is given by
$$V_h(\alpha_t) = c_3 I_{1,t}(\alpha_t) + c_4 I_{2,t}(\alpha_t) + c_5 I_{3,t}(\alpha_t). $$

Finally, the last cost component in the expenditure function relates to the economic cost referred to as $V_l (\alpha_t)$, which are primarily associated with the decline in economic output. Since our analysis focuses on the short term, the Cobb-Douglas production function is simplified under the assumption that capital adjustments are not feasible. Based on the premise that human capital is the main driver of production changes during an outbreak, we model the economic cost in terms of labour productivity. During the pandemic, we assume that individuals who are exposed to or infected with the disease cannot work. To capture the working population, we use $\psi$ to represent the percentage of workers in the total population. This is defined as the product of the proportion of the working-age population and the labour force participation rate, which together reflect the share of workers relative to the overall population. Pandemic dynamics result in a reduction in total labour input and this leads to the following estimation of economic loss:
$$V_l (\alpha_t) = c_6 \psi \{ 1 -[S_t(\alpha_t) + V_t(\alpha_t) + R_t(\alpha_t)]\}.$$ 
Throughout this scenario, disease decreases the labour force over time, thus adversely impacting economic output. We estimate the resulting losses to the economy by quantifying these changes in the working population.

To ensure clarity, the overall objective function is defined as:
\begin{equation}
J= \mathbb{E} \bigg[\ \int_{0}^{T}  V_v(\alpha_t) + V_q(\alpha_t) + V_h(\alpha_t) + V_l(\alpha_t) \ \textrm{d}t \bigg].
\label{eqn: objective functn}
\end{equation}
This expression represents the sum of the four cost components over time: vaccination costs, quarantine subsidy costs, healthcare costs, and economic losses. The primary objective of the proposed framework is to minimize cumulative expenditures across these dimensions throughout the pandemic period. The optimization problem is therefore framed as minimizing cumulative expenditures, $J$, by adjusting the control variable, the vaccination administration rate, $\alpha_t$, over time. After formulating this control setup, we demonstrate the existence of the solution to this stochastic control problem, as detailed in Appendix \ref{sec: appendix 8 Existence of control Solution}.

\section{Model solving}
\label{sec: quantitative analysis}

After developing our framework, we turn to the numerical methodologies used to address the proposed problem. In recent years, machine learning has garnered significant attention because of its broad range of potential applications across numerous scientific disciplines. Specifically, neural networks emerge as powerful, reliable, and efficient computational methods, largely because of their role as universal function approximators, as demonstrated by \citet{Hornik_1989}. According to \citet{Balázs_2001}, any nonlinear function can be approximated given a sufficient number of neurons and a well-chosen configuration. This principle underpins the development of neural networks for the study of dynamic systems. One well-known approach to using neural networks for dynamic systems is the Physics Informed Neural Networks (PINN), which enables the model to learn not only from observed data but also from the intrinsic dynamics within those data. PINN provides a robust framework for solving SDE-based models and inferring parameters, as demonstrated by \citet{Raissi_2019}, \citet{Meng_2020}, and \citet{Chen_2021_pinn}.

In the proposed framework, we apply this concept by estimating the compartmental model's parameters using neural networks. Recent work has explored PINNs in the context of SIR-type compartmental models, where \citet{Berkhahn_2022}, \citet{Ning_2023}, and \citet{He_2023} show that PINNs are able to capture the complex dynamics of infectious diseases and calibrate model parameters within the system of equations. Accordingly, we employ this approach in the initial phase of our framework to estimate the unknown parameters in the compartmental model.

Regarding the control problem in the second layer of the proposed framework, we survey the literature on numerical methodologies for solving high-dimensional control problems. It emerges that machine learning emulators can significantly reduce computation time once operational, outperforming traditional local process modules. Consequently, deep neural networks possess powerful nonlinear fitting capabilities for approximating high-dimensional functions and offer strong potential to solve high-dimensional partial differential equations (PDEs). Numerous researchers have explored the use of neural networks for optimal control problems, including \citet{Effati_2013}, \citet{Han_2016}, \citet{Bach_2017}, \citet{Böttcher_2022}, \citet{Ji_2022}, and \citet{Raissi_2024}. Among these methods, each specific neural network structure tends to have its own advantages and drawbacks. Depending on the particular problem setup and the available training dataset, researchers should select the appropriate network architecture that is most suitable for their needs.

Within our framework, we employ the approach developed by \citet{Han_2016} to solve the stochastic optimal control problem. This fundamental method uses neural networks to derive the optimal control solution directly. Although more advanced network structures could also address this problem, they tend to be computationally demanding and require large training datasets to ensure generalization, which would significantly increase data processing time. Given these considerations, our framework is designed to utilize the most recent disease transmission data and provide updated vaccination strategies to social planners in response to changing conditions. The control problem therefore should be solved within a reasonable timeframe. By balancing accuracy and computational efficiency, we have chosen this approach, and a more detailed discussion of the selected network structure is presented in Section \ref{subsec in control solving: nn to solve high dimension problem}.

\subsection{PINN for parameter calibration}

Relying on the high-dimensional structure of our compartmental models, we utilize the PINN approach developed by \citet{Raissi_2019} to numerically calibrate parameters in Equation \ref{eqn: sto SVEI3RD without control}). PINNs offer the advantage of combining data-fitting capabilities with adherence to the underlying physics of the system. This dual capability enables PINNs to accurately describe physical processes, which leads to their application in a broad array of solving differential equation problems. On the basis of this neural network framework, \citet{Shaier_2021} extends the use of PINNs to SIR-type compartmental models by introducing the DINN approach. This approach demonstrates how neural networks can identify unique parameters and capture patterns of disease spread. By incorporating prior knowledge through the DINN framework, the model search space is effectively constrained, thus reducing the amount of data required. Building on this methodology, we also utilize PINNs with the structure shown in Figure \ref{fig: pinn structure}. This approach ensures that the network outputs remain consistent with the measured data and the physical laws described by the system of differential equations when estimating the parameters of the compartmental model. By leveraging this synergy between precise data fitting and compliance with physical laws, PINNs emerge as a powerful platform for addressing complex epidemiological challenges.

To illustrate the application of PINN within our compartmental framework, we first introduce the notation used in designing the neural network for the SVEI3RD model. In this framework, the compartment states $\{S_{t_i}, V_{t_i}, E_{t_i}, I_{1,{t_i}}, I_{2,{t_i}}, I_{3,{t_i}}, R_{t_i}, D_{t_i}\}$ are collectively denoted as $X_{t_i}$, while the set of trainable parameters $\{\alpha, \beta_1, \beta_2, \beta_3, \sigma, \gamma, \delta_1, \delta_2, \delta_3, p_1, p_2, \mu \}$ are referred to as $\Theta$. Each time point is represented by $t_i$, and the data consists of $N$ points, ending at time $t_N$.

During the training process, the neural network is trained using data detailing the spread of the disease over time. As the model learns the underlying system dynamics, it generates estimations for the parameters driving these dynamics. Given that our compartmental model includes a wide range of parameters requiring calibration, we employ the procedure described in \citet{Shaier_2021}, which constrains the parameter search to predefined grids informed by existing literature. In addition, we evaluated the model's performance using a test dataset to validate its prediction accuracy.

In this case, the input of the neural network is time variable $t_i$, and its output is a tensor denoted as $X_{t_i}^{NN}$ which consists of $\{S^{NN}_{t_i}, V^{NN}_{t_i}, E^{NN}_{t_i}, I^{NN}_{1,{t_i}}, I^{NN}_{2,{t_i}}, I^{NN}_{3,{t_i}}, R^{NN}_{t_i}, D^{NN}_{t_i}\}$ and it provides an estimated representation of the disease's compartment states at each time step. The parameters of the differential equations $\Theta$, which are associated with the physically constrained loss component, are iteratively adjusted to minimize the total loss function as part of the training process. The entire network is trained using the backpropagation algorithm introduced by \citet{Hecht-Nielsen_1992}. Appendix \ref{subsec: appendix 5 det model fitting algorithm} summarizes the algorithm used to apply PINNs for deterministic compartmental model parameter estimation.

In relation to the loss function, the governing physical principles of the system of equations are integrated directly into the learning process. Specifically, there are two principal components to the loss function: a data fidelity component and a physics-based component. The data fidelity component represents the extent to which the neural network's predictions match the empirical data, typically measured using the mean squared error (MSE) between the fitted results $X_{t_i}^{{NN}}$ generated by the neural network and the observed values $X_{t_i}$ over time. The physics-based term, in contrast, ensures that the neural network’s predictions conform to the physical laws underlying the system. Since the PINN architecture is designed to calibrate deterministic compartmental models, the physical residual loss is computed by comparing the automated differentiation of the neural network $\dot{X}_{t_i}^{{NN}}= \{\dot{S}^{NN}_{t_i}, \dot{V}^{NN}_{t_i}, \dot{E}^{NN}_{t_i}, \dot{I}^{NN}_{1,{t_i}}, \dot{I}^{NN}_{2,{t_i}}, \dot{I}^{NN}_{3,{t_i}}, \dot{R}^{NN}_{t_i}, \dot{D}^{NN}_{t_i}\}$ with the ordinary differential equations (ODEs) governing compartments over time. 

Accordingly, the aggregate loss function for a general PINN, which is a weighted summation of these two terms, is defined as:

\begin{equation*}
\begin{aligned}
    L &= \lambda_{{Data}} L_{{Data}} + \lambda_{DE} L_{DE}, 
    \label{eqn: pinn loss}
\end{aligned}
\end{equation*}
where
 \begin{equation*}
    \begin{aligned}       
    L_{Data} = \frac{1}{N} \sum_{i=1}^{N}&  (X_{t_i} - X_{t_i}^{{NN}})^2 \\
    =\frac{1}{N} \sum_{i=1}^{N} 
    &\bigg[\ (S_{t_i} - S_{t_i}^{{NN}})^2 + (V_{t_i} - V_{t_i}^{{NN}})^2 
        +(E_{t_i} - E_{t_i}^{{NN}})^2 + (I_{1,t_i} - I_{1, t_i}^{{NN}})^2 \\
        +& (I_{2,t_i} - I_{2, t_i}^{{NN}})^2   + (I_{3,t_i} - I_{3, t_i}^{{NN}})^2 
        + (R_{t_i} - R_{t_i}^{{NN}})^2 + (D_{t_i} - D_{t_i}^{{NN}})^2 \bigg]\
    \end{aligned}
    \end{equation*}
and
    \begin{equation*}
    \begin{aligned} 
    L_{DE} 
    =\frac{1}{N} \sum_{i=1}^{N} 
   & \bigg\{ [\Lambda -(\beta_1 I_{1,t_i} + \beta_2 I_{2,t_i} + \beta_3 I_{3,t_i}) S_{t_i} - \alpha S_{t_i} -\zeta S_{t_i} -  \dot{S}_{t_i}^{NN}]^2 \\
    +& [\alpha S_{t_i} - (\beta_1 I_{1, t_i} + \beta_2 I_{2, t_i} + \beta_3 I_{3, {t_i}}) \sigma V_{t_i} -\zeta V_{t_i}-  \dot{V}_{t_i}^{NN} ]^2 \\
   + &[ (\beta_1 I_{1,t_i} + \beta_2 I_{2,t_i} + \beta_3 I_{3,t_i}) (S_{t_i} + \sigma V_{t_i}) - \gamma E_{t_i} -\zeta E_{t_i} -  \dot{E}_{t_i}^{NN}]^2 \\
    +& [\gamma E_{t_i} - (\delta_1 + p_1) I_{1,t_i} -\zeta I_{1,t_i}-  \dot{I}_{1,t_i}^{NN}]^2 \\
     +& [p_1 I_{1,t_i} - (\delta_2 + p_2) I_{2,t_i} -\zeta I_{2,t_i}-  \dot{I}_{2,t_i}^{NN}]^2 \\
      +& [p_2 I_{2,t_i} - (\delta_3 + \mu ) I_{3,t_i} -\zeta I_{3,t_i} -  \dot{I}_{3,t_i}^{NN}]^2 \\
      +& [\delta_1 I_{1,t_i} + \delta_2 I_{2,t_i}  + \delta_3 I_{3,t_i} -\zeta R_{t_i}-  \dot{R}_{t_i}^{NN}]^2 \\
      + &[\mu I_{3,t_i} - \dot{D}_{t_i}^{NN}]^2 \bigg\}.
    \end{aligned}
    \end{equation*}

In this context, the hyperparameters $\lambda_{{Data}}$ and $\lambda_{{DE}}$ define the relative weights assigned to data fidelity and adherence to physical laws within the overall loss function. By adjusting these parameters, researchers can control the trade-off between the model's fidelity to specific data points and its compliance with governing physical equations, thereby enhancing its predictive capability.
\begin{figure}[htbp]
    \centering
    \includegraphics[width=\textwidth]{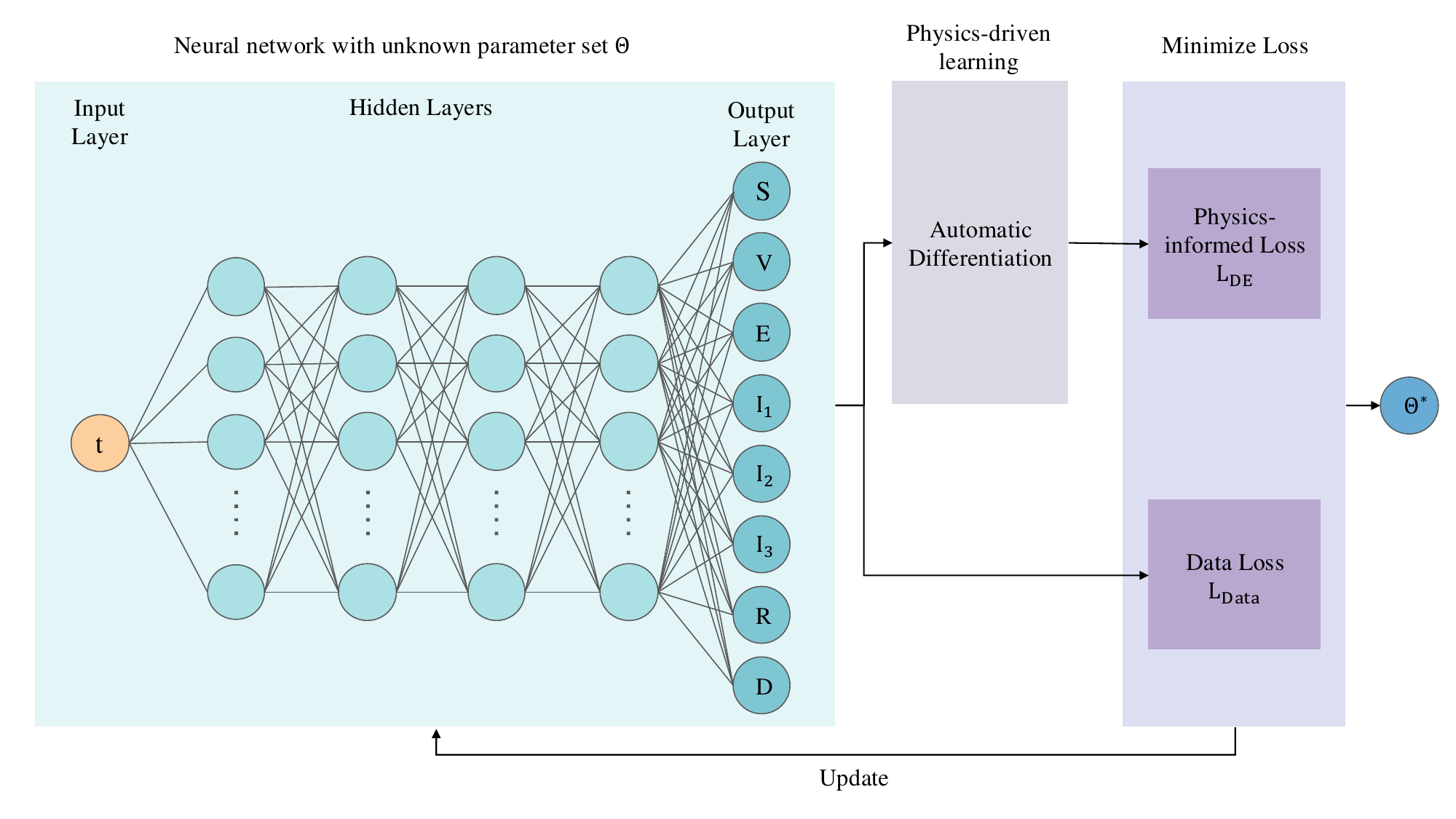}
    \caption{Structural overview of the Physics Informed Neural Network (PINN).}
    \label{fig: pinn structure}
\end{figure}
Moreover, to account for the stochastic nature of our compartmental model, we modify the PINN algorithm by incorporating an inner loop that performs Monte Carlo iterations by iteratively updating the iteration value of $j$, and the total number of such iterations is represented as $N_{MC}$. To capture the stochasticity of the compartmental model, we perform repeated simulations using varying random noise inputs and compute their average effects. For the construction of the loss function in this stochastic setting, we employ the Euler-Maruyama discretization method, which allows the process to be updated based on the SDE terms. The discrete update equations for each compartment are detailed as follows:
\begin{equation*}
    \begin{aligned}
                S_{t_{i+1},j}^{{DE}} &= S_{t_i,j}^{{NN}} + [\Lambda -(\beta_1 I_{1,t_i,j}^{{NN}} + \beta_2 I_{2, t_i,j}^{{NN}} + \beta_3 I_{3,t_i,j}^{{NN}}) S_{t_i,j}^{{NN}} - \alpha_t S_{t_i,j}^{{NN}} -\zeta S_{t_i,j}^{{NN}} ] \Delta t + \sigma_1 \sqrt{\Delta t} \Delta W_{1, t_i,j} ,
                \\
                V_{t_{i+1},j}^{{DE}} &= V_{t_i,j}^{{NN}} + [\alpha_t S_{t_i,j}^{{NN}} - (\beta_1 I_{1, t_i,j}^{{NN}} + \beta_2 I_{2, t_i,j}^{{NN}} + \beta_3 I_{3, t_i,j}^{{NN}}) \sigma V_{t_i,j}^{{NN}} -\zeta V_{t_i,j}^{{NN}}] \Delta t + \sigma_2 \sqrt{\Delta t} \Delta W_{2,{t_i},j},
                \\
                E_{t_{i+1},j}^{{DE}} &= E_{t_i,j}^{{NN}} + [(\beta_1 I_{1,t_i,j}^{{NN}} + \beta_2 I_{2,t_i,j}^{{NN}} + \beta_3 I_{3,t_i,j}^{{NN}}) (S_{t_i,j}^{{NN}} +  \sigma V_{t_i,j}^{{NN}} )- \gamma E_{t_i,j}^{{NN}} -\zeta E_{t_i,j}^{{NN}} ] \Delta t + \sigma_3 \sqrt{\Delta t} \Delta W_{3,t_i,j} ,
                \\
                I_{1, t_{i+1},j}^{{DE}} &= I_{1,t_i,j}^{{NN}} + [\gamma E_{t_i,j}^{{NN}} - (\delta_1 + p_1(\alpha_t)) I_{1,t_i,j}^{{NN}} -\zeta I_{1,t_i,j}^{{NN}} ] \Delta t + \sigma_4 \sqrt{\Delta t} \Delta W_{4,t_i,j} ,
                \\
                I_{2, t_{i+1},j}^{{DE}} &= I_{2,t_i,j}^{{NN}} + [p_1(\alpha_t) I_{1,t_i,j}^{{NN}} - (\delta_2 + p_2) I_{2,t_i,j}^{{NN}} -\zeta I_{2,t_i,j}^{{NN}}] \Delta t + \sigma_5 \sqrt{\Delta t} \Delta W_{5,t_i,j} ,
                \\
                I_{3, t_{i+1},j}^{{DE}} &= I_{3,t_i,j}^{{NN}} + [p_2 I_{2,t_i,j}^{{NN}} - (\delta_3 + \mu ) I_{3,t_i,j}^{{NN}} -\zeta I_{3,t_i,j}^{{NN}} ] \Delta t + \sigma_6 \sqrt{\Delta t} \Delta W_{6,t_i,j} ,
                \\
                R_{t_{i+1},j}^{{DE}} &= R_{t_i,j}^{{NN}} + (\delta_1 I_{1,t_i,j}^{{NN}} + \delta_2 I_{2,t_i,j}^{{NN}}  + \delta_3 I_{3,t_i,j}^{{NN}} -\zeta R_{t_i,j}^{{NN}}) \Delta t + \sigma_7 \sqrt{\Delta t} \Delta W_{7,t_i,j},
                \\
                D_{t_{i+1},j}^{{DE}} &= D_{t_i,j}^{{NN}} +   \mu I_{3,t_i,j}^{{NN}} \Delta t + \sigma_8 \sqrt{\Delta t} \Delta W_{8,t_i,j}.
            \end{aligned}
\end{equation*}

In this set of equations, the time step size is denoted by $\Delta t=t_{i+1}-{t_i}$, and $Z$ is defined as the set of noise intensity parameters given by $\{\sigma_1, \sigma_2, \sigma_3, \sigma_4, \sigma_5, \sigma_6, \sigma_7, \sigma_8\}$, each scaled by the square root of $\Delta t$. The term $\Delta W_{t_i,j}$ represents the increment at each time point during the $j$-th iteration and follows a standard normal distribution. Under this model configuration, the overall loss function is calculated as the average over iterations, denoted as:

\begin{equation*}
    \begin{aligned}
    L = \lambda_{{Data}} L_{{Data}} + \lambda_{{DE}} L_{{DE}}  = \lambda_{{Data}} L_{{Data}} + \frac{1}{N_{MC}} \sum_{j=1}^{N_{MC}} \lambda_{{DE}} L_{{DE, j}} ,
     \end{aligned}
\end{equation*}

where 
\begin{equation*}
    \begin{aligned}
        L_{{Data}} &= \frac{1}{N} \sum_{i=1}^{N} \bigg[\
    (S_{t_i} - \bar{S}_{t_i}^{NN})^2 + (V_{t_i} - \Bar{V}_{t_i}^{{NN}})^2 
            +(E_{t_i} - \Bar{E}_{t_i}^{{NN}})^2 + (I_{1,t_i} - \Bar{I}_{1, t_i}^{{NN}})^2 \\
            &+ (I_{2,t_i} - \Bar{I}_{2, t_i}^{{NN}})^2  + (I_{3,t_i} - \Bar{I}_{3, t_i}^{{NN}})^2  
            + (R_{t_i} - \Bar{R}_{t_i}^{{NN}})^2 + (D_{t_i} - \Bar{D}_{t_i}^{{NN}})^2 \bigg]\
\end{aligned}
\end{equation*}
with 
\begin{equation*}
    \begin{aligned}
    \Bar{S}_{t_i}^{\text{NN}} = \frac{1}{{N}_{MC}} \sum_{j=1}^{{N}_{MC}} {S}_{t_i,j},\
    \Bar{V}_{t_i}^{\text{NN}} = \frac{1}{{N}_{MC}} \sum_{j=1}^{{N}_{MC}} {V}_{t_i,j},\ 
    \Bar{E}_{t_i}^{\text{NN}} = \frac{1}{{N}_{MC}} \sum_{j=1}^{{N}_{MC}} {E}_{t_i,j},\
    \Bar{I}_{1,t_i}^{\text{NN}} = \frac{1}{{N}_{MC}} \sum_{j=1}^{{N}_{MC}} {I}_{1,t_i,j},\ \\
    \Bar{I}_{2,t_i}^{\text{NN}} = \frac{1}{{N}_{MC}} \sum_{j=1}^{{N}_{MC}} {I}_{2,t_i,j},\
    \Bar{I}_{3,t_i}^{\text{NN}} = \frac{1}{{N}_{MC}} \sum_{j=1}^{{N}_{MC}} {I}_{3,t_i,j},\
    \Bar{R}_{t_i}^{\text{NN}} = \frac{1}{{N}_{MC}} \sum_{j=1}^{{N}_{MC}} {R}_{t_i,j},\
    \Bar{D}_{t_i}^{\text{NN}} = \frac{1}{{N}_{MC}} \sum_{j=1}^{{N}_{MC}} {D}_{t_i,j},\ .
\end{aligned}
\end{equation*}
and 
\begin{equation*}
            \begin{aligned}
                 L_{{DE, j}} =\frac{1}{N} \sum_{i=0}^{N-1}  
            &\bigg[\ (S_{t_{i+1}}^{{DE}} - S_{t_{i+1}}^{{NN}})^2 + (V_{t_{i+1}}^{{DE}}- V_{t_{i+1}}^{{NN}})^2 
            +(E_{t_{i+1}}^{{DE}} - E_{i+1}^{{NN}})^2 + (I_{1,t_{i+1}}^{{DE}} - I_{1,t_{i+1}}^{{NN}})^2  \\
            + &(I_{2,t_{i+1}}^{{DE}} - I_{2,t_{i+1}}^{{NN}})^2  + (I_{3,t_{i+1}}^{{DE}} - I_{3,t_{i+1}}^{{NN}})^2  
            + (R_{t_{i+1}}^{{DE}}- R_{t_{i+1}}^{{NN}})^2 + (D_{t_{i+1}}^{{DE}} - D_{t_{i+1}}^{{NN}})^2 \bigg]\ ,
            \end{aligned}
            \end{equation*}
            
The detailed algorithm for this stochastic disease neural network modelling is presented in Algorithm \ref{subsec: appendix 5 sto model fitting algorithm}. Through this iterative process, the method improves the stability of parameter estimates by mitigating the influence of randomness.

\subsection{Deep neural network for high-dimensional stochastic control}
\label{subsec in control solving: nn to solve high dimension problem}

Following the estimation of parameters for the compartmental model, we turn to the methodology for deriving the numerical solution to the control problem. The formulation of our compartmental model results in a stochastic control problem that is inherently high-dimensional. In traditional stochastic control theory, problems are typically addressed using the dynamic programming principle. However, such methods often face significant technical challenges when applied to high-dimensional settings, particularly due to the issue of "the curse of dimensionality", as described in \citet{Bellman_2015}. As a result, it becomes essential to consider modern machine learning techniques as a viable alternative.

\begin{figure}[htbp]
    \centering
    \input{Figure_codes/Han_NN}
    \caption{Neural network architecture for solving the high-dimensional control problem.}
    \label{fig: control nn}
\end{figure}
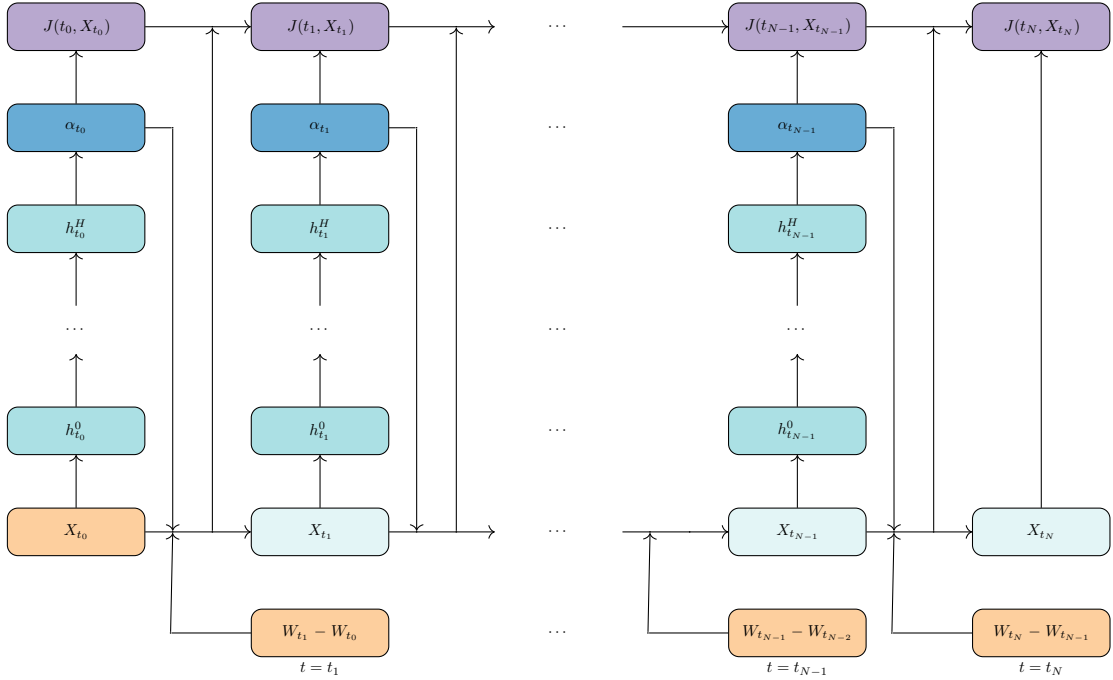

Recent advancements in artificial intelligence have demonstrated remarkable progress across various complex challenges, highlighting that deep neural networks can be highly useful in addressing high-dimensional problems. Accordingly, we utilize the network proposed by \citet{Han_2016} to solve the optimal control problem. Rather than estimating the value function, this approach directly approximates the control at each time step. Specifically, this neural network employs multiple feedforward subnetworks, each tasked with approximating the control variable at an individual time step, and these subnetworks are collectively trained to form the deep neural network.

According to our model, the aggregate government expenditure serves as the loss function for the neural network, with the vaccination administration rate functioning as the control variable. The proposed problem is formulated as:
\begin{equation*}
\min_{\alpha_{t_n},n=0,...,N-1 } E [J_{t_N}|X_{t_0}] = \min_{\alpha_{t_n},n=0,...,N-1 } E[\sum_{n=0}^{N-1} j_{t_n}(X_{t_n},a_{t_n}(X_{t_n})) +j_{t_N}(X_{t_N})|X_{t_0}],
\end{equation*}
where $j_{t_n}(X_{t_n},\alpha_{t_n})$ denotes the intermediate cost, and $j_{t_N}(X_{t_N})$ denotes the final cost, which results in $J_{t_N}$, representing the total cost. Also, we define the cumulative cost as $J_{t_n} = \sum_{\tau=0}^{n} j_{t_\tau} (X_{t_\tau}, \alpha_{t_\tau}(X_{t_\tau}))$ with $n<N$.

In this case, adapting the weights and bias parameters of the subnetwork allows us to directly approximate the control variable while minimizing the objective function. Accordingly, the optimization problem arising from our control framework is stated as follows:
\begin{equation*}
\min_{\{ \phi_{t_n}\}_{n=0}^{N-1} } \sum_{n=0}^{N-1}  j_{t_n} ( X_{t_n}, \alpha_{t_n}(X_{t_n} | \phi_{t_n}) )  + j_{t_N} ( X_{t_N} ),
\label{eqn: loss for high-dimension nn}
\end{equation*}
where $X_{t_n}$ is the set of compartment state variables, $\phi_{t_n}$ represents the parameters inside each subnetwork we aim to optimize, and $\alpha_{t_n}$ denotes the control variable.

The general network structure is illustrated in Figure \ref{fig: control nn} and is characterized by three types of connections. The core component of the network is a set of multilayer feedforward neural networks, which are used to approximate the control variables at each time point $t_n$, represented as $X_{t_n} \rightarrow h_{t_n}^1 \rightarrow \cdots \rightarrow h_{t_n}^H \rightarrow \alpha_{t_n}$. The parameters $\phi_{t_n}$ within each subnetwork are adjusted during the training process. The second type of connection concerns the transition of combined state and noise variables over time. This process is denoted by $(X_{t_n}, \alpha_{t_n}, W_{t_{n+1}} - W_{t_{n}}) \rightarrow X_{t_{n+1}}$, allowing for the linkage of subnetworks across different time points. Finally, the third type of connection directly contributes to the network's final output, described as $(X_{t_n}, \alpha_{t_n}, J_{t_n}) \rightarrow J_{t_{n+1}}$. Based on this structural design, the operation of the proposed deep neural network begins with sampling the noise terms $\Delta W_{t_n} = W_{t_{n+1}} - W_{t_{n}}$ at each time step. These noise terms are then combined with state proportion data and a feedforward neural network is subsequently applied to approximate the control variable at each point in time. Following this step, the total expenditure for the modelling period is calculated, serving as the deep neural network's cost function. For implementation, the neural network is constructed and trained using the TensorFlow library. Parameter optimization is conducted by employing the Adam optimizer to minimize the overall loss function.

\section{Numerical case study}
\label{sec: numerical results}
In this section, we present a numerical case study that utilizes data from Victoria, Australia, drawn from the COVID-19 pandemic. We analyze the optimal vaccination strategy within the economic epidemiology framework. This case study demonstrates the practicality and applicability of the proposed approach.

\subsection{The Victoria dataset}
\label{subsec in fitting: dataset}
To conduct the numerical analysis, historical COVID-19 data for Victoria were obtained from the COVID-19 Australia dataset\footnote{COVID-19 Data for Australia. Available online at the address https://github.com/M3IT/COVID-19 Data.}. This dataset includes a variety of variables, which are detailed in Table \ref{table: dataset variable names}. The dataset provides daily COVID-19 records for all Australian states. For this study, records specific to Victoria, Australia, during the COVID-19 pandemic, are extracted to analyze the disease transmission process in the modelling phase.

\begin{table}[htbp]
    \captionsetup{skip=5pt}
	\centering
	\begin{tabular}{cccc}
		\toprule
		Date & State & Confirmed & Confirmed\_cum \\
            Deaths & Deaths\_cum & Tests & Tests\_cum \\
            Recovered & Recovered\_cum & Hosp & Hosp\_cum \\
            Vaccines & Vaccines\_cum  & ICU	& ICU\_cum  \\
		\bottomrule
	\end{tabular}
        \caption{Key variables from the Australian COVID-19 dataset.}
	\label{table: dataset variable names}
\end{table}

An exploratory data analysis is performed before fitting the theoretical models to the observed data. The number of confirmed cases is interpreted as the total count of individuals infected with COVID-19. Death and recovery records are directly incorporated into the model framework to calculate values for the Dead and Recovered states. The dataset also classifies infected individuals by symptom severity. Individuals requiring hospitalization are categorized as "Hosp" and "Hosp\_cum", while those admitted to the Intensive Care Unit are listed as "ICU" and "ICU\_cum". The number of mildly infected individuals was estimated by subtracting the total number of hospitalized and intensive care patients from the overall infection count at each point in time. Additionally, testing data is used to approximate the number of exposed individuals. In accordance with Victorian Government guidelines, testing is offered to individuals exhibiting symptoms indicative of COVID-19; thus, those tested are assumed to be exposed to the virus.

\begin{table}[htbp]
    \captionsetup{skip=5pt}
	\centering
	\begin{tabular}{cccccc}
		\toprule
		Dataset & State  &  Initial Popn Proportion  & Dataset  &  State & Initial Popn Proportion \\
		\midrule
		Training Set  & S  &  0.554181  & Test Set  & S & 0.191591 \\
            Training Set  & V  &  0.429185  & Test Set & V & 0.779949 \\
            Training Set  & E  &  0.010225  & Test Set & E & 0.009535 \\
            Training Set  & $\text{I}_1$  &  0.001827  & Test Set & $\text{I}_1$ & 0.001896 \\
            Training Set  &  $\text{I}_2$  &  0.000075  & Test Set & $\text{I}_2$ & 0.000044 \\
            Training Set  &  $\text{I}_3$  &  0.000014  & Test Set & $\text{I}_3$ & 0.000006 \\
            Training Set  & R  &  0.004361  & Test Set & R & 0.016773 \\
            Training Set  & D  &  0.000132  & Test Set & D & 0.000206 \\
		\bottomrule   
	\end{tabular}
        \caption{Initial population proportions in training and test sets.}
	\label{table: initial state}
\end{table}

This dataset is limited because it records only the total number of vaccinations administered per day without distinguishing between injection types, such as the first dose, the second dose, and booster doses. To address this limitation and support data modelling, we obtain a supplementary dataset containing detailed vaccination records from the Department of Health and Aged Care\footnote{COVID-19 Vaccination Data. Available online at: https://www.health.gov.au/resources/publications.}. Using this dataset, we calculate the vaccination administration rates for each dose and derive the relevant coverage percentages based on the official reports.

Furthermore, we partition the dataset into training and test sets, selecting the starting point of the training set to coincide with a phase when vaccination coverage in Victoria reaches a relatively stable stage. The training dataset spans two months, from October 4, 2021, to December 2, 2021, while the test dataset encompasses a three-week period starting on December 3, 2021. To calculate the state proportions, the number of individuals in each compartment is divided by the total population. Table \ref{table: initial state} displays the initial compartmental proportions for both the training and test sets. For consistency, individuals are categorized as vaccinated only if they have received two doses of the vaccine. Figure \ref{fig: vic vaccination rollout} provides a visualization of the cumulative number of vaccine doses administered in Victoria, highlighting the training and test periods. In this figure, the cyan segment corresponds to the training period, while the deep blue segment represents the test period\footnote{Victoria COVID-19 Breakdown. Available online at: https://covidbaseau.com/vic/.}.

\subsection{Modelling parameters}

\subsubsection{Regression relationship}
\label{subsec in fitting: regression result}

We initiate our numerical analysis by investigating the relationship between hospitalization rates and vaccination administration rates. Studying this relationship prior is crucial for formulating the control problem in a well-founded manner. To identify this correlation using empirical evidence, we perform a regression analysis with the real-life dataset described in Section \ref{subsec in fitting: dataset}.

\begin{figure}[htbp]
\centering
\begin{minipage}[t]{0.48\textwidth}
\centering
\includegraphics[width=0.9\textwidth]{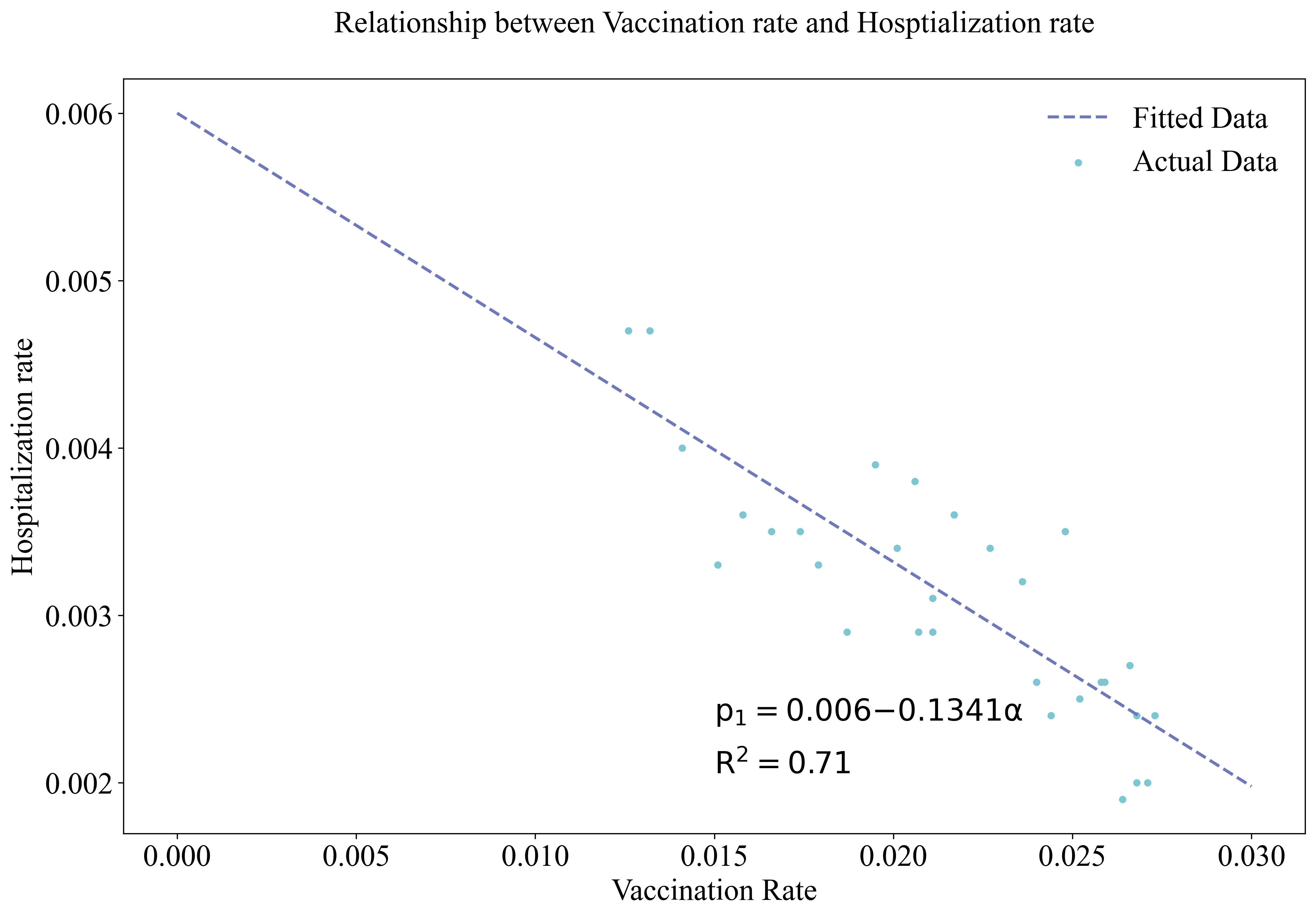}
\captionsetup{width=.9\textwidth} 
\caption{Correlation between hospitalization rates and vaccination administration rates.}
\label{fig: regression}
\end{minipage}
\begin{minipage}[t]{0.48\textwidth}
\centering
\includegraphics[width=0.92\textwidth]{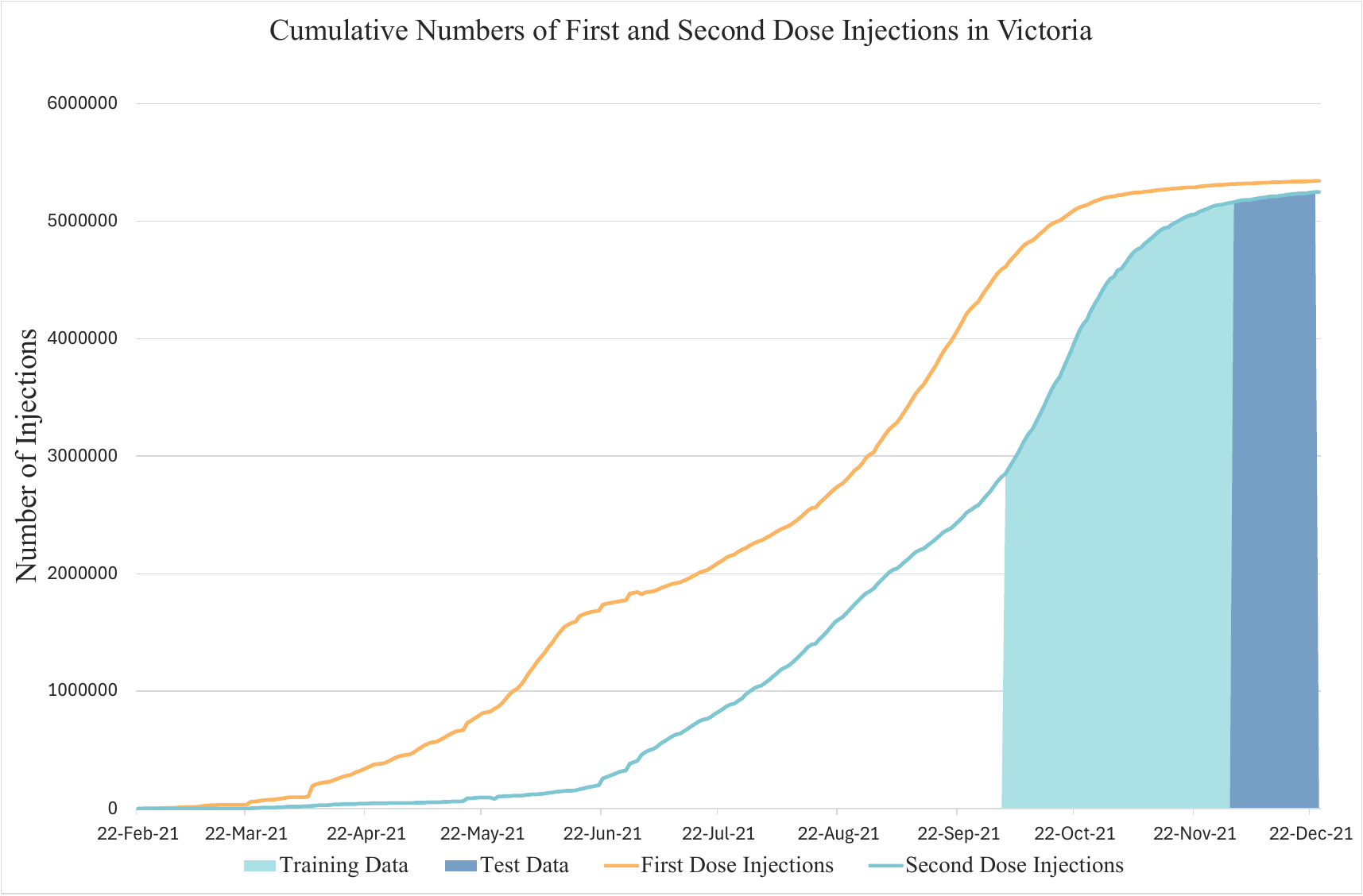}
\captionsetup{width=.9\textwidth} 
\caption{Trajectory of first and second dose vaccination rollout in Victoria, Australia.}
\label{fig: vic vaccination rollout}
\end{minipage}
\end{figure}

It is essential to first disaggregate the total changes in each infectious state into distinct inflow and outflow components to facilitate the analysis. This is because both entries into and exits from compartments contribute to the state population changes over time. Using the change in recovered individuals, denoted by $\Delta {R}_t$, and weighting it by the proportion of mildly infected individuals to the total infected population at that time, we quantify the number of individuals leaving the state $\text{I}_1$, represented as $\Delta_{-} I_{1,t}$. From this, the inflows into state $\text{I}_1$, represented by $\Delta_{+} {I}_{1,t}$, are derived by subtracting the outflows from the total observed change in the state, as expressed in Equation (\ref{eqn: I1 derivation}). A similar methodology is applied to separate the inflows and outflows from states $\text{I}_2$ and $\text{I}_3$, with all calculations ensuring adherence to the principle of non-negative population movements.

\begin{equation}
\left\{
\begin{aligned}
    \Delta_+{ {I}_{1, t}} &= \Delta{{I}_{1, t}} - \Delta_-{{I}_{1, t}} =  \text{max} (0,\Delta{{I}_{1, t}}+\frac{{I}_{1, t}}{{I}_{1, t}+ {I}_{2, t} + {I}_{3, t}} \Delta {R}_t ) \\
    \Delta_+{{I}_{2, t}} &= \Delta{{I}_{2, t}} - \Delta_-{{I}_{2, t}} =  \text{max} (0,\Delta {I}_{2, t}+\frac{{I}_{2, t}}{{I}_{1, t} + {I}_{2, t} + {I}_{3, t}} \Delta {R}_t ) \\
    \Delta_+{{I}_{3, t}} &= \Delta{{I}_{3, t}} - \Delta_-{{I}_{3, t}} = \text{max} (0,\Delta {I}_{3, t}+ \frac{{I}_{3, t}}{{I}_{1, t} + {I}_{2, t} + {I}_{3, t}} \Delta {R}_t )
\end{aligned}\right.
\label{eqn: I1 derivation}
\end{equation}

The vaccination administration rate is determined by dividing the number of vaccination injections by the number of individuals who can be vaccinated. 
Moreover, the rate at which mild infections progress to a severity requiring hospitalization is quantified using the hospital transmission rate, calculated as:
$$p_{1, t} = \frac{\Delta_+ {I}_{2, t}}{{I}_{1, t-1}}.$$

To estimate the empirical parameters, we conducted a regression analysis using historical data from Victoria, Australia. Our results indicate that the relationship between vaccination administration rates and hospitalization rates can be expressed as:
$$
p_1(\alpha_t) = 0.0060 - 0.1341 \alpha_{t}
$$ 
with a p-value of less than $10^{-8}$, underscoring the statistical significance of this relationship. As illustrated in Figure \ref{fig: regression}, hospitalization rates decrease as vaccination administration rates increase. This finding aligns with the results of \citet{Chen_2021}, who documents a negative linear relationship between vaccination rates and transmission rates among symptomatic individuals requiring hospitalization.

The observed relationship is also intuitively plausible, as rising vaccination rates protect a larger proportion of the population, resulting in a decline in hospitalization rates. This analysis underscores the important economic trade-off. Although increasing vaccination rates entails higher public expenditures on immunization programs, it simultaneously reduces healthcare budgets by decreasing hospitalization rates. Further, while expanding vaccination increases policy costs, it produces fiscal benefits through reduced healthcare spending. The empirically observed inverse relationship between vaccination and hospitalization rates confirms the validity of the control problem design, which seeks to minimize total government expenditures.

\subsubsection{Vital dynamic parameters}
\label{subsec in fitting: para calibration}

For the parameter calibration of the compartmental model in this case study, we first estimate the inflow parameter $\Lambda$ and the outflow parameter $\zeta$, which are linked to vital dynamics. These parameters can be directly estimated using data from public sources, and we rely on literature-based approximations to determine these values. The overall population inflow into the system occurs through the Susceptible state, considering the effects of both births and net migration. For the birth rate, we utilize statistics from the ABS report\footnote{Deaths, Australia. Available online at: https://www.abs.gov.au/statistics/people/population/deaths-australia/2022}, calculating it by dividing the total number of births (75,363 in 2021) in Victoria by the overall population and converting this value into a daily rate. This yields a birth rate of $0.0000315$. Additionally, the net migration rate is also derived from ABS statistics
, which indicates 13,100 net migrants in the fourth quarter of 2021. Based on this, the net migration rate for Victoria is estimated as $0.0000217$ on a daily basis. Combining these two rates gives us an overall population inflow rate as 
$$\Lambda =0.0000217+0.0000315=0.000053.$$ 
Moreover, we calculate the death rate based on life expectancy, following a similar methodology to \citet{Adak_2021}. According to Victoria's sex ratio in 2021, the average life expectancy is calculated as the weighted average of male and female life expectancies, and the corresponding daily mortality rate is calculated as: $$\zeta=\frac{1}{(81.7 \times \frac{0.98}{1.98} + 85.7 \times \frac{1}{1.98}) \times 365} = 0.000033.$$

\begin{table}[ht]
	\centering
	\begin{tabular}{cccccc}
		\toprule
		Parameters  &  Values          & Source         & Parameters  &  Values & Source \\
		\midrule
		$\Lambda$ & 0.000053 & Assumed & $\zeta$ & 0.000033 & Assumed \\
            $\beta_1$  & 0.28120  & Calibrated & $\sigma_1 $ & 0.09275 & Calibrated\\
            $\beta_2$  & 0.15838  & Calibrated & $\sigma_2 $ & 0.03887 & Calibrated \\
            $\beta_3$  & 0.03880 & Calibrated & $\sigma_3 $ & 0.07517 & Calibrated\\
            $\sigma$  & 0.06352 & Calibrated & $\sigma_4 $ & 0.06302 & Calibrated\\
            $\gamma$  & 0.30954  & Calibrated  & $\sigma_5 $ & 0.07878 & Calibrated\\
            $\delta_1$  & 0.28505 & Calibrated & $\sigma_6 $ & 0.06123 & Calibrated\\
            $\delta_2$ & 0.28269 & Calibrated  & $\sigma_7 $ & 0.06110 & Calibrated\\
            $\delta_3$  & 0.14206 & Calibrated & $\sigma_8 $ & 0.06199 & Calibrated\\
            $p_2$  & 0.14310 & Calibrated &   $\mu$  & 0.00420  & Calibrated \\
		\bottomrule   
	\end{tabular}
        \caption{Description of parameters in the baseline setting.}
	\label{table: fitted para assumpn}
\end{table}

\subsubsection{Deterministic VS stochastic SVEI3RD model fitting}
\label{subsec in fitting: det vs sto}

In this section, we compare the fitted results for the proposed stochastic model with its deterministic version, as specified in Equation (\ref{eqn: det SVEI3RD without control}), to evaluate the necessity of incorporating stochastic elements in the framework. We employ the PINN approach, as described in Section \ref{subsec in control solving: nn to solve high dimension problem}, to estimate the parameter set $\Theta$ in the compartmental system using our training data. Because the PINN method is highly data-intensive, as highlighted by \citet{Karniadakis_2021}, we implement data augmentation to improve its learning performance. Following the recommendations of \citet{Wen_2020}, we apply cubic splines as a robust method to interpolate training data points, thus transforming daily frequency data into higher-frequency datasets. Specifically, we augment the data points to five, ten, and twenty times their original daily frequency. Using the augmented data, we train our neural network with the training dataset to calibrate parameter values tailored to this case study. In order to assess the effectiveness of the model, no changes are made to the test dataset.

\begin{equation}
\left\{
\begin{aligned}
\frac{dS}{dt} & =  \Lambda -(\beta_1 I_1 + \beta_2 I_2 + \beta_3 I_3) S - \alpha S -\zeta S ,
\\
\frac{dV}{dt} & = \alpha S - (\beta_1 I_1 + \beta_2 I_2 + \beta_3 I_3) \sigma V -\zeta V,
\\
\frac{dE}{dt} & =  (\beta_1 I_1 + \beta_2 I_2 + \beta_3 I_3) S + (\beta_1 I_1 + \beta_2 I_2 + \beta_3 I_3) \sigma V - \gamma E -\zeta E ,
\\
\frac{dI_1}{dt} & = \gamma E - (\delta_1 + p_1) I_1 -\zeta I_1 ,
\\
\frac{dI_2}{dt} & = p_1 I_1 - (\delta_2 + p_2) I_2 -\zeta I_2 ,
\\
\frac{dI_3}{dt} & = p_2 I_2 - (\delta_3 + \mu ) I_3 -\zeta I_3,
\\
\frac{dR}{dt} & = \delta_1 I_1 + \delta_2 I_2  + \delta_3 I_3 -\zeta R,
\\
\frac{dD}{dt} & =  \mu I_3 .
\label{eqn: det SVEI3RD without control}
\end{aligned}\right.
\end{equation} 

We utilize the Adam optimizer from the Tensorflow package, developed by \citet{Kingma_2014}, to update weights, biases, and dynamic parameters in our neural network design. Through iterative parameter updates, the algorithm works to minimize the loss function and improve the accuracy of model predictions. In this case study, our neural network architecture comprises four hidden layers, each with 128 neurons. The tanh activation function is applied to all hidden layers, while the sigmoid function is used for the output layer. The learning rate is set to $10^{-6}$, and the model is trained for $100,000$ epochs. We set the number of Monte Carlo iterations $N_{MC}$ to five to average out the noisy effect. Furthermore, the regularization parameters are set equally, with $\lambda_{{Data}} = \lambda_{{DE}} = 1$ for data loss and residual loss.

\begin{table}[htbp]
    \captionsetup{skip=5pt}
	\centering
	\begin{tabular}{cccccc}
		\toprule
		\multicolumn{2}{c}{} & \multicolumn{2}{c}{Test MSE} & \multicolumn{2}{c}{Test MAE} \\
		\cmidrule(r){3-6}
		  Number of Data points ($N$) & Search grid & Deterministic & Stochastic & Deterministic & Stochastic \\
		\midrule
		Ndays$\times 1$ & 10\% &  0.00544343  &  0.00427258  &  0.00433145   &  0.00306000   \\
		    & 50\% &   0.00506920   &  0.00427108  &   0.00395384    &  0.00305833    \\
		    & 90\% &   0.00505704   &  0.00427115  &   0.00394915   &  0.00305878   \\
            \midrule
		Ndays$\times 5$ & 10\% &  0.00544348  &  0.00418813  &  0.00433146   &  0.00297853  \\
		    & 50\% &  0.00507398  &  0.00418457  &  0.00395876   &  0.00296288  \\
		    & 90\% &  0.00506203  &  0.00418470  &  0.00395400 &  0.00296323   \\
            \midrule
		Ndays$\times 10$  & 10\% &  0.00544343  &  0.00418795  &  0.00433145   &  0.00297842  \\
		       & 50\% &  0.00507392  &  0.00418471  &  0.00395866   &  0.00296296  \\
		       & 90\% &  0.00505700  &  0.00418476  &  0.00394918  &  0.00296320  \\
            \midrule
		Ndays$\times 20$  & 10\% &  0.00544343 &  0.00418805  &  0.00433145  &  0.00297849  \\
		       & 50\% &  0.00506919  &  0.00418478  &  0.00395387   &  0.00296319  \\
		       & 90\% &  0.00505700  &  0.00418480  &  0.00394928  &  0.00296316  \\
		\bottomrule
	\end{tabular}
        \caption{Evaluation of model calibration outcomes using test set data.}
	\label{table: calibration test mse results}
\end{table}

\begin{figure}[htbp]
    \centering
    \includegraphics[width=\textwidth]{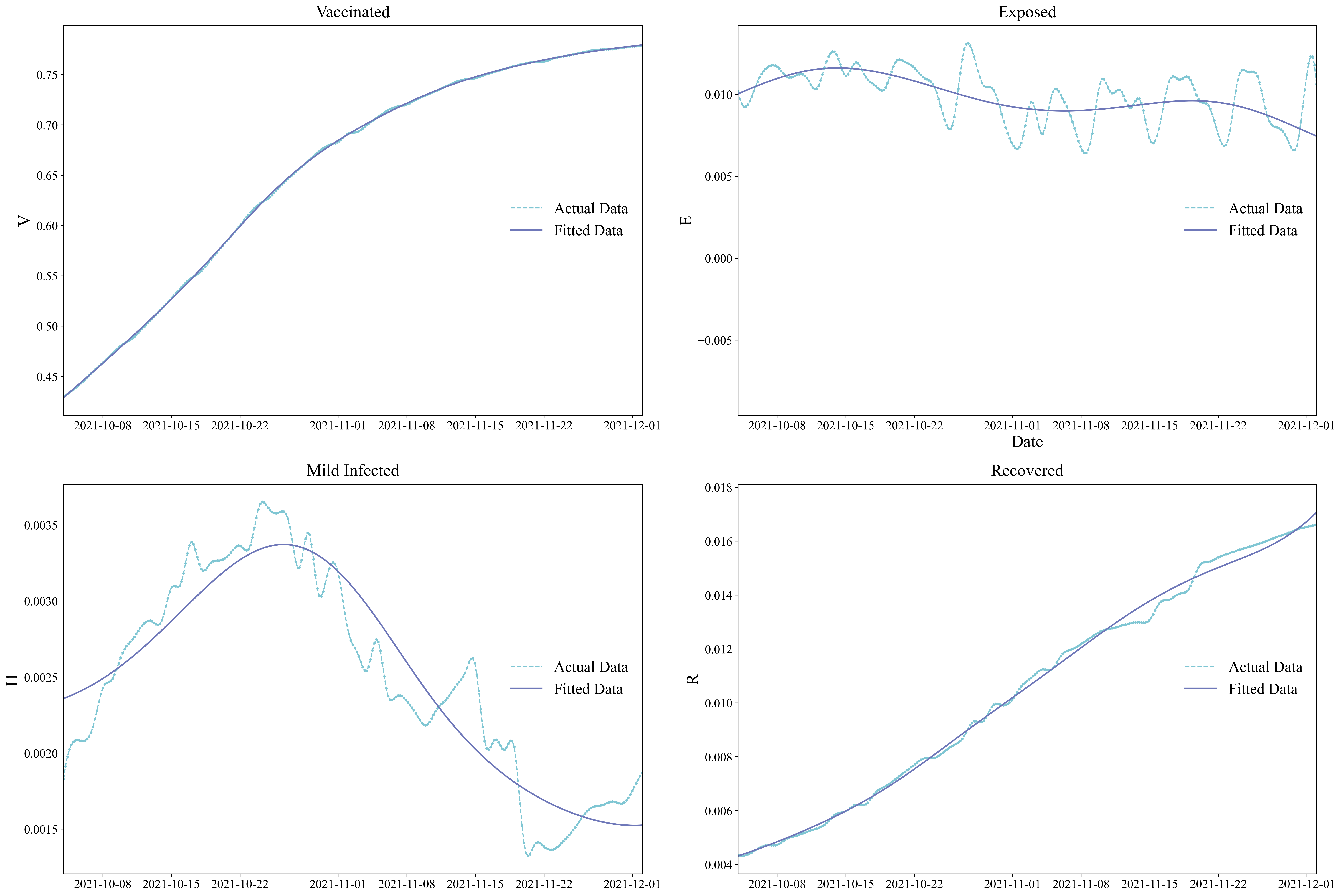}
    \vspace{0.2\baselineskip}
    \caption{Results of model fitting with PINN.}
    \label{fig: training dataset fitting}
\end{figure}
To evaluate the model's performance, we compare the fitted results of the deterministic and stochastic models and use the test MSE and mean absolute error (MAE) to assess the quality of the fit. The results of our model fitting procedure, presented in Table \ref{table: calibration test mse results}, show that across all cases, the stochastic model significantly outperforms the deterministic model. For the search ranges of parameters, we employ a search grid of 50\% around the initial parameter assumptions, which is appropriate for our dataset. Moreover, selecting the data frequency for the stochastic model is also essential. Generally, increasing the number of data points used in training the model enhances the representation of the training dataset's performance. However, overfitting the neural network can occur when there are too many interpolated data points. When balancing predictability and interpretability, we find that the test MSE and MAE increase when the number of data points rises from fivefold to tenfold within a 50\% search grid. Therefore, we select a model with fivefold data points. The calibrated parameters are shown in Table \ref{table: fitted para assumpn}, along with a graphical representation of the model fitting performance on key compartments in Figure \ref{fig: training dataset fitting}. In this context, the Vaccinated and Recovered states have relatively large population proportions compared to the other compartments. Due to these relatively large proportions, the neural network assigns greater weight to aligning these two states, resulting in a closer fit to the corresponding observable data points. In the second and third panels of Figure \ref{fig: training dataset fitting}, the Exposed and Mild Infected compartments exhibit periodic patterns. According to \citet{Soukhovolsky_2021}, cyclical fluctuations with a period of approximately seven days may occur without obvious biological causes, potentially reflecting social effects. One hypothesis is that the so-called “weekend effect” influences these cycles, because individuals presenting initial disease symptoms often seek medical attention only after the weekend, at the start of the following week. Since this cyclical behavior does not affect the overall trajectory of disease transmission over longer time horizons, we consider the proposed neural network sufficiently robust to capture the primary trends in the Exposed and Mild Infected populations.

\subsubsection{Expenditure function parameters}
\label{subsec in fitting: hyperparameter definition}

In this section, we examine how the parameters for the pandemic-related cost function are selected. Upon formulating our objective function in Section \ref{subsec in setup: objective function formulation}, it is imperative to appropriately weigh each component of total government expenditure, as variations in these weight factors can impact the optimal vaccination strategy derived. Firstly, according to the report by Informed Decision\footnote{Australia Community Profile. Available online at https://profile.id.com.au/australia/employment-status?BMID=50.}, the working-age population in Victoria is 81.8\%, and the participation rate is 61.1\%, resulting in the labour proportion parameter $\psi$ being approximately 50\%. Subsequently, we estimate the other parameters in the cost function. By analyzing historical data from the Australian government, we calibrate the weights of each component to reflect the actual cost profile, thus enhancing the model's ability to apply to real-life scenarios.

According to the Parliament of Australia's report\footnote{Budget Review 2020-21. Available online at: https://parlinfo.aph.gov.au/parlInfo/search/display/display.w3p;query=Id \\
\%3A\%22library\%2Fprspub\%2F7622081\%22.}, expenditures relating to COVID-19 are mainly classified as economic response expenses, as well as costs associated with policy campaigns and healthcare. Within the economic sector, a significant portion of the government's expenditure was allocated to the JobKeeper program and the JobMaker Hiring Credit. Drawing on official records from the Australian Government, we compile total government expenditures in this category across the financial years 2019-20, 2020-21, and 2021-22 to assess the total economic assistance provided by the government in relation to the past pandemic.

Furthermore, we categorize the rest of government spending into policy-related costs and healthcare system costs based on the Australian Institute of Health and Welfare report\footnote{Health system spending on the response to COVID-19 in Australia 2019-20 to 2021-22. Available online at: https://www.aihw.gov.au/reports/health-welfare-expenditure/health-system-spending-on-the-response-to-covid-19/data}, which details government's expenditures in response to COVID-19 from 2019–20 to 2021–22. We filter the data to include only expenditures related to the national and state governments and assume that direct healthcare costs encompass spending in the categories of "public health", "private hospitals", and "public hospitals". Thus, all other expenditures are classified as policy-related expenses in our case study.

The analysis indicates that economic losses account for a substantial share of COVID-19-related costs in Australia, comprising approximately 75\%. In comparison, policy-related costs constitute 9\%, while healthcare system costs represent 16\%. Figure \ref{fig: two cost breakup pie} illustrates a financial year breakdown of policy expenditures, highlighting a significant increase in 2021–22 compared to the initial phase of the pandemic. This growth reflects the Australian government’s intensified efforts to expand vaccine administration beginning in mid-2021, resulting in substantial spending during the third financial year following the onset of COVID-19.

\begin{table}[htbp]
    \captionsetup{skip=5pt}
	\centering
	\begin{tabular}{lcc}
		\toprule
		Cost components  &  Parameters  & Values \\
		\midrule
             Vaccination cost & $c_1$ & 100  \\
             Quarantine cost  & $c_2$ & 20 \\
		Medicine cost  & $c_3$ & 50  \\
             Hospial cost & $c_4$ & 200 \\
             ICU cost  & $c_5$ & 1000 \\
             Economic cost  & $c_6$ & 100 \\
		\bottomrule
	\end{tabular}
        \caption{Parameters of the cost objective function for pandemic response.}
	\label{table: hyperparameter objective functn}
\end{table}

We base our parameter assumptions for the expenditure expression on the actual expenditures incurred by the Australian government during the COVID-19 pandemic, as summarized in Table \ref{table: hyperparameter objective functn}. In practice, payments for medical expenses vary among infected individuals and reflect the severity of their symptoms. Patients with severe symptoms are expected to incur higher medical costs. In the context of the recent Coronavirus pandemic, \citet{Kaier_2020} reports that costs per non-ventilated ICU day range between 924 and 1074, and \citet{Ohsfeldt_2021} demonstrates that treating patients in the ICU is nearly five times as expensive as treating other patients in a public hospital. Accordingly, we assume that the cost of regular hospital care is $200$ per patient per day, represented as $c_4 = 200$, while the expense for intensive care amounts to $1,000$ per patient per day, expressed as $c_5 = 1,000$. Furthermore, we compare our model's expenditure pattern with the government's actual expenditure structure. With the actual vaccination administration rate from the modelling period, we calculate the summation of vaccination administration costs and quarantine subsidies, $\mathbb{E} [ \int_o^T V_v(\alpha_t) + V_q(\alpha_t)  \textrm{d}t ] $, as described in Section \ref{subsec in setup: objective function formulation}, to represent the aggregate policy-related costs in our model. In addition, $\mathbb{E} [ \int_o^T V_l(\alpha_t) \textrm{d}t ] $ is used to represent the economic loss, and $\mathbb{E} [ \int_o^T V_h(\alpha_t) \textrm{d}t ] $ represents the direct expenditures related to the healthcare system. This approach ensures our cost model closely reflects the actual expenditure proportions for each category by the local government, as shown in Figure \ref{fig: two cost breakup pie}. Accordingly, the parameters are selected based on the specific expenditure profile for the particular modelling period of this case study.

\begin{figure}[htbp]
\centering
\begin{minipage}[t]{0.48\textwidth}
\centering
\input{Figure_codes/pie_actual}
\captionsetup{width=.9\textwidth} 
\end{minipage}
\begin{minipage}[t]{0.48\textwidth}
\centering
\input{Figure_codes/pie_model}
\captionsetup{width=.9\textwidth} 
\end{minipage}
\caption{Comparison of actual (left) and modelled (right) government expenditures by component across vaccination strategies.}
\label{fig: two cost breakup pie}
\end{figure}
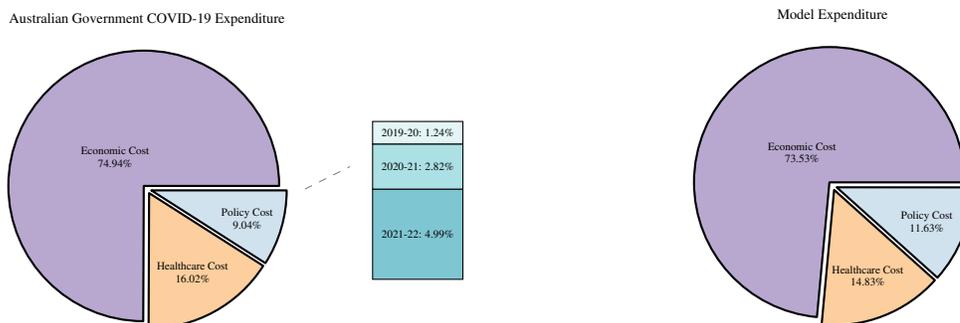

\subsection{Control results}
\subsubsection{The optimal vaccination strategy}
\label{subsec in control solving: optimal vacc in vic}

After determining the parameters of the economic epidemiological model, we apply the neural network detailed in Section \ref{subsec in control solving: nn to solve high dimension problem} to derive the numerical solution for our case study. For this case study, the subnetwork architecture consists of an input layer of initial compartmental proportions, three hidden layers, and an output layer illustrating the vaccination administration rate at each time point. In this setup, the tanh function is chosen as the activation function for the input and all hidden layers, and the number of iterations is set at 10,000. Given that neural network outputs depend on the seed settings, we perform this process five times using different seed values and compute the average result to obtain the final output.

\begin{figure}[htbp]
\centering
\begin{minipage}[t]{0.48\textwidth}
\centering
\includegraphics[width=0.9\textwidth]{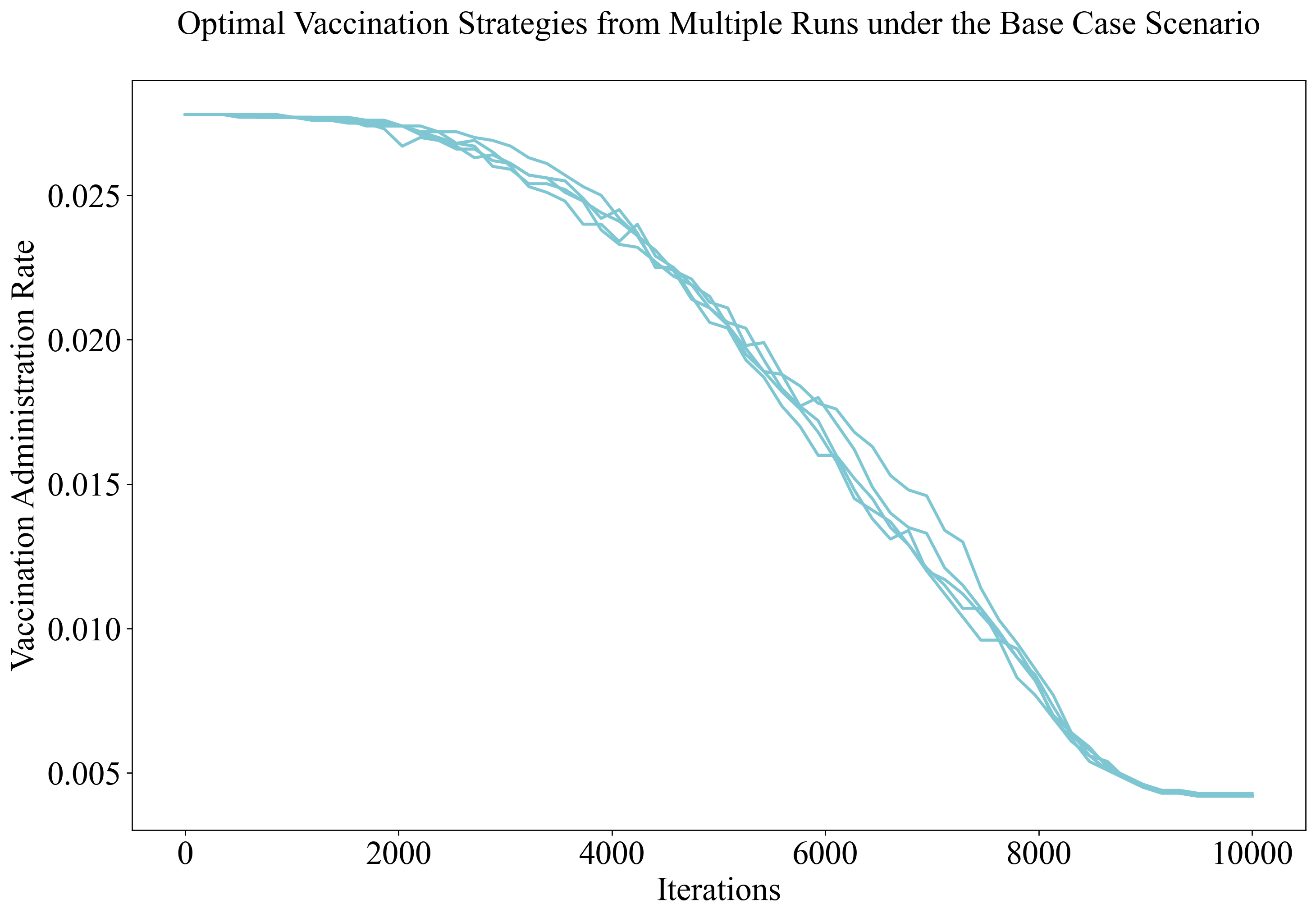}
\captionsetup{width=.9\textwidth} 
\caption{Optimal vaccination administration rates from multiple runs in the base case scenario.}
\label{fig: 5 runs opt vacc}
\end{minipage}
\begin{minipage}[t]{0.48\textwidth}
\centering
\includegraphics[width=0.9\textwidth]{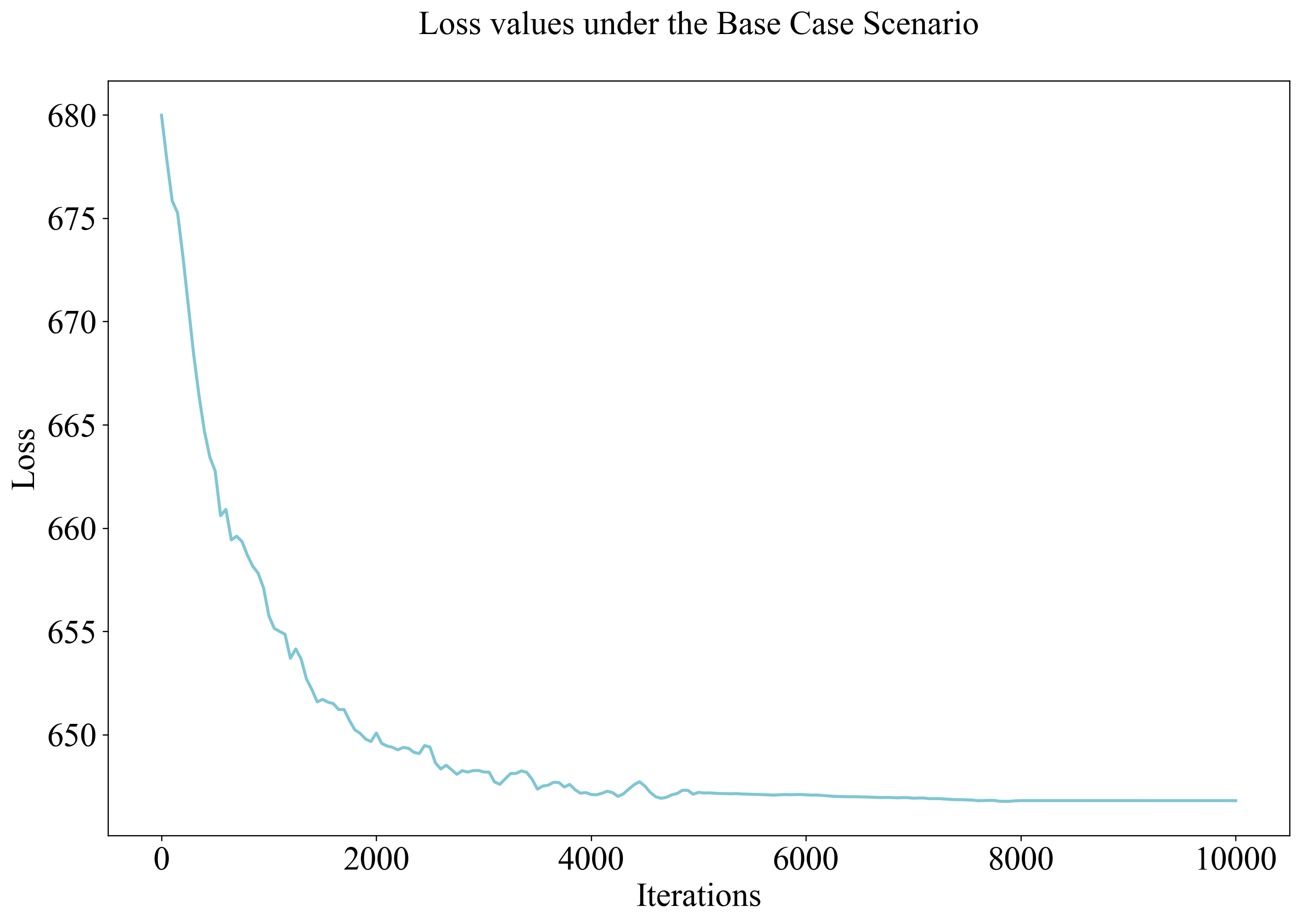}
\captionsetup{width=.9\textwidth} 
\caption{Average loss values from multiple runs in the base case scenario.}
\label{fig: 5 runs avg loss}
\end{minipage}
\end{figure}

As a first step, we evaluate the stability of the numerical results before considering the rationale for the optimal solution. Given the nature of neural networks, Figure \ref{fig: 5 runs opt vacc} displays the optimal control solutions across five runs, demonstrating consistent results that support the stability of our numerical outcomes. In addition, we also examine changes in the loss function to verify convergence. As illustrated in Figure \ref{fig: 5 runs avg loss}, the loss function value decreases as the number of iterations increases and it eventually converges. Therefore, we consider the numerical solution to our control problem to be reliable.

Under our baseline setting, the numerical solution corresponding to the optimal vaccination administration rate is shown as the cyan line in Figure \ref{fig: vacc strategies base case}. According to this optimal vaccination strategy, the government should initially implement a large number of vaccinations. Subsequently, the dose injection rate should gradually decrease as the virus becomes controlled. As a result, the proposed strategy can effectively reduce government expenditures by managing a gradual decrease in vaccination administration rates.

\subsubsection{Alternative strategies}
\label{subsec: comparison of strategies}

Once we have determined the optimal vaccination strategy within our control framework, we conduct an analysis to compare our strategy with three alternative scenarios. The first comparison is with the government's historical vaccination rate, depicted as the purple line in Figure \ref{fig: vacc strategies base case} over the modelling period. Second, we consider a constant vaccination rate, which is equal to the average of actual vaccination rates over the period, regardless of the disease's progression. This is represented by the orange line, which shows the fixed vaccination rate throughout the modelling period. The final scenario examines the case where there is no vaccination, capturing the natural progression of the virus without intervention, corresponding to the red line at the bottom of the figure.

\begin{figure}[htbp]
\centering
\begin{minipage}[t]{0.48\textwidth}
\centering
\includegraphics[width=1.02\textwidth]{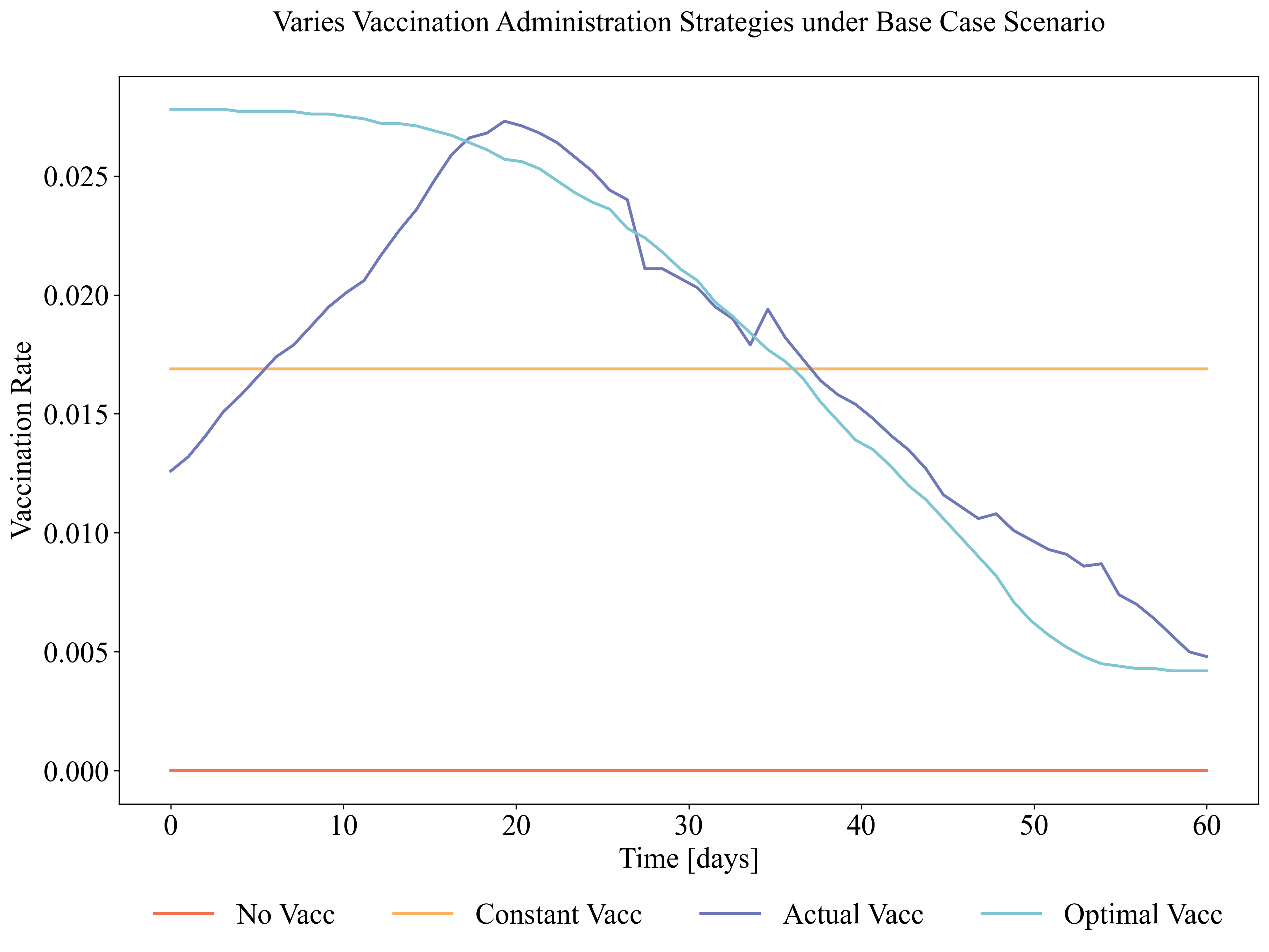}
\captionsetup{width=.9\textwidth} 
\caption{Comparison of optimal vaccination administration rates with alternative strategies.}
\label{fig: vacc strategies base case}
\end{minipage}
\begin{minipage}[t]{0.48\textwidth}
\centering
\input{Figure_codes/base_cost_draw}
\captionsetup{width=.9\textwidth} 
\caption{Allocation of government expenditures by component under various vaccination strategies.}
\label{fig: cost comparison base case}
\end{minipage}
\end{figure}

We begin by comparing the optimal vaccination administration rate with the actual rate implemented by the government. Our analysis shows a divergence between the government's approach and the proposed optimal strategy during the first half of the modelling period. According to our optimal strategy, the government should have launched vaccination programs at a high rate. This finding aligns with \citet{Acuna_2021}, which advocates for a high initial vaccination rate to achieve broad coverage and curb the virus's spread early. Notably, the actual rate implemented by the government begins slowly and then increases over time, differing from the recommended strategy. This initial slow rollout is likely due to uncertainties surrounding the newly introduced vaccine, leading us to consider vaccination hesitancy, discussed further in Section \ref{subsec in discussion: hesitancy range}. Furthermore, our analysis reveals a close alignment between the recommended vaccination administration rates and the Australian government's strategy during the latter half of the modelling period. As a result, the Victorian government effectively determines the strategy for controlling virus spread during the later stages of the pandemic, when they have gained sufficient experience to decide the best course of action to be taken.

\subsubsection{Population dynamics comparsion}
\label{subsubsec in vacc strategy comparsion: popn comparsion of strategies under base case}

As a subsequent step, we examine the population dynamics across different compartments over time in response to various vaccination strategies. Our focus is on the states $ \text{E} $, $ \text{I}_1 $, $ \text{I}_2 $, and $ \text{I}_3 $, analyzing how different vaccination administration rates influence both the number of individuals exposed to the virus and the progression of infections of varying severities over time. This analysis aims to investigate the differential impact of these strategies on virus progression through the evolution of state populations. By contrasting the development pathways of the virus under diverse vaccination regimes, we seek to gain insight into the effectiveness and limitations of each strategy in mitigating and controlling the spread of the virus.

\begin{figure}[htbp]
    \centering
    \includegraphics[width=\textwidth]{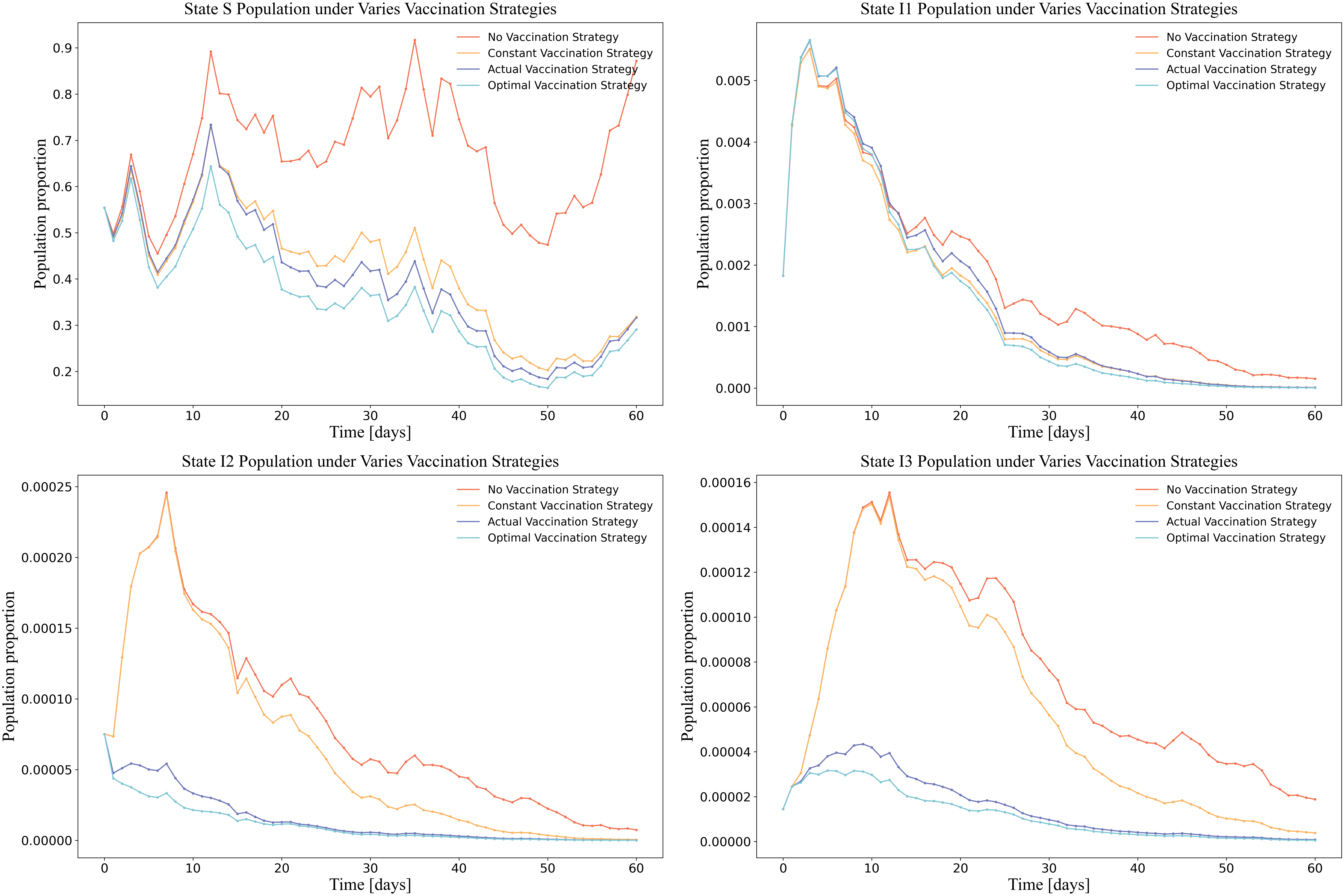}
    \vspace{0.2\baselineskip}
    \caption{Comparison of policy impacts on exposed, mildly infected, hospitalized, and ICU populations.}
    \label{fig: popn comparsion base case}
\end{figure}

When examining the number of individuals exposed to the virus, the top left graph in Figure \ref{fig: popn comparsion base case} illustrates that, without vaccination, a substantial portion of the population remains at risk. This is evidenced by the red line appearing at the top of the data series. In contrast, the other three vaccination strategies lead to a significant reduction in the number of exposed individuals over time. Notably, the optimal vaccination strategy, which features a high initial injection rate, achieves a more rapid decline in exposure than the other two strategies. This is shown by the cyan line consistently being the lowest.

Following an analysis of the exposed population, we move to the changes within the infected compartments, specifically those who are mildly infected, those requiring hospitalization, and those with severe symptoms that require intensive care. The last three graphs in Figure \ref{fig: popn comparsion base case} illustrate these changes. As in the exposed population graph, the scenario without vaccination experiences significant increases in the number of infected individuals across all symptom levels compared to scenarios that utilize vaccination. Comparing all three strategies with vaccination involved, the constant vaccination administration rate proves to be less effective in reducing the number of individuals infected across all severity levels, especially patients in hospitals or ICUs. It is evident from the graphs in the second row that the yellow lines are substantially higher than the purple and cyan lines. These results suggest that maintaining a constant vaccination administration rate throughout the period would not be effective in preventing the spread of the disease.

Furthermore, in the case of the actual vaccination strategy implemented by the government, it is evident that the number of infected patients significantly decreases when compared to the no and constant vaccination administration approaches. This decrease highlights the effectiveness of the past Victorian government's vaccination campaign, which was launched with an initial upswing followed by a gradual decline during the COVID-19 pandemic. Such an approach played a significant role in controlling the spread of the disease and alleviating societal pressures.

In addition, the optimal vaccination strategy developed in Section \ref{subsec in control solving: optimal vacc in vic} demonstrates further potential to reduce the number of infections more rapidly and sustain the lowest proportion of infected individuals over a longer duration. Consequently, it is regarded as the most effective strategy evaluated in this case study. While the government's existing initiatives are commendable, it is considered that further enhancements can be made to control the disease through the proposed optimal strategy under our framework.

\subsubsection{Expenditures comparsion}
\label{subsubsec in vacc strategy comparsion: cost comparsion of strategies under base case}

After analyzing changes in population dynamics, we turn to the costs associated with various government strategies. Table \ref{table: cost breakdown base case} presents a comparison of total expenditures under the four scenarios, with the corresponding graphical representation in Figure \ref{fig: cost comparison base case}. In this case, we break down total government expenditures into categorized sectors and compare expenditures within each group separately. The results are displayed in Figure \ref{fig: cost breakup comparison base case}. For clarity in this analysis, we consolidate the vaccination cost and quarantine subsidy expenses as outlined in Section \ref{subsec in fitting: hyperparameter definition}, referring to them collectively as the aggregate policy cost.

\begin{figure}[htbp]
    \centering
    \includegraphics[width=\textwidth]{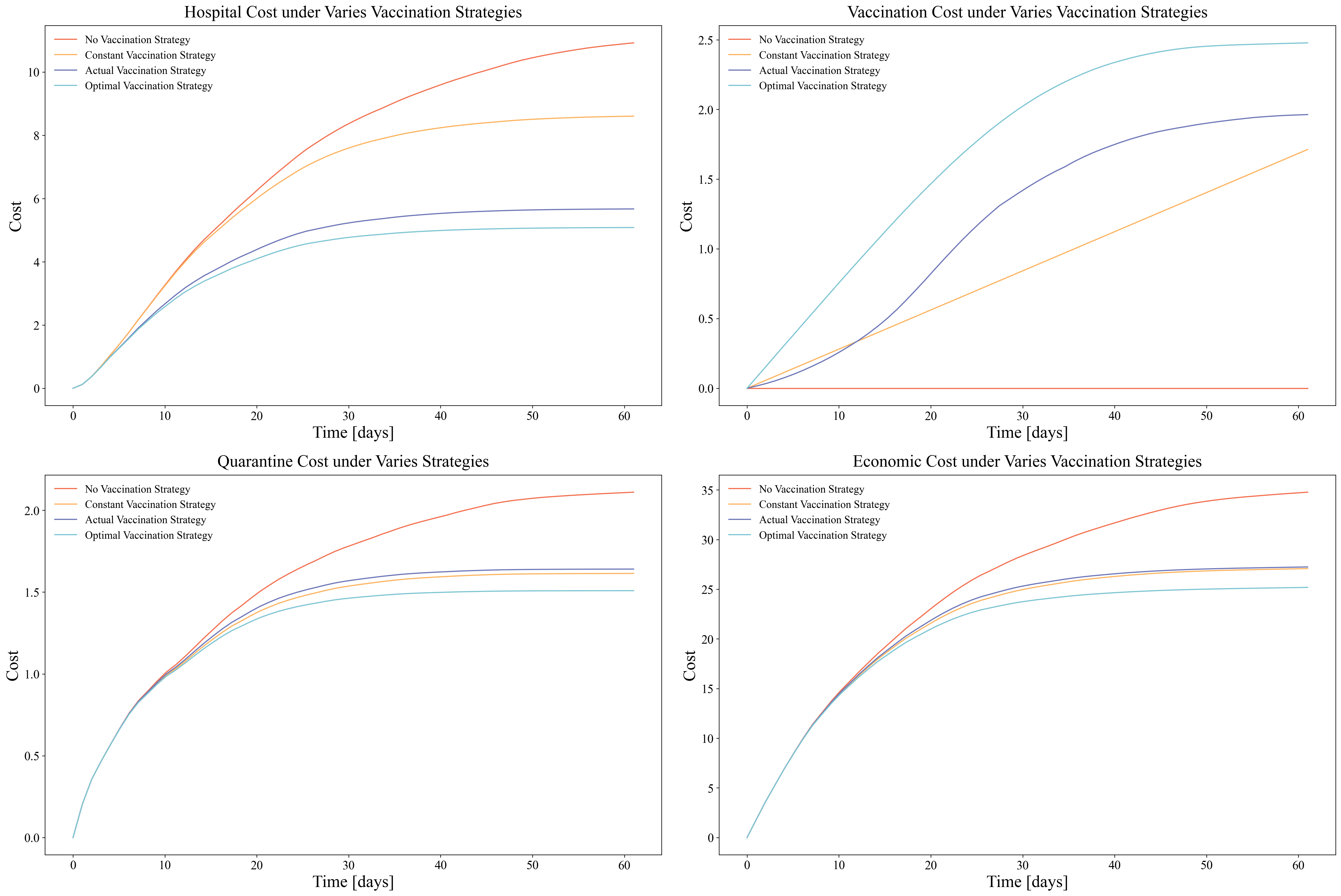}
    \vspace{0.2\baselineskip}
    \caption{Comparison of policy expenditures on cost components during disease progression.}
    \label{fig: cost breakup comparison base case}
\end{figure}

\begin{table}[htbp]
    \captionsetup{skip=5pt}
	\centering
	\begin{tabular}{ccccc}
		\toprule
		Strategy & Policy cost & Healthcare cost & Economic cost & \textbf{Total} \\ 
		\midrule
		No vaccination & 2.11 & 10.95 & 34.89 & \textbf{47.89}    \\
            \midrule
		Constant vaccination & 3.33 & 8.61 & 27.09 & \textbf{39.03}  \\
            \midrule
		  Actual vaccination  & 3.60 & 5.67 & 27.26 & \textbf{36.53}   \\
            \midrule
		Optimal vaccination & 3.99 & 5.08 & 25.20 &  \textbf{34.27}  \\
		\bottomrule
	\end{tabular}
 	\caption{Comparison of government expenditures under different strategies.}
	\label{table: cost breakdown base case}
\end{table}

We begin by considering the scenario with no vaccination during the modelling period. In this case, overall expenditures significantly exceed the other three vaccination strategies, as reflected in the highest red bar on the left in Figure \ref{fig: cost comparison base case}. Although no costs are incurred for vaccination administration in the first scenario, increased exposure to the disease results in higher quarantine costs, thereby increasing the aggregate policy costs. In terms of healthcare system costs, the proportion of infected individuals requiring hospitalization and ICU is considerably higher in the absence of vaccination, which drives up healthcare expenditures. As a result of the widespread transmission of the virus, substantial economic losses have been incurred. Consequently, the total governmental expenditure reaches its highest level in the absence of vaccines.

Following this, we examine the constant vaccination approach and the actual rollout of the government. Under the government's actual strategy, policy costs are higher than they would be under the uniform vaccination approach. This occurs because the government significantly increases vaccination administration rates at the beginning of the period, which leads to a rise in policy costs. Meanwhile, healthcare expenditures under the actual vaccination strategy are lower compared to a constant vaccination approach, as fewer individuals require hospitalization or intensive care. In terms of economic costs, the actual vaccination strategy is intermediate between the zero vaccination strategy and the optimal strategy. The results of this study suggest that, under pressure from social and healthcare system capacities, the Victorian government effectively reduces the number of severe cases through vaccination, but at the expense of vaccination deployment costs and reduced economic productivity. This rationale prompts our study on optimal strategy analysis, aiming to reduce the overall governmental cost while simultaneously lowering the burden on the healthcare system.

Based on the analysis of all vaccination plans, we find that the optimal vaccination strategy achieves the lowest overall cost, as indicated by the cyan bar being the lowest in Figure \ref{fig: cost comparison base case}. However, regarding policy costs, this optimal approach incurs the highest vaccination expense compared to other scenarios, as depicted by the highest curve of the cyan line in the top right graph of Figure \ref{fig: cost breakup comparison base case}. The reason for this substantial vaccination cost lies in the high injection rate recommended by the optimal strategy. Nevertheless, excluding the rollout expense, the government could achieve significant reductions in other cost areas, with the cyan lines being the lowest in the remaining graphs of Figure \ref{fig: cost breakup comparison base case}. Specifically, the optimal vaccination strategy reduces healthcare costs by more than 50\% compared to the zero vaccination strategy. The results align with the previously observed population data for state $\text{I}_1$, $\text{I}_2$, and $\text{I}_3$ in Figure \ref{fig: popn comparsion base case}, suggesting that by adopting the appropriate strategy, the government could also alleviate the hospital burden. Moreover, this optimal vaccination strategy minimizes economic losses. Thus, by optimally regulating the vaccination administration rate, the government can minimize total expenditures while effectively reducing both social and hospital pressures by reducing the number of infections across all levels of illnesses. In this manner, the proposed strategy is expected to balance costs and benefits, resulting in a viable compromise during the pandemic.

\section{Sensitivity Analysis}
\label{sec: discussion}

In this section, we examine how uncertainty affects the optimal vaccination policy by utilizing parameter settings that differ from the policymaker's previous beliefs. Based on the works of \citet{Knight_1921} and \citet{Arrow_1951}, we first differentiate between risk and uncertainty. When model parameters are known, risk represents the range of possible outcomes. Conversely, uncertainty arises when model parameters are unknown or potentially misspecified. Within our framework, we introduce risk by allowing the disease to spread in a non-deterministic manner. Furthermore, uncertainty emerges when the true parameters governing the transmission and severity of a disease are obscured. In light of these uncertainties, the policymakers should consider the implications of different parameter settings when making decisions. Consequently, we conduct a range of sensitivity tests to evaluate the key parameters' impacts embedded in the framework.


\subsection{Uncertainty inside the compartmental model}
\subsubsection{Effect of noise intensities}
\label{subsec in discussion: noise intensity}

Based on recent studies by \citet{Adak_2021}, \citet{Barnett_2023}, and \citet{Gunasekaran_2023}, it is essential to examine the impact of risk in the stochastic system by modifying the set of noise intensity parameters. Following the concept of \citet{La_2024}, to formulate vaccination strategies under varying external risks, policymakers should evaluate the effect of differing noise levels on the proposed model design. In relation to this numerical case study, we examine three additional sets of noise intensity values compared to our baseline. Using the values from Table \ref{table: fitted para assumpn} as a reference level, we multiply the set of calibrated noise intensity levels $Z=\{ \sigma_1, \sigma_2, \sigma_3, \sigma_4, \sigma_5, \sigma_6, \sigma_7, \sigma_8 \}$ by factors of 0.1, 0.5, and 2. These correspond to scenarios with tiny, small, and large noise levels, respectively.

\begin{figure}[htbp]
\centering
\begin{minipage}[t]{0.48\textwidth}
\centering
\includegraphics[width=1\textwidth]{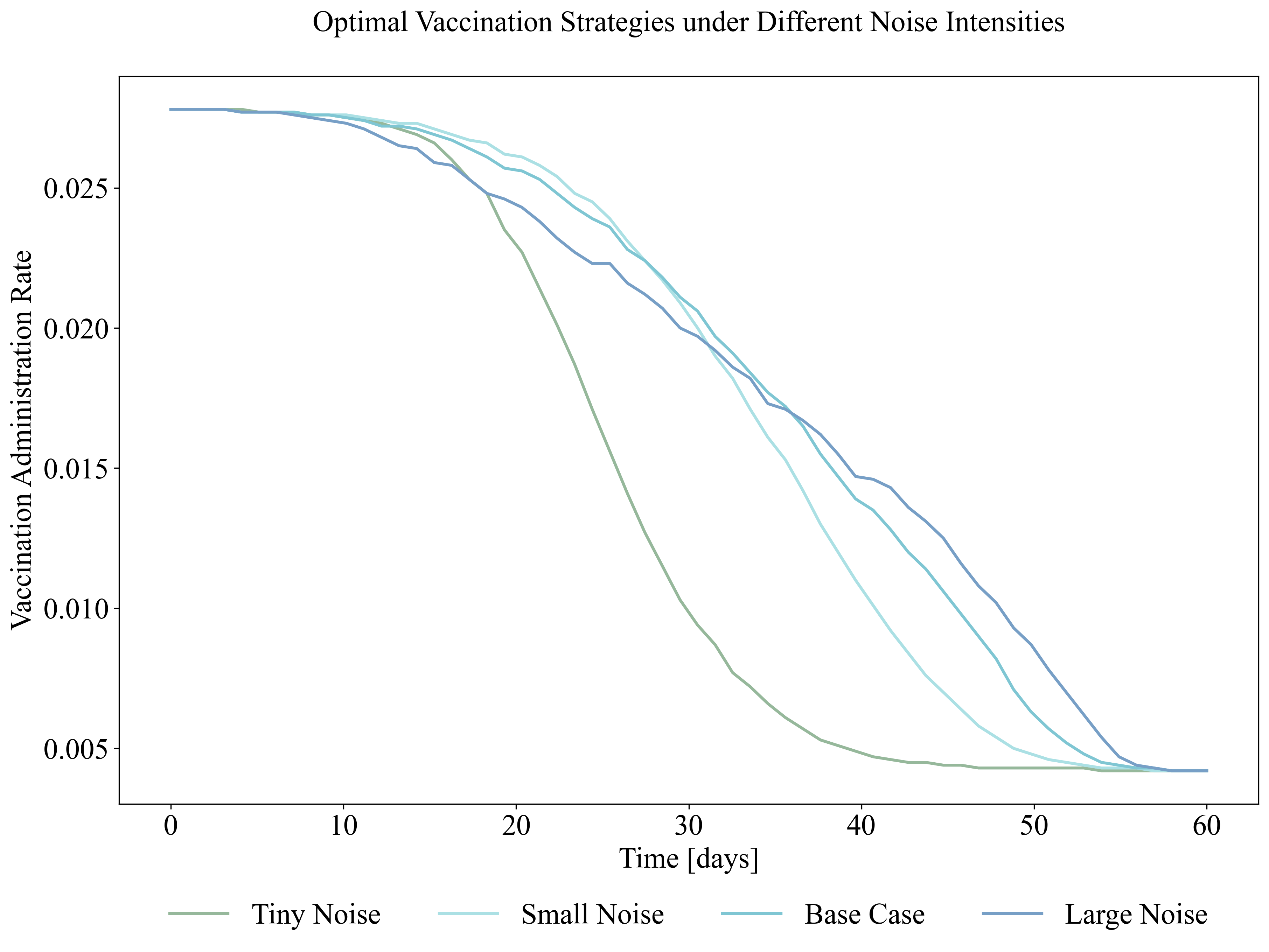}
\captionsetup{width=.9\textwidth} 
\caption{Comparison of optimal vaccination strategies under base case and different noise intensity levels.}
\label{fig: optimal vacc of different noise}
\end{minipage}
\begin{minipage}[t]{0.48\textwidth}
\centering
\input{Figure_codes/noise_cost}
\captionsetup{width=.9\textwidth} 
\caption{Comparison of expenditures on cost components under base case and different noise intensity levels.}
\label{fig: cost of different noise}
\end{minipage}
\end{figure}

It is crucial to examine the evolution of optimal vaccination strategies under different levels of noise intensity. Figure \ref{fig: optimal vacc of different noise} illustrates how vaccination rollout paths change with varying noise levels, while Figure \ref{fig: cost of different noise} presents the corresponding expenditures for each scenario. The results indicate that the reduction in vaccination administration rates midway through the period is less pronounced when noise intensity is high, as shown by the dark blue line. This finding demonstrates that heightened uncertainty compels risk-averse policymakers to adopt a more conservative vaccination strategy, sustaining a high vaccination administration rate over a prolonged period. Consequently, this approach incurs higher costs than the other scenarios, with the dark blue bar in Figure \ref{fig: cost of different noise} representing the highest total expense. In contrast, under conditions of reduced stochastic noise, the proposed vaccination strategy, depicted by the green line, declines more rapidly over time, ultimately lowering aggregate expenditure levels.

We then break down the aggregate costs into components and analyze them in detail. First, as the bottom bar in Figure \ref{fig: cost of different noise} appears to be trending upward with increasing noise intensity, this suggests that the policy expenditure increases with an increase in noise intensity. At the same time, other costs are likely to increase with increasing levels of stochastic noise as well. Observations from this indicate that environmental fluctuations have a significant impact on both the choice of vaccination rollout plans and the overall expenditures of the government. Consequently, policymakers should exercise caution when determining the noise intensity parameters for compartmental models pertaining to the pandemic.

\subsubsection{Effect of infection rates}
\label{subsec in discussion: infection rate}

As the infection rate parameters within the compartmental model, represented by $\beta_1$, $\beta_2$ and $\beta_3$, are the key driving factors of the spread of infectious diseases, we modify these parameters in our configurations and determine the optimal vaccination strategies based on the varying infection rates. To account for more severe diseases, we increase all $\beta$s by 50\%, setting them to 1.5 times the values in Table \ref{table: fitted para assumpn}. Conversely, for a less severe virus, all $\beta$ parameters are halved to account for a mild infection. Furthermore, we reduce the infection rates by 90\% in an extreme scenario, corresponding to a situation with a minimal number of infections. The effectiveness of our optimal policy is evaluated by analyzing the cost savings that the proposed strategy achieved over the scenario in which the government maintains a constant vaccination administration rate, irrespective of any changes in the epidemic over time.

\begin{figure}[htbp]
\centering
\begin{minipage}[t]{0.48\textwidth}
\centering
\input{Figure_codes/beta_cost}
\captionsetup{width=.9\textwidth} 
\caption{Comparison of expenditures under base case and different infection rate levels.}
\label{fig: optimal cost of different infection}
\end{minipage}
\begin{minipage}[t]{0.48\textwidth}
\centering
\input{Figure_codes/beta_saving}
\captionsetup{width=.9\textwidth} 
\caption{Comparison of cost savings between optimal and constant vaccination strategies under base case and different infection rates.}
\label{fig: cost of different infection}
\end{minipage}
\end{figure}

To begin with, we examine how optimal government expenditure is affected by different infection rates. As shown in Figure \ref{fig: optimal cost of different infection}, when the infection rate increases by 50\%, overall expenditure nearly doubles. The fourth bar illustrates a significant increase in spending across all components of the cost function compared to the third bar, which represents our base case optimal expenditure. Furthermore, when the infection rate is halved, or reduced to only 10\% of the base case setting, overall pandemic-related costs decrease significantly, with the economic expenditure sector being the most impacted.

Additionally, following a similar approach to \citet{Acemoglu_2021}, we analyze the government savings achievable through the implementation of the optimal control approach compared to the constant vaccination strategy. Our findings indicate that the government can save more money by applying the proposed framework to obtain the optimal administration strategy during a severe disease outbreak compared to the case of a mild virus. The dark blue bar, corresponding to the high infection scenario in Figure \ref{fig: cost of different infection}, illustrates significant savings in the healthcare and economic categories. Conversely, using this control strategy for a less severe disease tends to save a small percentage of funds compared to a uniform injection strategy over time. Therefore, during severe pandemics, it is imperative that the government prioritizes control strategies and makes rapid adjustments to plans in response to changing conditions, as this approach would result in significant savings for the country.

\subsection{Uncertainty inside the expenditure function}
\subsubsection{Effect of vaccination costs}
\label{subsec in discussion: vacc cost}

Subsequently, we examine the uncertainty associated with the parameters of the expenditure function. To analyze the sensitivity of cost parameters related to various components, we employ an approach similar to that described in \citet{Pike_2014}. With respect to vaccination policy costs, we consider three alternative levels. The first two scenarios explore reductions in current policy spending, achieved by multiplying the vaccination cost parameter $ c_1 $ by factors of 0.1 and 0.5, respectively. The third scenario involves an increase in current policy spending by doubling the value of $ c_1 $.

Taking into account different levels of policy costs, we analyze how the recommended strategy evolves over time and find that government spending on vaccination plays a pivotal role in shaping the optimal vaccination rollout. Figure \ref{fig: optimal vacc of different vacc cost} demonstrates that the optimal vaccination rate declines at a faster pace as government expenditures for vaccination implementation increase. The dark blue line, which corresponds to higher vaccination costs, remains consistently below the lines representing lower-cost scenarios. In cases where vaccination is more expensive, governments are likely to slow down vaccination efforts more rapidly. This deceleration allows resources to be redirected toward post-pandemic recovery and reconstruction initiatives. On the other hand, when the cost of vaccination implementation has only a marginal impact on overall expenditures, governments can pursue a more gradual reduction in the vaccination rate.

\begin{figure}[htbp]
\centering
\begin{minipage}[t]{0.48\textwidth}
\centering
\includegraphics[width=\textwidth]{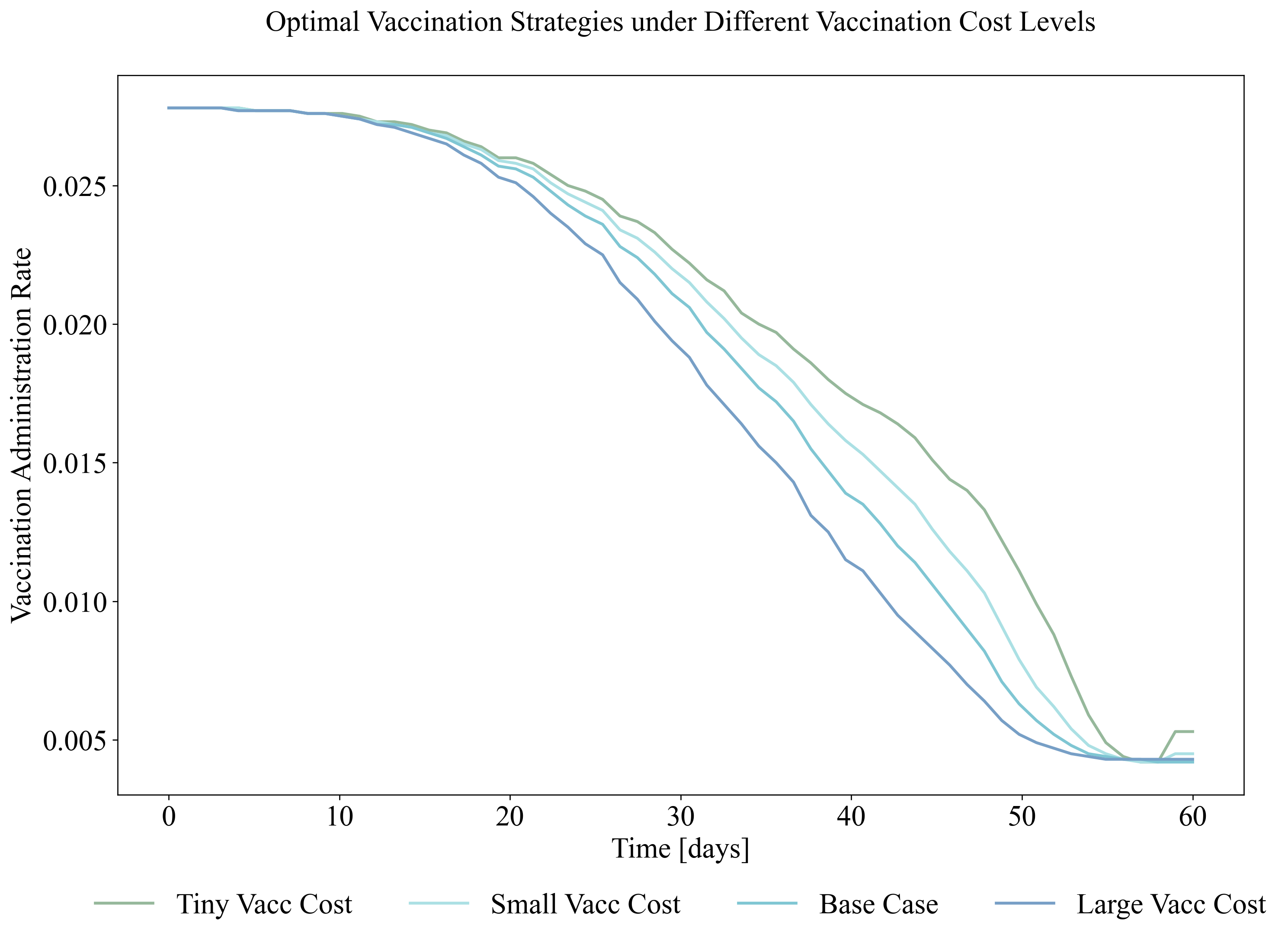}
\captionsetup{width=.9\textwidth} 
\caption{Comparison of optimal vaccination strategies under base case and different vaccination cost levels.}
\label{fig: optimal vacc of different vacc cost}
\end{minipage}
\begin{minipage}[t]{0.48\textwidth}
\centering
\includegraphics[width=\textwidth]{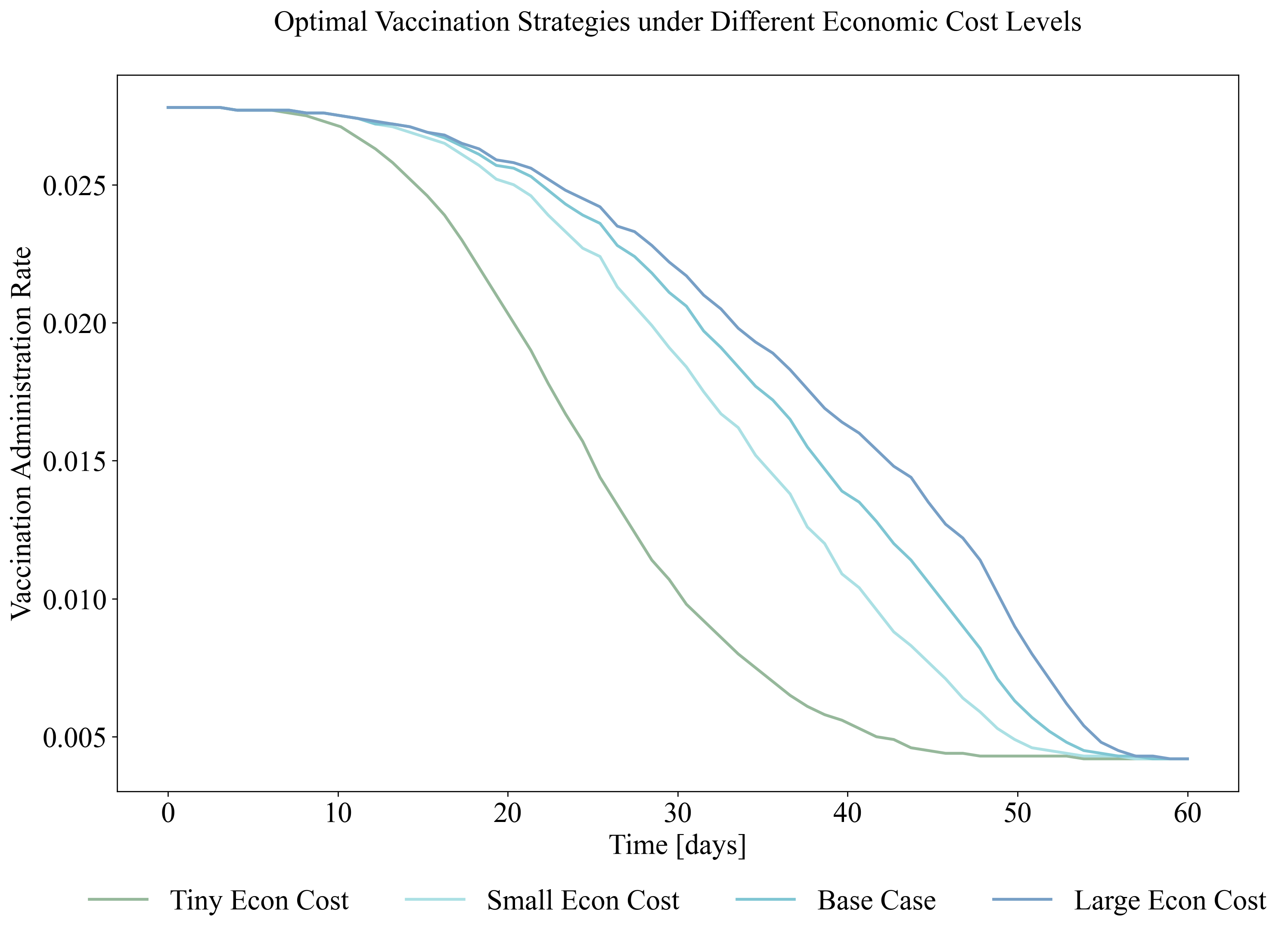}
\captionsetup{width=.9\textwidth} 
\caption{Comparison of optimal vaccination strategies under base case and different economic cost levels.}
\label{fig: cost of different econ cost}
\end{minipage}
\end{figure}


\subsubsection{Effect of economic costs}
\label{subsec in discussion: econ cost}

Considering that the government's expenditures on financial support constitute the most significant component of overall expenses, we also examine how variations in economic costs impact the outcomes of our optimal vaccination strategy. According to \citet{Jackson_2020}, the Australian Government provides substantial financial support during the COVID-19 pandemic. Similar to Section \ref{subsec in discussion: vacc cost}, we examine three additional values of the economic cost parameter $c_6$ relative to our baseline setting, using multiplying factors of 0.1, 0.5, and 2. These correspond to tiny, small, and large economic costs compared to our baseline setting.

Accordingly, the optimal vaccination time paths for each scenario are depicted in Figure \ref{fig: cost of different econ cost}. In situations where the government provides substantial support during a pandemic, the optimal vaccination rate decreases at a slower pace, as illustrated by the dark blue line. Conversely, when government support is significantly diminished, represented by the green line, the vaccination rate declines more rapidly, leading to a reduction in overall spending. This outcome aligns with expectations, given that policymakers are likely to aim to curb disease spread more aggressively when economic losses are substantial.

\subsection{Other scenarios}
\subsubsection{Varies vaccination hesitancy}
\label{subsec in discussion: hesitancy range}

This section examines the impact of changes in vaccination hesitancy on optimal vaccination strategies within our framework. As mentioned in Section \ref{subsec: comparison of strategies}, a crucial factor influencing the optimal vaccination administration rates is the maximum feasible rollout rate achievable by the government. This is connected to the concept of vaccination hesitancy, which, according to \citet{Macdonald_2015}, is defined by the World Health Organization (WHO) as a “delay in acceptance or refusal of safe vaccines despite the availability of vaccination services”. To establish effective search grids for the control variable, it is important to consider this concept. In our baseline analysis, the control variable's range is determined by the maximum and minimum daily vaccination administration rates observed during the modelling period. However, individual preferences regarding vaccination can change rapidly depending on the severity of the viral infection. Thus, it is more common for vaccination acceptance rates to fluctuate over time rather than remain constant. Generally, acceptance rates are initially low at the outset of a vaccine campaign, but as the public becomes more aware of the vaccine's efficacy, hesitancy tends to decrease over time. Furthermore, \citet{Edwards_2021} suggests that the demand for vaccination tends to increase following a pandemic outbreak, with individuals more eager to receive vaccinations after the emergence of a severe outbreak.

\begin{figure}[htbp]
\centering
\begin{minipage}[t]{0.48\textwidth}
\centering
\includegraphics[width=1.08\textwidth]{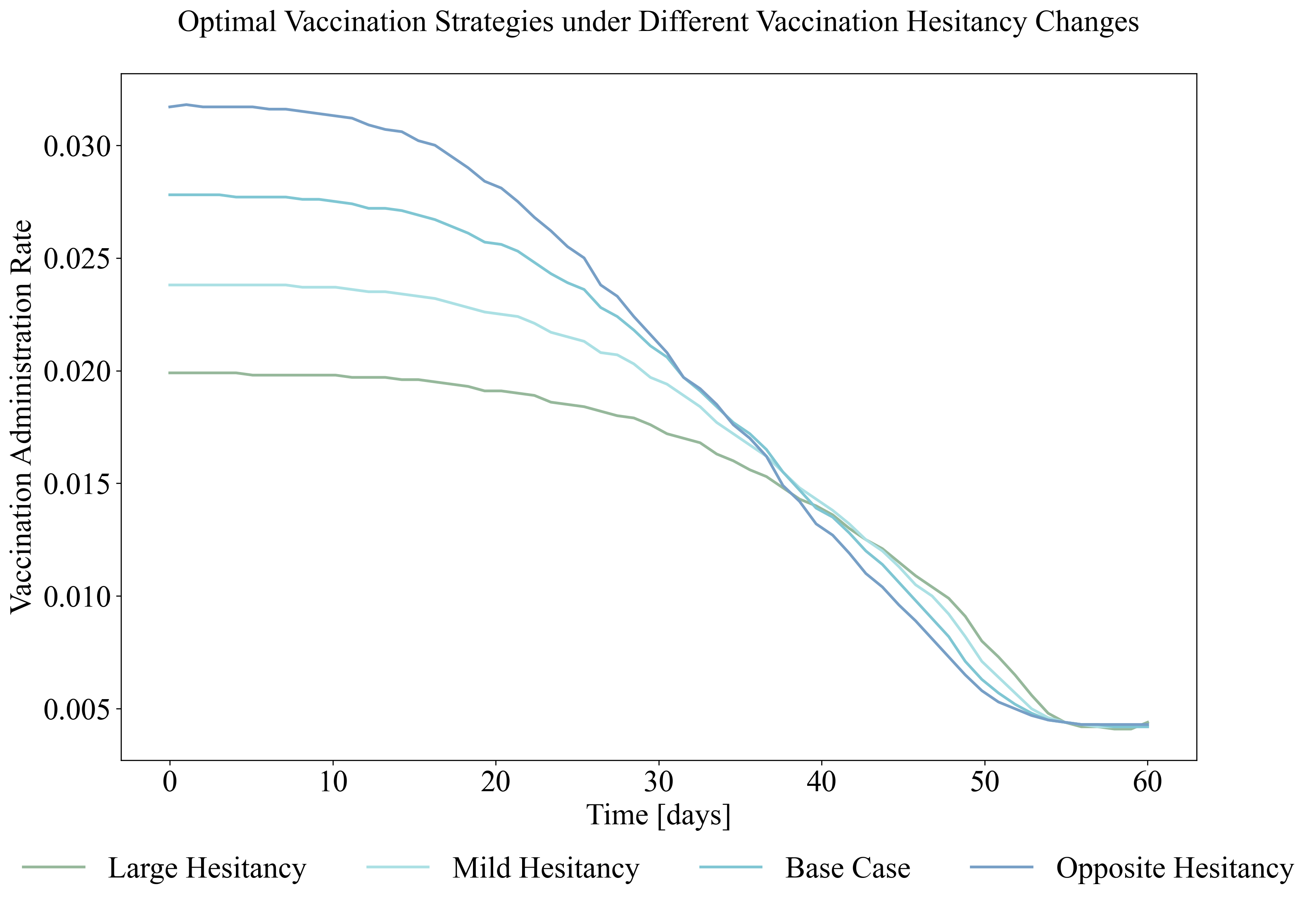}
\captionsetup{width=.9\textwidth} 
\caption{Comparison of optimal vaccination strategies under base case and different vaccination hesitancy levels.}
\label{fig: optimal vacc of different hesitancy}
\end{minipage}
\begin{minipage}[t]{0.48\textwidth}
\centering
\input{Figure_codes/range_saving}
\captionsetup{width=.9\textwidth} 
\caption{Comparison of cost savings between optimal and constant vaccination strategies under base case and different vaccination hesitancy levels.}
\label{fig: cost of different hesitancy}
\end{minipage}
\end{figure}

According to the proposed strategy obtained in Section \ref{subsec in control solving: optimal vacc in vic}, the optimal vaccination administration rate is highest at the beginning of the period. However, achieving such an elevated initial rate may be impractical due to significant hesitancy at the outset. To address this challenge, we consider the temporal evolution of attitudes toward vaccination. \citet{Biddle_2021} finds a 10\% reduction in hesitancy during the COVID-19 pandemic between August 2020 and January 2021. In light of this, we examine scenarios with hesitancy ratios of 15\% and 30\%. These adjustments reduce the upper limit of our control variable to 70\% and 85\% of the baseline value. Moreover, we examine the scenario of reduced vaccination hesitancy, which may be linked to the increased popularity of vaccination in post-outbreak periods. This change increases the maximum injection rate by a factor of 1.15 and results in an expanded control range.

Changes in individual attitudes toward vaccination impact the optimal vaccination policy in our case study, as illustrated in Figure \ref{fig: optimal vacc of different hesitancy}. In scenarios of high vaccination hesitancy, the green line indicates a compression in the range of the control variable and a decrease in the maximum possible injection rate. Furthermore, we also analyze the benefits of implementing an optimal vaccination policy by comparing the aggregate expenditure level achieved through the control framework to those under the constant vaccination policy. In this case, the government realizes significant savings in scenarios with reduced vaccination hesitancy, as shown in Figure \ref{fig: cost of different hesitancy}. Accordingly, from the perspective of the government, promoting the vaccine campaign and emphasizing its benefits in the early stages of vaccination is crucial for enhancing the effectiveness of the vaccination strategy. Increasing positive public perceptions of vaccination fosters high levels of acceptance, thereby significantly reducing aggregate expenditures.

\subsubsection{Various stages of time}
\label{subsec in discussion: initial stages}

\begin{figure}[htbp]
    \centering
    \includegraphics[width=\textwidth]{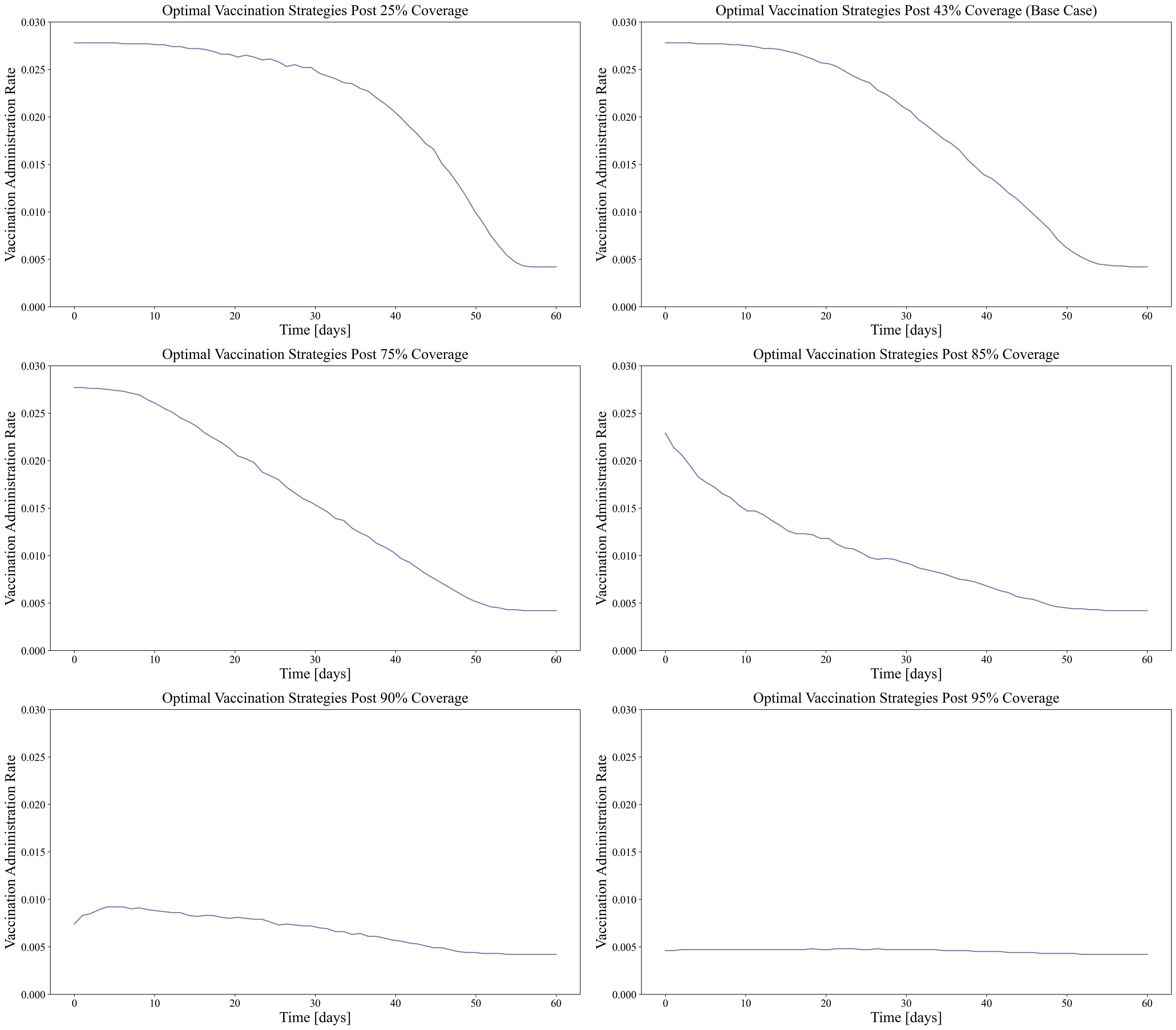}
    \vspace{0.2\baselineskip}
    \caption{Comparison of optimal vaccination strategies across phases of vaccination coverage.}
    \label{fig: optimal vacc of different initial stages}
\end{figure}

In practice, the government's vaccination rollout plan is expected to evolve over time as external environmental conditions change. Social planners should continuously update the vaccination administration campaign to ensure alignment with the latest situation of the virus. In this section, we examine the optimal vaccination plan with different initial vaccinated population proportions, as compared to Table \ref{table: initial state}, relating this to different phases of the vaccination rollout.

In the case of Victoria, the government develops an overall vaccination strategy\footnote{Australia’s COVID-19 Vaccine Rollout. Available online at: https://www.anao.gov.au/work/performance-audit/australia-covid-19-vaccine-rollout} prior to the administration of the first dose, which is published in late 2020. Despite this, the vaccination plan experiences slight delays due to supply-side issues, with the first dose being administered in February 2021. During the initial months, the actual rollout rate of the first dose is relatively low, likely due to the healthcare system not being fully operational during that period. Following the maturation of the system and increased public acceptance of vaccination, the rate of administration rises substantially. In the intermediate phase of first-dose distribution, the Australian government prioritizes promoting the administration of the second dose. As a result of the healthcare system's improved efficiency and experience, the coverage rate for the second dose increases at a faster pace than that of the first dose, as illustrated in Figure \ref{fig: vic vaccination rollout}.

Considering our aim to provide policy recommendations to the government over various time periods, we conduct an analysis of our model across different stages of the vaccination process. These stages are represented by varying percentages of vaccinated individuals. This analysis uses a baseline scenario that occurs during the middle stages of second-dose injections, when most settings have stabilized, and about 43\% of the population has received vaccinations. We examine changes in optimal policy by adjusting the percentage of vaccinated individuals to 25\%, 75\%, 80\%, 85\%, 90\%, and 95\%. As shown in Figure \ref{fig: optimal vacc of different initial stages}, in the early stages, maintaining a higher vaccination rate is essential for effective control. However, when the vaccinated proportion reaches approximately 90\%, the optimal administration rates gradually decrease, indicating that a lower ongoing rate suffices to contain the virus. Thus, by analyzing these different stages of time, our framework enables policymakers to develop adequate vaccination strategies tailored to the prevailing conditions at each phase of the rollout.

\section{Conclusion}
\label{sec: conclusion}

Pandemics of emerging infectious diseases pose significant threats to global health and the economy. This paper aims to provide policymakers with a comprehensive framework in economic epidemiology that assists in developing optimal vaccination strategies during pandemics, thereby reducing government expenditures. This framework comprises three phases: modelling, optimizing, and analyzing. By employing this structured approach, governments can leverage the most recent data and continuously update policy responses over time in accordance with the progression of the virus. In the modelling phase, we utilize observed data to train the model and calibrate parameters to accurately reflect the conditions under investigation. This initial phase incorporates both epidemiological and economic data records. We implement the comprehensive stochastic SVEI3RD compartmental model to capture the dynamics of the pandemic and employ PINNs to effectively estimate the parameters of the epidemic model. For the government's pandemic-related expenditures, we construct an overall cost function that encompasses vaccination implementation costs, quarantine subsidies, healthcare expenditures, and economic losses. By breaking down aggregate expenditures into various components, social planners can adjust parameter assumptions in the cost function based on their beliefs and align them with specific local economic conditions. Moving to the optimization stage, we address the government's trade-off during a pandemic, where simultaneously achieving low hospitalization rates and minimizing vaccination campaign expenditures is challenging. Thus, our focus is on identifying the optimal vaccination administration strategy that lowers the government's overall pandemic-related costs over time. Given the complexity of this high-dimensional stochastic control problem, we apply deep neural networks to determine the optimal vaccination strategy within our model framework. In the final phase of our framework, we analyze the impacts of key parameters through a series of sensitivity tests.

Furthermore, to demonstrate the implementation of our framework, we construct a numerical case study using real-world data from Victoria. Our case study examines the optimal vaccination administration rates for Victoria during the COVID-19 pandemic. By comparing our proposed solution to the actual vaccination rollout plan, we find that the Victorian government made significant progress in managing the pandemic during its later phases. However, our framework has the potential to further reduce the burden on the public healthcare system throughout the entire modelling period. The proposed optimal approach effectively suppresses the spread of the disease while reducing aggregate costs to local governments. Additionally, we conduct a sensitivity analysis of our framework for the case study, which provides policymakers with insights into how the driving factors influence final optimal vaccination campaign decisions. Our analysis reveals that factors such as increased noise intensity levels, higher infection rates, reduced vaccination rollout costs, and enhanced governmental economic support can induce social planners to maintain a persistently high vaccination administration rate over an extended period. We also consider the impact of vaccine hesitancy and different time stages. In this case, policymakers can gain a comprehensive understanding of how to adjust rollout plans over time, enabling effective planning as external conditions change. In summary, our study aims to develop guidelines for policymakers by providing valuable strategic insights on managing pandemics and preparing for future outbreaks.

While our framework offers valuable insights, there are several ways in which the modelling can be improved to enhance accuracy. For instance, because we focus on short-term vaccination strategies, our model does not account for the possibility of reinfection after recovery. Incorporating the transition from the Recovered state back to the Susceptible state would enhance the model if the modelling period were extended. Furthermore, regarding the neural network employed to solve the optimal control problem, we utilize a foundational deep neural network to address this stochastic optimal control challenge. Given the rapid advancement in machine learning, recent studies (e.g., \citet{Han_2018} and \citet{Ji_2020}) have built upon this approach by transforming the optimal control problem into a forward-backward stochastic differential equation (FBSDE) problem. Thus, we could consider using alternative deep neural network architectures to solve our problem, potentially enhancing the efficiency and accuracy of our modelling. Additionally, since the results of the neural network depend on the range of the control variable, further research could be conducted to analyze this issue, which would be beneficial for planning future vaccination programs.

As a final remark, the recent outbreak has posed significant challenges for both the public and governments worldwide. While we cannot change what has already occurred, learning from the past remains crucial. To enable policymakers to prepare for future pandemics, it is essential to understand how pandemics spread and to develop appropriate management strategies over time. If this difficult period prompts a re-evaluation of pandemic responses and facilitates the development of optimal policies through analytical and empirical modelling, we glean valuable lessons from past experience. In this spirit, we write this paper to stimulate further research by economists on epidemic-related social planner problems. Given the uncertainty surrounding future pandemics, researchers are urged to develop effective frameworks to assist policymakers in both preparing for and responding to potential outbreaks. 

\newpage
\appendix 
\input{Appendix/existence_sys}	

\input{Appendix/existence_control} 

\newpage
\input{Appendix/PINN_algo} 


\newpage
\bibliography{references}  

\end{document}

%% file: Figure_codes/Han_NN.tex
\usetikzlibrary{positioning, shapes, arrows.meta, calc}

\begin{tikzpicture}[auto, scale=0.6, every node/.style={scale=0.6}]

\tikzstyle{block} = [rectangle, draw, text width=8em, text centered, rounded corners, minimum height=3em]
\tikzstyle{line} = [draw, -latex']

\definecolor{green1}{RGB}{227, 245, 246}
\definecolor{green2}{RGB}{171, 224, 228}
\definecolor{green3}{RGB}{127, 198, 211}
\definecolor{green4}{RGB}{209, 226, 239}
\definecolor{color1}{RGB}{104, 172, 213}
\definecolor{color6}{RGB}{220, 223, 242}
\definecolor{color2}{RGB}{184, 168, 207}
\definecolor{color4}{RGB}{171, 224, 228}
\definecolor{color3}{RGB}{253, 207, 158}
\definecolor{color5}{RGB}{227, 245, 246}

\node[block, fill=color1] (control0) {$\alpha_{t_0}$};
\node[block, fill=color1, right=4em of control0] (control1) {$\alpha_{t_1}$};
\node[rounded corners, align=center, minimum height=3em,  minimum width=8em, right=4em of control1] (control2) {$\cdots$};
\node[block, fill=color1, right=4em of control2] (control3) {$\alpha_{t_{N-1}}$};

\node[block, fill=color2, above=2em of control0] (cost0) {$J(t_0, X_{t_0})$};
\node[block, fill=color2, above=2em of control1] (cost1) {$J(t_1, X_{t_1})$};
\node[rounded corners, align=center, minimum height=3em,  minimum width=8em, above=2em of control2] (cost2) {$\cdots$};
\node[block, fill=color2, above=2em of control3] (cost3) {$J(t_{N-1}, X_{t_{N-1}})$};
\node[block, fill=color2, right=4em of cost3] (cost4) {$J(t_N, X_{t_N})$};

\node[block, fill=color4, below=2em of control0] (top1) (final_layer0) {$h_{t_0}^H$};
\node[block, fill=color4, below=2em of control1] (final_layer1) {$h_{t_1}^H$};
\node[rounded corners, align=center, minimum height=3em,  minimum width=8em, below=2em of control2] (final_layer2) {$\cdots$};
\node[block, fill=color4, below=2em of control3] (final_layer3) {$h_{t_{N-1}}^H$};

\node[rounded corners, align=center, minimum height=3em,  minimum width=8em, below=2em of final_layer0] (dots0){$\cdots$};
\node[rounded corners, align=center, minimum height=3em,  minimum width=8em, below=2em of final_layer1] (dots1) {$\cdots$};
\node[rounded corners, align=center, minimum height=3em,  minimum width=8em, below=2em of final_layer2] (dots2) {$\cdots$};
\node[rounded corners, align=center, minimum height=3em,  minimum width=8em, below=2em of final_layer3] (dots3) {$\cdots$};

\node[block, fill=color4, below=2em of dots0] (first_layer0) {$h_{t_0}^0$};
\node[block, fill=color4, below=2em of dots1] (first_layer1) {$h^0_{t_1}$};
\node[rounded corners, align=center, minimum height=3em,  minimum width=8em, below=2em of dots2] (first_layer2) {$\cdots$};
\node[block, fill=color4, below=2em of dots3] (first_layer3) {$h^0_{t_{N-1}}$};

\node[block, fill=color3, below=2em of first_layer0] (X0) {$X_{t_0}$};
\node[block, fill=color5, below=2em of first_layer1] (X1) {$X_{t_1}$};
\node[rounded corners, align=center, minimum height=3em,  minimum width=8em, below=2em of first_layer2] (X2) {$\cdots$};
\node[block, fill=color5, below=2em of first_layer3] (X3) {$X_{t_{N-1}}$};
\node[block, fill=color5, right=4em of X3] (X4) {$X_{t_{N}}$};

\node[rounded corners, align=center, right=1em of X1] (dot1) {$ $};

\node[block, fill=color3, below=2em of X1] (W1) {$W_{t_1}-W_{t_0}$};
\node[rounded corners, align=center, minimum height=3em,  minimum width=8em, below=2em of X2] (W2) {$\cdots$};
\node[block, fill=color3, below=2em of X3] (W3) {$W_{t_{N-1}} - W_{t_{N-2}}$};
\node[block, fill=color3, below=2em of X4] (W4) {$W_{t_N} - W_{t_{N-1}}$};

\foreach \i in {1,3}
{
    \draw[<-] (cost\i) -- (control\i);
    \draw[<-] (control\i) -- (final_layer\i);
    \draw[<-] (final_layer\i) -- (dots\i);
    \draw[<-] (dots\i) -- (first_layer\i);
    \draw[<-] (first_layer\i) -- (X\i);
}

\draw[<-] (cost0) -- (control0);
\draw[<-] (control0) -- (final_layer0);
\draw[<-] (final_layer0) -- (dots0);
\draw[<-] (dots0) -- (first_layer0);
\draw[<-] (first_layer0) -- (X0);

\draw[->] (cost0) -- (cost1);
\draw[->] (cost1) -- (cost2);
\draw[->] (cost2) -- (cost3);
\draw[->] (cost3) -- (cost4);

\node[circle, inner sep=0.1pt, fill=black, right=1em of control0] (control_mid0) {};
\node[circle, inner sep=0.1pt, fill=black, right=1em of control1] (control_mid1) {};
\node[circle, inner sep=0.1pt, fill=black, right=1em of control3] (control_mid3) {};

\node[circle, inner sep=0.1pt, fill=black, right=1em of X0] (X_mid0) {};
\node[circle, inner sep=0.1pt, fill=black, right=1em of X1] (X_mid1) {};
\node[circle, inner sep=0.1pt, fill=black, right=1em of X2] (X_mid2) {};
\node[circle, inner sep=0.1pt, fill=black, right=1em of X3] (X_mid3) {};

\draw[-] (control0) -- (control_mid0);
\draw[->] (control_mid0) -- (X_mid0);
\draw[-] (control1) -- (control_mid1);
\draw[->] (control_mid1) -- (X_mid1);
\draw[-] (control3) -- (control_mid3);
\draw[->] (control_mid3) -- (X_mid3);

\node[circle, inner sep=0.1pt, fill=black, left=3em of W1] (W_mid0) {};
\node[circle, inner sep=0.1pt, fill=black, left=3em of W3] (W_mid2) {};
\node[circle, inner sep=0.1pt, fill=black, left=3em of W4] (W_mid3) {};

\draw[-] (W1) -- (W_mid0);
\draw[->] (W_mid0) -- (X_mid0);
\draw[-] (W3) -- (W_mid2);
\draw[->] (W_mid2) -- (X_mid2);
\draw[-] (W4) -- (W_mid3);
\draw[->] (W_mid3) -- (X_mid3);

\node[circle, inner sep=0.1pt, fill=black, right=2.5em of cost0] (cost_mid0) {};
\node[circle, inner sep=0.1pt, fill=black, right=2.5em of cost1] (cost_mid1) {};
\node[circle, inner sep=0.1pt, fill=black, right=2.5em of cost3] (cost_mid3) {};

\node[circle, inner sep=0.1pt, fill=black, right=2.5em of X0] (X_midmid0) {};
\node[circle, inner sep=0.1pt, fill=black, right=2.5em of X1] (X_midmid1) {};
\node[circle, inner sep=0.1pt, fill=black, right=2.5em of X2] (X_midmid2) {};
\node[circle, inner sep=0.1pt, fill=black, right=2.5em of X3] (X_midmid3) {};

\draw[->] (X_midmid0) -- (cost_mid0);
\draw[->] (X_midmid1) -- (cost_mid1);
\draw[->] (X_midmid3) -- (cost_mid3);

\draw[->] (X0) -- (X1);
\draw[->] (X1) -- (X2);
\draw[->] (X2) -- (X3);
\draw[->] (X3) -- (X4);

\draw[->] (X4) -- (cost4);

\node[below=0em of W1] (lblk1) {$t = t_1$};
\node[below=0em of W3] (lblkNm1) {$t = t_{N-1}$};
\node[below=0em of W4] (lblkNy) {$t = t_{N}$};
\end{tikzpicture}

%% file: Figure_codes/pie_actual.tex
\definecolor{new_red}{RGB}{184, 168, 207}
\definecolor{new_orange}{RGB}{253, 207, 158}
\definecolor{new_purple}{RGB}{111, 120, 185}

\definecolor{green1}{RGB}{227, 245, 246}
\definecolor{green2}{RGB}{171, 224, 228}
\definecolor{green3}{RGB}{127, 198, 211}
\definecolor{green4}{RGB}{209, 226, 239}

\begin{tikzpicture}[auto, scale=0.6, every node/.style={scale=0.6}]

  \pie[
    radius=3,
    color={new_red, new_orange, green4},
    text=inside,
    explode=0.1,
    font=\scriptsize 
  ]{
    74.94/Economic Cost,
    16.02/Healthcare Cost,
    9.04/Policy Cost
  }

  \node[above, font=\small] at (0, 3.5) {Australian Government COVID-19 Expenditure};

  \begin{scope}[xshift=5cm, yshift=-2cm] 
    \draw[fill=green3] (0,0) rectangle (2,2) node[midway, anchor=center, font=\scriptsize] { 
      2021-22: 
      4.99\%
    };
    \draw[fill=green2] (0,2) rectangle (2,3) node[midway, anchor=center, font=\scriptsize] {
      2020-21: 
      2.82\%
    };
    \draw[fill=green1] (0,3) rectangle (2,3.5) node[midway, anchor=center, font=\scriptsize] {
      2019-20: 
      1.24\%
    };
  \end{scope}

  \draw[dashed, gray] (3.5, 0) -- (4.5, 0.5); 

\end{tikzpicture}

%% file: Figure_codes/pie_model.tex
\usetikzlibrary{patterns}
\definecolor{new_red}{RGB}{184, 168, 207}
\definecolor{new_orange}{RGB}{253, 207, 158}
\definecolor{new_purple}{RGB}{111, 120, 185}

\definecolor{green1}{RGB}{227, 245, 246}
\definecolor{green2}{RGB}{171, 224, 228}
\definecolor{green3}{RGB}{127, 198, 211}
\definecolor{green4}{RGB}{209, 226, 239}
\begin{tikzpicture}[auto, scale=0.6, every node/.style={scale=0.6}]

  \pie[
    radius=3,
    color={new_red, new_orange, green4},
    text=inside,
    explode=0.1,
    font=\scriptsize
  ]{
    73.53/Economic Cost,
    14.83/ Healthcare Cost,
    11.63/Policy Cost
  }

  \node[above, font=\small] at (0, 3.5) {Model Expenditure};

\end{tikzpicture} 

%% file: Figure_codes/base_cost_draw.tex
\usetikzlibrary{patterns, patterns.meta}
\usetikzlibrary{pgfplots.groupplots}
\pgfplotsset{compat=1.17}

\begin{tikzpicture}[auto, scale=0.52]
\pgfdeclarelayer{background}
\pgfdeclarelayer{foreground}
\pgfsetlayers{background,main,foreground}

\definecolor{color1}{RGB}{243, 114, 82}
\definecolor{color2}{RGB}{252, 180, 97}
\definecolor{color4}{RGB}{111, 120, 185}
\definecolor{color3}{RGB}{127, 198, 211}
\begin{axis}[
    ybar stacked,
    ymin=0, ymax=50,
    width=14cm,
    height=11cm,
    ylabel=Cost,
    xtick=data,
    symbolic x coords={No Vaccination, Constant Vaccination, Actual Vaccination, Optimal Vaccination},
    x tick label style={align=center, text width=3cm},
    enlarge x limits=0.2,
    legend style={at={(0.5,-0.1)}, anchor=north, legend columns=-1, 
      /tikz/column sep=2ex,draw=none},
    ymajorgrids = false,
    grid style={dashed, gray!30},
    title={Government Expenditures under Various Vaccination Strategies},
    nodes near coords,
    every node near coord/.style={
        font=\scriptsize,
        black,
        /pgf/number format/.cd,
        fixed,
        precision=2,
    },
    bar width=1cm,
]

\addplot[draw=white, pattern=dots, pattern color=black] coordinates {
    (No Vaccination, 0)
    (Constant Vaccination, 0)
    (Actual Vaccination, 0)
    (Optimal Vaccination, 0)
};

\addplot[draw=white, pattern=grid, pattern color=black] coordinates {
    (No Vaccination, 0)
    (Constant Vaccination, 0)
    (Actual Vaccination, 0)
    (Optimal Vaccination, 0)
};

\addplot[draw=white, pattern=north east lines, pattern color=black] coordinates {
    (No Vaccination, 0)
    (Constant Vaccination, 0)
    (Actual Vaccination, 0)
    (Optimal Vaccination, 0)
};

\addplot[draw=white, fill=color1,
    postaction={
        pattern=dots,
        pattern color = white
    }] coordinates {
    (No Vaccination, 2.11)
    (Constant Vaccination, 0)
    (Actual Vaccination, 0)
    (Optimal Vaccination, 0)
};

\addplot[draw=white, fill=color2,
    postaction={
        pattern=dots,
        pattern color = white
    }] coordinates {
    (No Vaccination, 0)
    (Constant Vaccination, 3.33)
    (Actual Vaccination, 0)
    (Optimal Vaccination, 0)
};

\addplot[draw=white, fill=color4,
    postaction={
        pattern=dots,
        pattern color = white
    }] coordinates {
    (No Vaccination, 0)
    (Constant Vaccination, 0)
    (Actual Vaccination, 3.60)
    (Optimal Vaccination, 0)
};

\addplot[draw=white, fill=color3,
    postaction={
        pattern=dots,
        pattern color = white
    }] coordinates {
    (No Vaccination, 0)
    (Constant Vaccination, 0)
    (Actual Vaccination, 0)
    (Optimal Vaccination, 3.99)
};

\addplot[draw=white, fill=color1,
    postaction={
        pattern=grid,
        pattern color = white
    }] coordinates {
    (No Vaccination, 10.95)
    (Constant Vaccination, 0)
   (Actual Vaccination, 0)
    (Optimal Vaccination, 0)
};

\addplot[draw=white, fill=color2,
    postaction={
        pattern=grid,
        pattern color = white
    }] coordinates {
    (No Vaccination, 0)
    (Constant Vaccination, 8.61)
   (Actual Vaccination, 0)
    (Optimal Vaccination, 0)
};

\addplot[draw=white, fill=color4,
    postaction={
        pattern=grid,
        pattern color = white
    }] coordinates {
(No Vaccination, 0)
    (Constant Vaccination, 0)
    (Actual Vaccination, 5.67)
    (Optimal Vaccination, 0)
};

\addplot[draw=white, fill=color3,
    postaction={
        pattern=grid,
        pattern color = white
    }] coordinates {
(No Vaccination, 0)
    (Constant Vaccination, 0)
    (Actual Vaccination, 0)
    (Optimal Vaccination, 5.08)
};

\addplot[draw=white, fill=color1,
    postaction={
        pattern=north east lines,
        pattern color = white
    }] coordinates {
    (No Vaccination, 34.83)
    (Constant Vaccination, 0)
    (Actual Vaccination, 0)
    (Optimal Vaccination, 0)
};

\addplot[draw=white, fill=color2,
    postaction={
        pattern=north east lines,
        pattern color = white
    }] coordinates {
    (No Vaccination, 0)
    (Constant Vaccination, 27.09)
    (Actual Vaccination, 0)
    (Optimal Vaccination, 0)
};

\addplot[draw=white, fill=color4,
    postaction={
        pattern=north east lines,
        pattern color = white
    }] coordinates {
    (No Vaccination, 0)
    (Constant Vaccination, 0)
    (Actual Vaccination, 27.26)
    (Optimal Vaccination, 0)
};

\addplot[draw=white, fill=color3,
    postaction={
        pattern=north east lines,
        pattern color = white
    }] coordinates {
(No Vaccination, 0)
    (Constant Vaccination, 0)
    (Actual Vaccination, 0)
    (Optimal Vaccination, 25.20)
};

\legend{Policy Cost, Healthcare Cost, Economic Cost}

\end{axis}
\end{tikzpicture}

%% file: Figure_codes/noise_cost.tex
\usetikzlibrary{patterns, patterns.meta}
\usetikzlibrary{pgfplots.groupplots}
\pgfplotsset{compat=1.17}

\begin{tikzpicture}[auto, scale=0.56]
\pgfdeclarelayer{background}
\pgfdeclarelayer{foreground}
\pgfsetlayers{background,main,foreground}

\definecolor{color1}{RGB}{150, 184, 155}
\definecolor{color2}{RGB}{171, 224, 228}
\definecolor{color4}{RGB}{127, 198, 211}
\definecolor{color3}{RGB}{119, 159, 198}

\begin{axis}[
    ybar stacked,
    ymin=0, ymax=40,
    width=14cm,
    height=9.8cm,
    ylabel=Cost,
    xtick=data,
    symbolic x coords={Tiny Noise, Small Noise, Base Case, Large Noise},
    x tick label style={align=center, text width=3cm},
    enlarge x limits=0.2,
    legend style={at={(0.5,-0.1)}, 
    anchor=north, 
    legend columns=4, 
    /tikz/column sep=2ex,draw=none},
    ymajorgrids = false,
    grid style={dashed, gray!30},
    title={Government Expenditures under Various Vaccination Strategies},
    nodes near coords,
    every node near coord/.style={
        font=\scriptsize,
        black,
        /pgf/number format/.cd,
        fixed,
        precision=2,
    },
    bar width=1cm
]

\addplot[draw=white, pattern=dots, pattern color=black] coordinates {
    (Tiny Noise, 0)
    (Small Noise, 0)
    (Base Case, 0)
    (Large Noise, 0)
};

\addplot[draw=white, pattern=grid, pattern color=black] coordinates {
    (Tiny Noise, 0)
    (Small Noise, 0)
    (Base Case, 0)
    (Large Noise, 0)
};

\addplot[draw=white, pattern=north east lines, pattern color=black] coordinates {
    (Tiny Noise, 0)
    (Small Noise, 0)
    (Base Case, 0)
    (Large Noise, 0)
};

\addplot[draw=white, fill=color1,
    postaction={
        pattern=dots,
        pattern color = white
    }] coordinates {
    (Tiny Noise, 3.17)
    (Small Noise, 0)
    (Base Case, 0)
    (Large Noise, 0)
};

\addplot[draw=white, fill=color2,
    postaction={
        pattern=dots,
        pattern color = white
    }] coordinates {
    (Tiny Noise, 0)
    (Small Noise, 3.81)
    (Base Case, 0)
    (Large Noise, 0)
};

\addplot[draw=white, fill=color4,
    postaction={
        pattern=dots,
        pattern color = white
    }] coordinates {
    (Tiny Noise, 0)
    (Small Noise, 0)
    (Base Case, 3.99)
    (Large Noise, 0)
};

\addplot[draw=white, fill=color3,
    postaction={
        pattern=dots,
        pattern color = white
    }] coordinates {
    (Tiny Noise, 0)
    (Small Noise, 0)
    (Base Case, 0)
    (Large Noise, 3.98)
};

\addplot[draw=white, fill=color1,
    postaction={
        pattern=grid,
        pattern color = white
    }] coordinates {
    (Tiny Noise, 4.75)
    (Small Noise, 0)
   (Base Case, 0)
    (Large Noise, 0)
};

\addplot[draw=white, fill=color2,
    postaction={
        pattern=grid,
        pattern color = white
    }] coordinates {
    (Tiny Noise, 0)
    (Small Noise, 4.91)
   (Base Case, 0)
    (Large Noise, 0)
};

\addplot[draw=white, fill=color4,
    postaction={
        pattern=grid,
        pattern color = white
    }] coordinates {
(Tiny Noise, 0)
    (Small Noise, 0)
    (Base Case, 5.08)
    (Large Noise, 0)
};

\addplot[draw=white, fill=color3,
    postaction={
        pattern=grid,
        pattern color = white
    }] coordinates {
(Tiny Noise, 0)
    (Small Noise, 0)
    (Base Case, 0)
    (Large Noise, 5.14)
};

\addplot[draw=white, fill=color1,
    postaction={
        pattern=north east lines,
        pattern color = white
    }] coordinates {
    (Tiny Noise, 23.16)
    (Small Noise, 0)
    (Base Case, 0)
    (Large Noise, 0)
};

\addplot[draw=white, fill=color2,
    postaction={
        pattern=north east lines,
        pattern color = white
    }] coordinates {
    (Tiny Noise, 0)
    (Small Noise, 24.18)
    (Base Case, 0)
    (Large Noise, 0)
};

\addplot[draw=white, fill=color4,
    postaction={
        pattern=north east lines,
        pattern color = white
    }] coordinates {
    (Tiny Noise, 0)
    (Small Noise, 0)
    (Base Case, 25.20)
    (Large Noise, 0)
};

\addplot[draw=white, fill=color3,
    postaction={
        pattern=north east lines,
        pattern color = white
    }] coordinates {
(Tiny Noise, 0)
    (Small Noise, 0)
    (Base Case, 0)
    (Large Noise, 25.60)
};

\legend{Policy Cost, Healthcare Cost, Economic Cost}

\end{axis}
\end{tikzpicture}

%% file: Figure_codes/beta_cost.tex
\usetikzlibrary{patterns, patterns.meta}
\usetikzlibrary{pgfplots.groupplots} \usetikzlibrary{patterns}

\pgfplotsset{compat=1.17}

\begin{tikzpicture}[auto, scale=0.54]
\pgfdeclarelayer{background}
\pgfdeclarelayer{foreground}
\pgfsetlayers{background,main,foreground}

\definecolor{color1}{RGB}{150, 184, 155}
\definecolor{color2}{RGB}{171, 224, 228}
\definecolor{color4}{RGB}{127, 198, 211}
\definecolor{color3}{RGB}{119, 159, 198}

\begin{axis}[
    ybar stacked,
    ymin=0, ymax=60,
    width=14cm,
    height=9.8cm,
    ylabel=Cost,
    xtick=data,
    symbolic x coords={Tiny infection, Mild infection, Base Case, Large Infection},
    x tick label style={align=center, text width=3cm},
    enlarge x limits=0.2,
    legend style={at={(0.5,-0.1)}, anchor=north, legend columns=-1, 
      /tikz/column sep=1ex,draw=none},
    ymajorgrids = false,
    grid style={dashed, gray!30},
    title={Government Expenditures under Various Vaccination Strategies},
    nodes near coords,
    every node near coord/.style={
        font=\scriptsize,
        black,
        /pgf/number format/.cd,
        fixed,
        precision=2,
    },
    bar width=1cm,
]

\addplot[draw=white, pattern=dots, pattern color=black] coordinates {
    (Tiny infection, 0)
    (Mild infection, 0)
    (Base Case, 0)
    (Large Infection, 0)
};

\addplot[draw=white, pattern=grid, pattern color=black] coordinates {
    (Tiny infection, 0)
    (Mild infection, 0)
    (Base Case, 0)
    (Large Infection, 0)
};

\addplot[draw=white, pattern=north east lines, pattern color=black] coordinates {
    (Tiny infection, 0)
    (Mild infection, 0)
    (Base Case, 0)
    (Large Infection, 0)
};

\addplot[draw=white, fill=color1,
    postaction={
        pattern=dots,
        pattern color = white
    }] coordinates {
    (Tiny infection, 2.79)
    (Mild infection, 0)
    (Base Case, 0)
    (Large Infection, 0)
};

\addplot[draw=white, fill=color2,
    postaction={
        pattern=dots,
        pattern color = white
    }] coordinates {
    (Tiny infection, 0)
    (Mild infection, 3.51)
    (Base Case, 0)
    (Large Infection, 0)
};

\addplot[draw=white, fill=color4,
    postaction={
        pattern=dots,
        pattern color = white
    }] coordinates {
    (Tiny infection, 0)
    (Mild infection, 0)
    (Base Case, 5.08)
    (Large Infection, 0)
};

\addplot[draw=white, fill=color3,
    postaction={
        pattern=dots,
        pattern color = white
    }] coordinates {
    (Tiny infection, 0)
    (Mild infection, 0)
    (Base Case, 0)
    (Large Infection, 8.24)
};

\addplot[draw=white, fill=color1,
    postaction={
        pattern=grid,
        pattern color = white
    }] coordinates {
    (Tiny infection, 3.22)
    (Mild infection, 0)
   (Base Case, 0)
    (Large Infection, 0)
};

\addplot[draw=white, fill=color2,
    postaction={
        pattern=grid,
        pattern color = white
    }] coordinates {
    (Tiny infection, 0)
    (Mild infection, 3.46)
   (Base Case, 0)
    (Large Infection, 0)
};

\addplot[draw=white, fill=color4,
    postaction={
        pattern=grid,
        pattern color = white
    }] coordinates {
(Tiny infection, 0)
    (Mild infection, 0)
    (Base Case, 3.99)
    (Large Infection, 0)
};

\addplot[draw=white, fill=color3,
    postaction={
        pattern=grid,
        pattern color = white
    }] coordinates {
(Tiny infection, 0)
    (Mild infection, 0)
    (Base Case, 0)
    (Large Infection, 5.07)
};

\addplot[draw=white, fill=color1,
    postaction={
        pattern=north east lines,
        pattern color = white
    }] coordinates {
    (Tiny infection, 13.27)
    (Mild infection, 0)
    (Base Case, 0)
    (Large Infection, 0)
};

\addplot[draw=white, fill=color2,
    postaction={
        pattern=north east lines,
        pattern color = white
    }] coordinates {
    (Tiny infection, 0)
    (Mild infection, 17.01)
    (Base Case, 0)
    (Large Infection, 0)
};

\addplot[draw=white, fill=color4,
    postaction={
        pattern=north east lines,
        pattern color = white
    }] coordinates {
    (Tiny infection, 0)
    (Mild infection, 0)
    (Base Case, 25.20)
    (Large Infection, 0)
};

\addplot[draw=white, fill=color3,
    postaction={
        pattern=north east lines,
        pattern color = white
    }] coordinates {
(Tiny infection, 0)
    (Mild infection, 0)
    (Base Case, 0)
    (Large Infection, 41.44)
};

\legend{Policy Cost, Health System Cost, Economic Cost}
\end{axis}
\end{tikzpicture}

%% file: Figure_codes/beta_saving.tex
\usetikzlibrary{patterns}

\definecolor{color1}{RGB}{150, 184, 155}
\definecolor{color2}{RGB}{171, 224, 228}
\definecolor{color3}{RGB}{127, 198, 211}
\definecolor{color4}{RGB}{119, 159, 198}

\begin{tikzpicture}[auto, scale=0.54]
\begin{axis}[
    ybar,
    enlarge x limits=0.2,
    legend style={at={(0.5,-0.1)},
      anchor=north,legend columns=4, 
      /tikz/column sep=1ex,draw=none},
    ylabel={},
    width=14cm,
    height=9.8cm,
    symbolic x coords={Policy Savings, Healthcare Savings, Economic Savings},
    xtick=data,
    ymajorgrids = true,
    grid style={dashed, white},
    ymin = -2,
    ymax = 9,
    title={Optimal Policy Savings Compared to constant Vaccination Policies},
    nodes near coords,
    every node near coord/.append style={font=\tiny},
    bar width=16pt
]

\addplot[fill=color1] coordinates {(Policy Savings, -0.76) (Healthcare Savings, 1.62) (Economic Savings, 0.16)};
\addplot[fill=color2] coordinates {(Policy Savings, -0.75) (Healthcare Savings, 2.14) (Economic Savings, 0.44)};
\addplot[fill=color3] coordinates {(Policy Savings, -0.66) (Healthcare Savings, 3.53) (Economic Savings, 1.90)};
\addplot[fill=color4] coordinates {(Policy Savings, -0.33) (Healthcare Savings, 7.15) (Economic Savings, 7.75)};
\legend{Tiny Infection, Mild Infection, Base Case, Large Infection}

\end{axis}

\end{tikzpicture}

%% file: Figure_codes/range_saving.tex
\usetikzlibrary{patterns}

\definecolor{color1}{RGB}{227, 245, 246}
\definecolor{color1}{RGB}{150, 184, 155}
\definecolor{color2}{RGB}{171, 224, 228}
\definecolor{color3}{RGB}{127, 198, 211}
\definecolor{color4}{RGB}{119, 159, 198}

\begin{tikzpicture}[auto, scale=0.54]
\begin{axis}[
    ybar,
    enlarge x limits=0.2,
    y tick label style={/pgf/number format/fixed, /pgf/number format/fixed zerofill, /pgf/number format/precision=2},
    legend style={at={(0.5,-0.15)},
      anchor=north,legend columns=0, 
      /tikz/column sep=1ex,draw=none},
    ylabel={},
    width=14cm,
    height=9.8cm,
    symbolic x coords={Policy Savings, Healthcare Savings, Economic Savings},
    xtick=data,
    ymajorgrids = true,
    grid style={dashed, white},
    ymin = -1,
    ymax = 3,
    title={Optimal Policy Savings Compared to constant Vaccination Policies},
    nodes near coords,
    every node near coord/.append style={font=\tiny, /pgf/number format/fixed, /pgf/number format/fixed zerofill, /pgf/number format/precision=2},
    bar width=16pt
]

\addplot[fill=color1] coordinates {(Policy Savings, 0.48) (Healthcare Savings, 0.08) (Economic Savings, 0.43)};
\addplot[fill=color2] coordinates {(Policy Savings, 0.05) (Healthcare Savings, 0.35) (Economic Savings, 1.31)};
\addplot[fill=color3] coordinates {(Policy Savings, -0.38) (Healthcare Savings, 0.59) (Economic Savings, 2.06)};
\addplot[fill=color4] coordinates {(Policy Savings, -0.82) (Healthcare Savings, 0.80) (Economic Savings, 2.70)};
\legend{Large Hesitancy, Mild Hesitancy, Base Case, Opposite Hesitancy}

\end{axis}

\end{tikzpicture}

%% file: Appendix/existence_sys.tex
\section{Prove of existence to the stochastic system} 
\label{sec: appendix 1 existence of sde system}
\begin{theorem}
For any initial values $(S_0,V_0,E_0,I_{1,0},I_{2 ,0},I_{3,0},R_0,D_0) \in \mathbb{R}^8_+$, there exists a unique positive solution $(S_t,V_t,E_t,I_{1,t},I_{2,t},I_{3,t},R_t,D_t)$ of Equation \ref{eqn: sto SVEI3RD without control} on $t \geq 0$
and the solution will remain in $\mathbb{R}^8_+$ with probability 1.
\end{theorem}

\begin{proof}
As the coefficients of the considered equation are locally defined and satisfy the Lipschitz condition with respect to the available initial data for various compartments, namely 
$(S_0,V_0,E_0,I_{1,0},I_{2,0},I_{3,0}, R_0,D_0) \in \mathbb{R}_+^8$, there exists one local solution for $(S_t,V_t,E_t,I_{1,t},I_{2,t},I_{3,t},R_t,D_t)$ on the time interval $t \in [0,\tau_e)$, 
where $\tau_e$ represent the occurrence time. To establish that the solution is global, we demonstrate that $\tau_e=\infty$ almost surely.
Define a sufficiently large non-negative real constant $k_0$, ensuring that the initial values of all states are contained within the interval $[\frac{1}{k_0}, k_0]$.

Then, we define the stopping time 
\begin{equation*}
\begin{aligned}
\tau_k= \bigg\{ t \in [0,\tau_k) \colon \frac{1}{k}
\geq \min \{ S_t, V_t, E_t, I_{1, t}, I_{2, t}, I_{3, t}, R_t, D_t \}
\\
\text{ or } 
\max \{ S_t,V_t,E_t,I_{1, t},I_{2, t},I_{3, t},R_t,D_t \} \leq k \bigg\} \text{ for } \forall k \geq k_0.
\end{aligned}
\end{equation*}

In this case, we assume that $\inf \phi=\infty$ when $\phi$ represents the void set. The stopping time $\tau_k$ grows monotonically as $k \to \infty$. 
Take $\lim_{k \rightarrow \infty} \tau_k  = \tau_\infty$ 
with $\tau_e \geq \tau_\infty$ almost surely.

If $\forall t \in [0, \tau_k)$, we demonstrate that $\tau_\infty = \infty$ almost surely. Furthermore, we establish that $\tau_e = \infty$ almost surely, with the state variables $(S_t,V_t,E_t,I_{1, t},I_{2, t},I_{3, t},R_t,D_t)\in \mathbb{R}_+^8 $.
Consequently, it follows that $\tau_e=\infty$ almost surely.

If the stated assumption is not valid, then there exist two constants $0<T$ and $\varepsilon \in (0,1)$ for 
$P \left( \tau_\infty \leq T \right) > \varepsilon$.
Further, we define an operator $H: \mathbb{R}^8_+ \rightarrow \mathbb{R}_+$, as 
\begin{equation*}
\begin{aligned}
H(S,V,E,I_1,I_2,I_3,R,D)=(S+V+E+I_1+I_2+I_3+R+D)-8
\\
-(logS+logV+logE+logI_1+logI_2+logI_3+logR+logD) 
\end{aligned}
\end{equation*}

Applying the basic inequality, which states that $\forall x>0$, $-\log x + x-1 \geq 0$, we write $H \geq 0$. Assuming that $k \geq k_0$ and $T>0$, we then use Ito’s result to obtain:
\begin{equation*}
\begin{aligned}
dH(S, V, E, I_1, I_2, I_3, R, D) &=LH(S+V+E+I_1+I_2+I_3+R+D) dt + \sigma_1 (S-1)dW_{1,t}+\sigma_2 (V-1)dW_{2,t}  
\\
&+ \sigma_3 (E-1)dW_{3,t}+\sigma_4 (I_1-1)dW_{4,t}+\sigma_5 (I_2-1)dW_{5,t} + \sigma_6 (I_3-1)dW_{6,t}
\\
&+\sigma_7 (R-1)dW_{7,t}+\sigma_8 (D-1)dW_{8,t}
\end{aligned}
\end{equation*}
where LH: $\mathbb{R}_+^8 \rightarrow \mathbb{R}_+$ is given by the following definition: 
\begin{equation*}
\begin{aligned}
&LH(S+V+E+I_1+I_2+I_3+R+D)=(1-\frac{1}{S})[\Lambda -(\beta_1 I_1+\beta_2 I_2+\beta_3 I_3 )S-\alpha S-\zeta S]
\\
&+(1-\frac{1}{V})[\alpha S-(\beta_1 I_1+\beta_2 I_2+\beta_3 I_3 )\sigma V-\zeta V]
+(1-\frac{1}{E})[(\beta_1 I_1+\beta_2 I_2+\beta_3 I_3 )S
+(\beta_1 I_1+\beta_2 I_2+\beta_3 I_3 )\sigma V-\gamma E-\zeta E]
\\
&+(1-\frac{1}{I_1} )[\gamma E-(\delta_1+p_1 ) I_1-\zeta I_1]
+(1-\frac{1}{I_2} )[p_1 I_1-(\delta_2+p_2 ) I_2-\zeta I_2 ]
+(1-\frac{1}{I_3} )[p_2 I_2-(\delta_3+\mu) I_3 - \zeta I_3 ]
\\
&+(1-\frac{1}{R})(\delta_1 I_1+\delta_2 I_2+\delta_3 I_3-\zeta R)
+(1-\frac{1}{D}) \mu I_3 +\frac{1}{2} (\sigma_1^2+\sigma_2^2+\sigma_3^2+\sigma_4^2+\sigma_5^2+\sigma_6^2+\sigma_7^2+\sigma_8^2 )
\end{aligned}
\end{equation*}

\begin{equation*}
\begin{aligned}
LH&= \Lambda-(\beta_1 I_1 + \beta_2 I_2 + \beta_3 I_3 )S-\alpha S-\zeta S
+ \alpha S-(\beta_1 I_1+\beta_2 I_2+\beta_3 I_3 )\sigma V-\zeta V 
\\
&+(\beta_1 I_1+\beta_2 I_2+\beta_3 I_3 )(S+\sigma V)-\gamma E-\zeta E
\\
&+ \gamma E-(\delta_1+p_1 ) I_1-\zeta I_1
\\
&+p_1 I_1-(\delta_2+p_2 ) I_2-\zeta I_2
\\
&+ p_2 I_2-(\delta_3+\mu) I_3-\zeta I_3
\\
&+\delta_1 I_1+\delta_2 I_2+\delta_3 I_3-\zeta R
\\
&+ \mu I_3+\frac{1}{2}(\sigma_1^2+\sigma_2^2+\sigma_3^2+\sigma_4^2+\sigma_5^2+\sigma_6^2+\sigma_7^2+\sigma_8^2 )
\\
&- [\frac{\Lambda}{S} +(\beta_1 I_1+\beta_2 I_2+\beta_3 I_3 )+\alpha+\zeta]
\\
&- [\frac{\alpha S}{V} +(\beta_1 I_1+\beta_2 I_2+\beta_3 I_3 )\sigma + \zeta]
\\
&-[(\beta_1 I_1+\beta_2 I_2+\beta_3 I_3 ) \frac{S+\sigma V}{E} + \gamma + \zeta]
\\
&- [\gamma \frac{E}{I_1}  +(\delta_1+p_1 )+\zeta]
\\
&-[p_1 \frac{I_1}{I_2} +(\delta_2+p_2 )+\zeta
-p_2 \frac{I_2}{I_3} +(\delta_3+\mu)+\zeta]
\\
&-[(\delta_1 I_1+\delta_2 I_2+\delta_3 I_3 )  \frac{1}{R}+\zeta ]
-\mu \frac{I_3}{D} 
\end{aligned}
\end{equation*}

\begin{equation*}
\begin{aligned}
LH& \leq \frac{1}{2}  (\sigma_1^2+\sigma_2^2+\sigma_3^2+\sigma_4^2+\sigma_5^2+\sigma_6^2+\sigma_7^2+\sigma_8^2)
\\
&+\Lambda+\beta_1+\beta_2+\beta_3+\alpha+\sigma+\gamma+\delta_1+p_1+\delta_2+p_2+\delta_3 +\mu+7\zeta=K
\end{aligned}
\end{equation*}

This implies that the mathematical form of $K$ is non-negative, and it does not depend on either the state variable or the independent variable.
\begin{equation*}
\begin{aligned}
dH(S,V,E,I_1,I_2,I_3,R,D) &\leq Kdt 
+ \sigma_1 (S-1) dW_{1, t}
+ \sigma_2 (V-1) dW_{2, t}
+ \sigma_3 (E-1) dW_{3, t}
\\
&+ \sigma_4 (I_1-1)dW_{4, t}
+ \sigma_5 (I_2-1)dW_{5, t}
+ \sigma_6 (I_3-1)dW_{6, t}
\\
&+ \sigma_7 (R-1) dW_{7, t}
+ \sigma_8 (D-1) dW_{8, t}
\end{aligned}
\end{equation*}

Integrating the equation above, we obtain
\begin{equation*}
\begin{aligned}
E& \bigg[\ H ( S(\tau_k \wedge T),V(\tau_k \wedge T),E(\tau_k \wedge T),I_1 (\tau_k \wedge T),I_2 (\tau_k \wedge T),I_3 (\tau_k \wedge T),R(\tau_k \wedge T),D(\tau_k \wedge T) ) \bigg]\
\\
\leq & H(S_0,V_0,E_0,I_{1,0},I_{2,0},I_{3,0},R_0,D_0)
+E \bigg[\ \int_0^{\tau_k \wedge T}K dt \bigg]\
\\
\leq & H(S_0, V_0, E_0, I_{1,0}, I_{2,0}, I_{3,0},R_0,D_0) + TK
\end{aligned}
\end{equation*}

Set $\Omega_k=\{\tau_k \leq T\}$ for $k \geq k_1$, we obtain $P(\Omega_k) \geq \epsilon$. In this case, for each $\omega$ from $\Omega(\omega)$, there must exist one or more than one $S(\tau_k, \omega ),V(\tau_k, \omega ),E(\tau_k, \omega ),I_1 (\tau_k, \omega ),I_2 (\tau_k, \omega ),I_3 (\tau_k, \omega ),R(\tau_k, \omega ),D(\tau_k, \omega )$ which equals $\frac{1}{k}$ or $k$. 

As a result, $H ( S(\tau_k ),V(\tau_k),E(\tau_k ),I_1 (\tau_k),I_2 (\tau_k),I_3 (\tau_k),R(\tau_k),D(\tau_k) ) $ is no less than $(\frac{1}{k}-1+\log k)$ and $(k-1-\log k)$, and we obtain
$$H(S(\tau_k ),V(\tau_k),E(\tau_k ),I_1 (\tau_k),I_2 (\tau_k),I_3 (\tau_k),R(\tau_k),D(\tau_k))\geq (\frac{1}{k}-1+\log k) \wedge (k-1-\log k) $$

Hence, we get 
\begin{equation*}
\begin{aligned}
&H(S_0,V_0,E_0,I_{1,0},I_{2,0},I_{3,0},R_0,D_0)+TK
\\
&\geq E[\mathbbm{1}_{\Omega(\omega)}  H(S(\tau_k ),V(\tau_k ),E(\tau_k ),I_1 (\tau_k ),I_2 (\tau_k ),I_3 (\tau_k ),R(\tau_k ),D(\tau_k ))]
\\
&\geq \epsilon [(k-\log k-1)\wedge(\frac{1}{k} +\log k-1)]
\end{aligned}
\end{equation*}

where $\mathbbm{1}_{\Omega(\omega)}$  represents the indicating operator of $\omega$. By letting $k \rightarrow \infty$, this leads to the contradiction $\infty > H(S_0,V_0,E_0,I_{1,0},I_{2,0},I_{3,0},R_0,D_0)+TK =\infty$. This contradiction implies that $\tau_\infty = \infty$ almost surely, thereby completing the proof.
\end{proof}

%% file: Appendix/existence_control.tex
\section{The existence of control solution} 
\label{sec: appendix 8 Existence of control Solution}

\label{subsec in setup: existence of sto control}

\begin{theorem}
For any $X \in \mathbb{R}^8_+$, if $J(\alpha^*)$ is finite, then the stochastic optimal control problem admits an optimal control.
\end{theorem}
 
\begin{proof}
To establish the existence of the optimal control, it suffices to verify that conditions [C1]-[C3] are satisfied for Equation \ref{eqn: sto SVEI3RD without control}, as the finiteness of the cost functional $J(\cdot)$ at the optimal control follows directly from the bounded nature of $X_t$ and $\alpha_t$.

\begin{enumerate}[label=\textbf{C\arabic*}]
    \item Since the set for the control variable $\alpha$ is closed and bounded, it follows that the set is a compact metric space under the usual metric.

    \item As Theorem 1 guarantees the existence of positive global solutions for Equation \ref{eqn: sto SVEI3RD without control}, the functions $b$ and $\mathbf{Z}$ are established to be Lipschitz continuous.
    
    \item From the study of \cite{Boyd_2004}, a set $C \subset \mathbb{R}^n$ is convex if the line segment connecting any two points in $C$ lies entirely within $C$. Specifically, for any $x_1, x_2 \in C$ and for any $0 \leq \theta \leq 1$, it holds that $\theta x_1 + (1-\theta) x_2 \in C$.

\end{enumerate}

In this case, we need to prove that, for every $(t, X)$, the set $(b, \mathbf{Z} \mathbf{Z}^T, J)$ is convex in $C \subseteq \mathbb{R}^8$, which ensures the existence of the optimal control variable.
Since $$ 
\mathbf{Z} = \mathbf{Z}^T =
\begin{bmatrix}
\sigma_1 S_t & 0 & 0 & 0 & 0 & 0 & 0 & 0 \\
0 & \sigma_2 V_t & 0 & 0 & 0 & 0 & 0 & 0 \\
0 & 0 & \sigma_3 E_t & 0 & 0 & 0 & 0 & 0\\
0 & 0 & 0 & \sigma_4 I_{1,t} & 0 & 0 & 0 & 0 \\
0 & 0 & 0 & 0 & \sigma_5 I_{2,t} & 0 & 0 & 0 \\
0 & 0 & 0 & 0 & 0 & \sigma_6 I_{3,t} & 0 & 0 \\
0 & 0 & 0 & 0 & 0 & 0 & \sigma_7 R_t & 0 \\
0 & 0 & 0 & 0 & 0 & 0 & 0 & \sigma_8 D_t
\end{bmatrix}
$$
and 
$$ 
\mathbf{Z} \mathbf{Z}^T (t, X_t, \alpha_t) =
\begin{bmatrix}
\sigma_1^2 S_t^2 & 0 & 0 & 0 & 0 & 0 & 0 & 0 \\
0 & \sigma_2^2  V_t^2  & 0 & 0 & 0 & 0 & 0 & 0 \\
0 & 0 & \sigma_3^2  E_t^2  & 0 & 0 & 0 & 0 & 0\\
0 & 0 & 0 & \sigma_4^2  I_{1,t}^2  & 0 & 0 & 0 & 0 \\
0 & 0 & 0 & 0 & \sigma_5^2  I_{2,t}^2  & 0 & 0 & 0 \\
0 & 0 & 0 & 0 & 0 & \sigma_6^2  I_{3,t}^2  & 0 & 0 \\
0 & 0 & 0 & 0 & 0 & 0 & \sigma_7^2  R_t^2  & 0 \\
0 & 0 & 0 & 0 & 0 & 0 & 0 & \sigma_8^2  D_t^2 
\end{bmatrix},
$$
$\mathbf{Z} \mathbf{Z}^T$ is independent of $\alpha$, the set {$\mathbf{Z} \mathbf{Z}^T$ } is convex. 

Then, consider $f=c_1 \alpha^2 +c_2 E_t +c_3 I_{1,t} + c_4 I_{2,t} + c_5 I_{3,t} + c_6 \psi [1-(S_t+V_t+E_t)]$, we let $f_1, f_2$ be two such elements given by 
$f_1 = c_1 \alpha_1^2 +c_2 E_t +c_3 I_{1,t} + c_4 I_{2,t} + c_5 I_{3,t} + c_6 \psi [1-(S_t+V_t+E_t)] $ and $f_2 = c_1 \alpha_2^2 +c_2 E_t +c_3 I_{1,t} + c_4 I_{2,t} + c_5 I_{3,t} + c_6 \psi [1-(S_t+V_t+E_t)] $ where $\alpha_1, \alpha_2 \in C$.

Let $0 \leq \theta \leq 1$, then 
$\theta f_1+ (1-\theta) f_2 = c_2 E_t +c_3 I_{1,t} + c_4 I_{2,t} + c_5 I_{3,t} + c_6 \psi [1-(S_t+V_t+E_t)]+c_1[\theta \alpha_1^2 +(1-\theta)\alpha_2^2]$.

Since $y=\alpha^2$ is convex, $\exists \alpha_3 \in C $ where $\theta \alpha_1^2 +(1-\theta)\alpha_2^2 = \alpha_3^2$.

Therefore, $\theta f_1+ (1-\theta) f_2 = c_2 E_t +c_3 I_{1,t} + c_4 I_{2,t} + c_5 I_{3,t} + c_6 \psi [1-(S_t+V_t+E_t)]+c_1 \alpha_3^2  $. Hence, the set \{$f|\alpha \in C$\} is convex. 

Then, we are left to prove that \{$((b_1,b_2, b_3, b_4, b_5, b_6, b_7, b_8)|u \in U)$\} is convex, where 
\begin{equation}
\left\{
\begin{aligned}
b_1 & =  \Lambda -(\beta_1 I_1 + \beta_2 I_2 + \beta_3 I_3) S - \alpha S -\zeta S ,
\\
b_2 & = \alpha S - (\beta_1 I_1 + \beta_2 I_2 + \beta_3 I_3) \sigma V -\zeta V,
\\
b_3 & =  (\beta_1 I_1 + \beta_2 I_2 + \beta_3 I_3) S + (\beta_1 I_1 + \beta_2 I_2 + \beta_3 I_3) \sigma V - \gamma E -\zeta E ,
\\
b_4 & = \gamma E - (\delta_1 + p_1) I_1 -\zeta I_1 ,
\\
b_5 & = p_1 I_1 - (\delta_2 + p_2) I_2 -\zeta I_2 ,
\\
b_6 & = p_2 I_2 - (\delta_3 + \mu ) I_3 -\zeta I_3,
\\
b_7 & = \delta_1 I_1 + \delta_2 I_2  + \delta_3 I_3 -\zeta R,
\\
b_8 & =  \mu I_3 .
\end{aligned}\right.
\end{equation}

We obtain $b_1+b_2+b_3+b_4+b_5+b_6+b_7+b_8 = \Lambda-\zeta(S+V+E+I_1+I_2+I_3+R)$ and we let $D=\Lambda-\zeta(S+V+E+I_1+I_2+I_3+R)$.

Define $b^p = 
\begin{bmatrix}
b_1^p(t) \\
b_2^p(t) \\
b_3^p(t) \\
b_4^p(t) \\
b_5^p(t) \\
b_6^p(t) \\
b_7^p(t) \\
b_8^p(t) 
\end{bmatrix}$
and $b^q = 
\begin{bmatrix}
b_1^q(t) \\
b_2^q(t) \\
b_3^q(t) \\
b_4^q(t) \\
b_5^q(t) \\
b_6^q(t) \\
b_7^q(t) \\
b_8^q(t) 
\end{bmatrix}$ be two distinct non-zero vectors. 

Since $\theta b^p+(1-\theta) b^q = \theta D + (1-\theta)D=D$, we conclude that the set $(b|\alpha \in C)$ is convex. 

Hence, for every $(t,X)$, the set $(b, \mathbf{Z} \mathbf{Z}^T, f) $ is convex in $C \subset \mathbb{R}^8$. Therefore, from Theorem 5.3 in the textbook of \cite{Yong_1999}, the existence of optimal control is guaranteed. 

\end{proof} 


%% file: Appendix/PINN_algo.tex
\section{Model fitting algorithms} 
\label{sec: appendix 5 model fitting algorithm}

\subsection{Deterministic Model Fitting}
\label{subsec: appendix 5 det model fitting algorithm}

\begin{algorithm}
\caption{Deterministic Data Fitting}
\begin{algorithmic}
\Require {$t$, $S_t$, $V_t$, $E_t$, $I_{1,t}$, $I_{2,t}$, $I_{3,t}$, $R_t$, $D_t$} 
\State Randomly initialize weights $\mathbf{w}$, biases $\mathbf{b}$ and parameters set $\Theta$ 
\For{\text{epoch in epochs}}
    \State Obtain the values of each compartment of the SVEI3RD model with the input as $t$: 
    \State $ \{S^{NN}_{t_i}, V^{NN}_{t_i}, E^{NN}_{t_i}, I^{NN}_{1,t_i}, I^{NN}_{2,t_i}, I^{NN}_{3,t_i}, R^{NN}_{t_i}, D^{NN}_{t_i} \} = {NN}(t_i)  $
    \State Calculate the data loss function denoting the mismatch of the output of the neural network and observation data
    \State ( $i$ representing the number of data points):
    \begin{equation*}
    \begin{aligned}       
    {L}_{{Data}} = \frac{1}{{N}} \sum_{i=1}^{{N}} 
    &\bigg[\ (S_{t_i} - S_{t_i}^{{NN}})^2 + (V_{t_i} - V_{t_i}^{{NN}})^2 
        +(E_{t_i} - E_{t_i}^{{NN}})^2 + (I_{1,t_i} - I_{1, t_i}^{{NN}})^2 \\
        +& (I_{2,t_i} - I_{2, t_i}^{{NN}})^2   + (I_{3,t_i} - I_{3, t_i}^{{NN}})^2 
        + (R_{t_i} - R_{t_i}^{{NN}})^2 + (D_{t_i} - D_{t_i}^{{NN}})^2 \bigg]\
    \end{aligned}
    \end{equation*}
    \State Obtain the auto-differentiation values 
     $\{ \dot{S}^{NN}_{t_i}, \dot{V}^{NN}_{t_i}, \dot{E}^{NN}_{t_i}, \dot{I}^{NN}_{1,t_i}, \dot{I}^{NN}_{2,t_i}, \dot{I}^{NN}_{3,t_i}, \dot{R}^{NN}_{t_i}, \dot{D}^{NN}_{t_i} \} $
    \State Work out the residual loss represents the mean of the sum of squared residual errors from each compartment:
    \begin{equation*}
    \begin{aligned} 
    {L}_{{DE}} = \frac{1}{N} \sum_{i=1}^{{N}} 
   & \bigg\{ [ \Lambda -(\beta_1 I_{1,t_i} + \beta_2 I_{2,t_i} + \beta_3 I_{3,t_i}) S_{t_i} - \alpha S_{t_i} -\zeta S_{t_i} -  \dot{S}_{t_i}^{NN} ] ^2 \\
    +& [\alpha S_{t_i} - (\beta_1 I_{1, t_i} + \beta_2 I_{2, t_i} + \beta_3 I_{3, {t_i}}) \sigma V_{t_i} -\zeta V_{t_i}-  \dot{V}_{t_i}^{NN} ]^2 \\
   + &[ (\beta_1 I_{1,t_i} + \beta_2 I_{2,t_i} + \beta_3 I_{3,t_i}) (S_{t_i} + \sigma V_{t_i}) - \gamma E_{t_i} -\zeta E_{t_i} -  \dot{E}_{t_i}^{NN}]^2 \\
    +& [p_1 I_{1,t_i} - (\delta_2 + p_2) I_{2,t_i} -\zeta I_{2,t_i}-  \dot{I}_{1,t_i}^{NN}]^2 \\
     +& [p_1 I_{1,t_i} - (\delta_2 + p_2) I_{2,t_i} -\zeta I_{2,t_i}-  \dot{I}_{2,t_i}^{NN}]^2 \\
      +& [p_2 I_{2,t_i} - (\delta_3 + \mu ) I_{3,t_i} -\zeta I_{3,t_i} -  \dot{I}_{3,t_i}^{NN}]^2 \\
      +& (\delta_1 I_{1,t_i} + \delta_2 I_{2,t_i}  + \delta_3 I_{3,t_i} -\zeta R_{t_i}-  \dot{R}_{t_i}^{NN})^2 \\
      + &( \mu I_{3,t_i} - \dot{D}_{t_i}^{NN})^2 \bigg\}
    \end{aligned}
    \end{equation*}
    \State Calculate the total loss function as:
    $$
    {L} = \lambda_{{DE}} {L}_{{DE}} +  \lambda_{{Data}} {L}_{{Data}}
    $$
    \State Update $\mathbf{w}$, $\mathbf{b}$ and the dynamic parameters set $\Theta$ by using the Adam optimizer toolkit in Tensorflow.
\EndFor
\end{algorithmic}
\end{algorithm}

\newpage
\subsection{Stochastic Model Fitting}
\label{subsec: appendix 5 sto model fitting algorithm}

\begin{algorithm}
\caption{Stochastic Data Fitting}
\begin{algorithmic}
\Require {$t$, $S_t$, $V_t$, $E_t$, $I_{1,t}$, $I_{2,t}$, $I_{3,t}$, $R_t$, $D_t$}, $ {N}_\text{MC}$ where ${N}_\text{MC}$ is the number of iteration over the SDE.
\State Randomly initialize weights $\mathbf{w}$, biases $\mathbf{b}$ and dynamic parameter sets $\Theta$ and $\mathbf{Z}$.
\For {epoch in epochs}
    \For{$j = 1$ to ${N}_\text{MC}$} 
        \State Obtain the values of each compartment of the SVEI3RD model with the input as $ t $: 
        \State $S^{NN}_{t_i,j}, V^{NN}_{t_i,j}, E^{NN}_{t_i,j}, I^{NN}_{1,t_i,j}, I^{NN}_{2,t_i,j}, I^{NN}_{3,t_i,j}, R^{NN}_{t_i,j}, D^{NN}_{t_i,j}, W_{1,t_i,j}, W_{2,t_i,j}, W_{3,t_i,j}, W_{4,t_i,j}, W_{5,t_i,j}, W_{6,t_i,j},$
        \State $W_{7,t_i,j}, W_{8,t_i,j} = {NN}(t_i,j) $
        \State Set $\{S_{t_0,j}^{{DE}}, V_{t_0,j}^{{DE}}, E_{t_0,j}^{{DE}}, I_{1,t_0,j}^{{DE}}, I_{2,t_0,j}^{{DE}}, I_{3,t_0,j}^{{DE}}, R_{t_0,j}^{{DE}}, D_{t_0,j}^{{DE}}\} = \{S_0, V_0, E_0, I_{1,0}, I_{2,0}, I_{3,0}, R_0, D_0\} $ 
        \For{time step $i = 0$ to ${N}-1$}
            \State Calculate the SDE terms by using Euler-Maruyama discretization:
            \State $S_{t_{i+1,j}}^{{DE}} = S_{t_i,j}^{{NN}} + (\Lambda -(\beta_1 I_{1,t_i,j}^{{NN}} + \beta_2 I_{2, t_i,j}^{{NN}} + \beta_3 I_{3,t_i,j}^{{NN}}) S_{t_i,j}^{{NN}} - \alpha S_{t_i,j}^{{NN}} -\zeta S_{t_i,j}^{{NN}} ) \Delta t + \sigma_1 \sqrt{\Delta t} \Delta W_{1, t_i,j}$
            \State   
                $V_{t_{i+1}}^{{DE}} = V_{t_i,j}^{{NN}} + (\alpha S_{t_i,j}^{{NN}} - (\beta_1 I_{1, t_i,j}^{{NN}} + \beta_2 I_{2, t_i,j}^{{NN}} + \beta_3 I_{3, t_i,j}^{{NN}}) \sigma V_{t_i,j}^{{NN}} -\zeta V_{t_i,j}^{{NN}}) \Delta t + \sigma_2 \sqrt{\Delta t} \Delta W_{2,{t_i},j}$
            \State     
                $E_{t_{i+1},j}^{{DE}} = E_{t_i,j}^{{NN}} + ((\beta_1 I_{1,t_i,j}^{{NN}} + \beta_2 I_{2,t_i,j}^{{NN}} + \beta_3 I_{3,t_i,j}^{{NN}}) (S_{t_i,j}^{{NN}} +  \sigma V_{t_i,j}^{{NN}} )- \gamma E_{t_i,j}^{{NN}} -\zeta     E_{t_i,j}^{{NN}} ) \Delta t$
            \State  
                $+ \sigma_3 \sqrt{\Delta t} \Delta  W_{3,t_i,j} $
            \State
                $I_{1, t_{i+1},j}^{{DE}} = I_{1,t_i,j}^{{NN}} + (\gamma E_{t_i,j}^{{NN}} - (\delta_1 + p_1 I_{1,t_i,j}^{{NN}} -\zeta I_{1,t_i,j}^{{NN}} ) \Delta t + \sigma_4 \sqrt{\Delta t} \Delta W_{4,t_i,j} $
            \State
                $I_{2, t_{i+1,j}}^{{DE}} = I_{2,t_i,j}^{{NN}} + (p_1 I_{1,t_i,j}^{{NN}} - (\delta_2 + p_2) I_{2,t_i,j}^{{NN}} -\zeta I_{2,t_i,j}^{{NN}}) \Delta t + \sigma_5 \sqrt{\Delta t} \Delta W_{5,t_i,j}$
            \State
                $I_{3, t_{i+1},j}^{{DE}} = I_{3,t_i,j}^{{NN}} + (p_2 I_{2,t_i,j}^{{NN}} - (\delta_3 + \mu ) I_{3,t_i,j}^{{NN}} -\zeta I_{3,t_i,j}^{{NN}} ) \Delta t + \sigma_6 \sqrt{\Delta t} \Delta W_{6,t_i,j} $
            \State
               $ R_{t_{i+1},j}^{{DE}} = R_{t_i,j}^{{NN}} + (\delta_1 I_{1,t_i,j}^{{NN}} + \delta_2 I_{2,t_i,j}^{{NN}}  + \delta_3 I_{3,t_i,j}^{{NN}} -\zeta R_{t_i,j}^{{NN}}) \Delta t + \sigma_7 \sqrt{\Delta t} \Delta W_{7,t_i,j}$
            \State
                $D_{t_{i+1},j}^{{DE}} = D_{t_i,j}^{{NN}} +  ( \mu I_{3,t_i,j}^{{NN}} ) \Delta t + \sigma_8 \sqrt{\Delta t} \Delta W_{8,t_i,j}$ 
        \EndFor    
        \State Calculate the physical imposed loss across time:
        \State $ {L}_{{DE, j}} =\frac{1}{N} \sum_{i=0}^{N-1}  
        [(S_{t_{i+1},j}^{{DE}} - S_{t_{i+1},j}^{{NN}})^2 + (V_{t_{i+1},j}^{{DE}}- V_{t_{i+1},j}^{{NN}})^2 
        +(E_{t_{i+1},j}^{{DE}} - E_{t_{i+1},j}^{{NN}})^2 + (I_{1,t_{i+1},j}^{{DE}} - I_{1,t_{i+1},j}^{{NN}})^2  $
        \State $ + (I_{2,t_{i+1},j}^{{DE}} - I_{2,t_{i+1},j}^{{NN}})^2  + (I_{3,t_{i+1},j}^{{DE}} - I_{3,t_{i+1},j}^{{NN}})^2  
        + (R_{t_{i+1},j}^{{DE}}- R_{t_{i+1},j}^{{NN}})^2 + (D_{t_{i+1},j}^{{DE}} - D_{t_{i+1},j}^{{NN}})^2 $
    \EndFor
    \State Calculate the overall physical imposed loss across iterations:
    \State ${L}_{{DE}} = \frac{1}{{N}_{MC}} \sum_{j=1}^{{N}_{MC}} {L}_{{DE, j}}$
    \State Calculate the data loss for each compartment across iterations:
    \State $\Bar{S}_{t_i}^{\text{NN}} = \frac{1}{{N}_{MC}} \sum_{j=1}^{{N}_{MC}} {S}_{t_i,j}$,\
    $\Bar{V}_{t_i}^{\text{NN}} = \frac{1}{{N}_{MC}} \sum_{j=1}^{{N}_{MC}} {V}_{t_i,j}$,\
    $\Bar{E}_{t_i}^{\text{NN}} = \frac{1}{{N}_{MC}} \sum_{j=1}^{{N}_{MC}} {E}_{t_i,j}$,\
    $\Bar{I}_{1,t_i}^{\text{NN}} = \frac{1}{{N}_{MC}} \sum_{j=1}^{{N}_{MC}} {I}_{1,t_i,j}$,\
    \State $\Bar{I}_{2,t_i}^{\text{NN}} = \frac{1}{{N}_{MC}} \sum_{j=1}^{{N}_{MC}} {I}_{2,t_i,j}$,\
    $\Bar{I}_{3,t_i}^{\text{NN}} = \frac{1}{{N}_{MC}} \sum_{j=1}^{{N}_{MC}} {I}_{3,t_i,j}$,\
    $\Bar{R}_{t_i}^{\text{NN}} = \frac{1}{{N}_{MC}} \sum_{j=1}^{{N}_{MC}} {R}_{t_i,j}$,\
    \State $\Bar{D}_{t_i}^{\text{NN}} = \frac{1}{{N}_{MC}} \sum_{j=1}^{{N}_{MC}} {D}_{t_i,j}$,\
    \State Obtain the overall data loss across all time points:
    \State $\text{L}_{\text{Data}} = \frac{1}{N} \sum_{i=1}^{N}  [(S_{t_i} - \Bar{S}_{t_i}^{\text{NN}})^2 + (V_{t_i} - \Bar{V}_{t_i}^{\text{NN}})^2 +(E_{t_i} - \Bar{E}_{t_i}^{\text{NN}})^2 + (I_{1,t_i} - \Bar{I}_{1, t_i}^{\text{NN}})^2 + (I_{2,t_i} - \Bar{I}_{2, t_i}^{\text{NN}})^2$
    \State  $  + (I_{3,t_i} - \Bar{I}_{3, t_i}^{\text{NN}})^2 + (R_{t_i} - \Bar{R}_{t_i}^{\text{NN}})^2 + (D_{t_i} - \Bar{D}_{t_i}^{\text{NN}})^2 ]$
    
    \State Thus, the total loss function is denoted as:
     \State ${L}= \lambda_{{Data}} {L}_{{Data}} + \lambda_{{DE}} {L}_{{DE}}$
    \State Update $\mathbf{w}$, $\mathbf{b}$, $\Theta$ and $\mathbf{Z}$ by using the Adam optimizer in Tensorflow to minimize the total loss function.
\EndFor
\end{algorithmic}
\end{algorithm}